\documentclass[a4paper,11pt]{article}

\usepackage[backref=page]{hyperref}
\usepackage{amsthm}
\usepackage{amsmath}
\usepackage{amssymb}
\usepackage{amsfonts}
\usepackage{amstext}
\usepackage{xspace}
\usepackage{graphicx}
\usepackage{mathrsfs} 
\usepackage{fullpage}
\usepackage{todonotes}

\newtheorem{theorem}{Theorem}
\newtheorem{lemma}{Lemma}[section]
\newtheorem{claim}{Claim}[section]
\newtheorem{corollary}{Corollary}
\newtheorem{definition}{Definition}[section]
\newtheorem{observation}{Observation}[section]
\newtheorem{proposition}{Proposition}[section]

\newcommand{\tw}{{\mathbf{tw}}}

\newcommand{\h}[1]{\end{document}}

\usepackage{boxedminipage}
 
\newcommand{\cS}{{\mathcal{S}}}
\newcommand{\cO}{{\mathcal{O}}}

\newcommand{\defparproblem}[4]{
  \vspace{1mm}
\noindent\fbox{
  \begin{minipage}{0.96\textwidth}
  \begin{tabular*}{\textwidth}{@{\extracolsep{\fill}}lr} #1  & {\bf{Parameter:}} #3 \\ \end{tabular*}
  {\bf{Input:}} #2  \\
  {\bf{Question:}} #4
  \end{minipage}
  }
  \vspace{1mm}
}

\usepackage{bold-extra}
\usepackage{lmodern}
\usepackage[T1]{fontenc}
\usepackage{textcomp}

\newcommand{\I}{\cal{I}}
\newcommand{\mat}{$M=(E,{\cal I})$}
\newcommand{\matl}[1]{$M_{#1}=(E_{#1},{\cal I}_{#1})$}

\newcommand{\whnd}[1]{\widehat{#1}}

\newcommand{\rep}[2] {$\widehat{{\cal #1}} \subseteq_{rep}^{#2} {\cal #1}$}

\newcommand{\minrep}[2] {$\widehat{{\cal #1}} \subseteq_{minrep}^{#2} {\cal #1}$}
\newcommand{\maxrep}[2] {$\widehat{{\cal #1}} \subseteq_{maxrep}^{#2} {\cal #1}$}

\newcommand{\tc}[4] {  $\mathcal{T}_{\mbox{{\sl #1}}}( {#2},{#3},{#4})  $}
\newcommand{\tcwd}[4] {  \mathcal{T}_{\mbox{{\sl #1}}}\left( {#2},{#3},{#4} \right)}

\newcommand{\bnoml}[2]{  $\binom{{#1}}{{#2}}$}
\newcommand{\bnomlwd}[2]{ \binom{{#1}}{{#2}}}
\newcommand{\rank}[1]{$\mbox{\sf rank}(#1)$}
\newcommand{\wf}{${w}:{\cal S} \rightarrow \mathbb{N}$}
\newcommand{\awf}{${w}:{\cal A} \rightarrow \mathbb{N}$}

\newcommand{\tgem}{$\cO\left({p+q \choose p} t p^\omega + t {p+q \choose q} ^{\omega-1} \right)$}

\newcommand{\repmat}[1]{$A_{#1}$}
\newcommand{\In}[1]{\mbox{\sf In}(#1)}
\newcommand{\Out}[1]{\mbox{\sf Out}(#1)}

\title{ Efficient Computation of Representative Sets with Applications in Parameterized and Exact Algorithms\thanks{Preliminary versions of this paper appeared in the proceedings of SODA 2014 and ESA 2014. 
Supported by Rigorous Theory of Preprocessing, ERC Advanced    Investigator Grant 267959 and 
Parameterized Approximation, ERC Starting Grant 306992.}}
\author{
{\large\sc Fedor V. Fomin\thanks{University of Bergen, Norway. \texttt{\{fomin|daniello\}@ii.uib.no}}}\addtocounter{footnote}{-1}
\and {\large \sc Daniel Lokshtanov}\footnotemark 
\and {\large \sc Fahad Panolan}\addtocounter{footnote}{-1}\footnotemark 
\and {\large\sc Saket Saurabh}\thanks{Institute of Mathematical Sciences, India. \texttt{saket@imsc.res.in}}
\addtocounter{footnote}{-2} \footnotemark
}

\date{}
\begin{document}
\maketitle

\thispagestyle{empty}
\begin{abstract}
\footnotesize{
Let \mat{} be a matroid and let  ${\cal S}=\{S_1, \dots, S_t\}$ be a family of subsets of $E$ of size  $p$. A subfamily $\widehat{\cal{S}}\subseteq \cal S$ is {\em $q$-representative} for $\cal S$ if  for every set $Y\subseteq  E$ of size at most $q$, if there is a set $X \in \cal S$ disjoint from $Y$ with $X\cup Y \in \I$, then there is a set $\widehat{X} \in \widehat{\cal S}$ disjoint from $Y$  with $\widehat{X} \cup  Y \in \I$.  By the classical result of Bollob{\'a}s, in a uniform matroid, every family of sets of size $p$ has a $q$-representative family with at most  $\binom{p+q}{p}$ sets. In  his famous ``two families theorem''  from 1977, Lov{\'a}sz proved that the same bound also holds for any   matroid     representable over a field $ \mathbb{F}$.  As observed by Marx, Lov{\'a}sz's proof is constructive.  In this paper we  show how Lov{\'a}sz's proof can be turned into an algorithm constructing a  $q$-representative family of size at most $\binom{p+q}{p}$ in time bounded by a polynomial in $\binom{p+q}{p}$,  $t$, and the time required for field operations. 

We demonstrate how the  efficient construction of representative families can be a powerful tool for designing single-exponential parameterized and exact exponential time algorithms. The applications of our approach include the following.
\begin{itemize}
\item  In the  \textsc{Long Directed Cycle} problem the input is a directed $n$-vertex graph $G$ and the positive integer $k$. The task is to find  a directed cycle of length at least $k$ in $G$, if such a cycle exists.  As a consequence of our $6.75^{k+o(k)} n^{\cO(1)}$ time algorithm, we have that a directed cycle of length at least $\log{n}$, if such cycle exists, can be found in polynomial time.     
As it was shown by Bj{\"o}rklund, Husfeldt,  and Khanna [ICALP 2004], under an appropriate complexity assumption, it is 
impossible to improve this guarantee by more than  a constant factor. Thus our algorithm not only improves over the best previous   $\log{n}/\log\log{n}$ bound of Gabow and Nie [SODA 2004]  but also closes the gap between known lower and upper bounds for this problem.
\item In the  \textsc{Minimum Equivalent Graph} (MEG) problem we are 
 seeking a spanning subdigraph $D'$  of a given $n$-vertex digraph $D$ with as few arcs as possible in which the reachability relation is the same as in the original digraph $D$. The existence of a single-exponential $c^n$-time algorithm for some constant $c>1$ for MEG was  open since the work of Moyles and Thompson [JACM 1969]. 
 \item To demonstrate the diversity of applications of the approach, we provide an alternative proof of the results recently obtained by Bodlaender, Cygan, Kratsch and Nederlof for algorithms on graphs of bounded treewidth, who showed that 
many ``connectivity''   problems such as  \textsc{Hamiltonian Cycle} or \textsc{Steiner Tree} can be solved in time $2^{\cO(t)}n$ on $n$-vertex graphs of treewidth at most $t$.   We believe that expressing graph problems in ``matroid language''  shed light on what makes it possible to solve connectivity problems single-exponential time parameterized by treewidth.      
\end{itemize}
For the special case of uniform matroids on $n$ elements, we give a faster algorithm to compute a representative family.
 We use this algorithm to provide the fastest known  deterministic parameterized algorithms for   \textsc{$k$-Path},   \textsc{$k$-Tree}, and more generally, for  \textsc{$k$-Subgraph Isomorphism}, where the $k$-vertex pattern graph is of constant treewidth. For example, our  \textsc{$k$-Path} algorithm runs in time $\cO(2.619^{k} n \log {n} \log{W})$ on weighted graphs with maximum edge weight $W$.  }




\end{abstract}
\tableofcontents

\section{Introduction}

The theory of matroids  provides   a deep insight into the tractability of many fundamental  problems in Combinatorial Optimizations like \textsc{Minimum Weight Spanning Tree} or \textsc{Perfect Matching}.  Marx in   \cite{Marx09}
was the first to apply matroids to design fixed-parameter tractable algorithms. The main tool used by  Marx was the notion of \emph{representative families}. Representative families for set systems were introduced by Monien in    \cite{Monien85}.

Let \mat{} be a matroid and let  ${\cal S}=\{S_1, \dots, S_t\}$ be a family of subsets of $E$ of size  $p$. A subfamily $\widehat{\cal{S}}\subseteq \cal S$ 
is {\em $q$-representative} for $\cal S$ 
if  for every set $Y\subseteq  E$ of size at most $q$, if there is a set $X \in \cal S$ disjoint from $Y$ with $X\cup Y \in \I$, then there is a set $\widehat{X} \in \widehat{\cal S}$ disjoint from $Y$  and $\widehat{X} \cup  Y \in \I$.  
In  other words, if a set $Y$ of size at most $q$ can be extended  to an independent set of size $|Y|+p$ by adding a subset from $\cal S$, then it also can be extended to an independent set of size $|Y|+p$ by adding a subset from $ \widehat{\cal S}$ as well.

The Two-Families Theorem of  Bollob{\'a}s \cite{Bollobas65} for extremal set systems and its generalization  to subspaces of a vector space of Lov{\'a}sz \cite{Lovasz77} (see also \cite{Frankl82}) imply that every family of sets of size $p$ has a $q$-representative family with at most  
 $\binom{p+q}{p}$ sets. These theorems are 
 the corner-stones in extremal set theory with 
numerous applications in graph and hypergraph theory, combinatorial geometry and theoretical computer science. We refer to  Section 9.2.2 of
\cite{jukna2011extremal}, surveys of Tuza \cite{Tuza94,Tuza96}, and Gil Kalai's blog\footnote{http://gilkalai.wordpress.com/2008/12/25/lovaszs-two-families-theorem/} for more information on the theorems and their applications.
  
  
 For set families, or equivalently for uniform matroids, Monien provided an algorithm computing a $q$-representative family  of size at most $\sum_{i=0}^q p^i$  in time  $\cO(p  q  \cdot \sum_{i=0}^q p^i \cdot t )$    \cite{Monien85}. 
 Marx in \cite{Marx:2006ys} provided another algorithm, also for uniform matroids, for finding $q$-representative families of size at most 
 $\binom{p+q}{p}$ in  time $\cO(p^q \cdot t^2)$.
For linear matroids, 
Marx~\cite{Marx09}  has shown how Lov{\'a}sz's proof can be transformed into an algorithm computing a $q$-representative family.  However, the running time of the algorithm 
given in~\cite{Marx09} is  $f(p,q)(||A_M||t)^{\cO(1)}$,  where $f(p,q)$ is a polynomial 
in $(p+q)^p$ and $\binom{p+q}{p}$,  that is, $f(p,q)=2^{\cO(p \log (p+q))} \cdot {p+q \choose p}^{\cO(1)}$, and $A_M$ is the matroid's representation matrix. 
Thus, when $p$ is a constant,  which is the way this lemma has been  recently used in the kernelization algorithms~\cite{KratschW12}, we have 
that $f(p,q)=(p+q)^{\cO(1)}$. However, for unbounded $p$ (for an example when $p=q=\frac{k}{2}$) the running time of this algorithm is bounded by 
$2^{\cO(k \log k)} (||A_M||t)^{\cO(1)}$.

\medskip
\noindent\textbf{Our results.}
We give two faster algorithms computing  representative families and show how they can be used to obtain improved parameterized and exact exponential algorithms for several fundamental and  well studied problems.

Our first result is the following 

\begin{theorem}
\label{thm:repsetlovasz}
Let \mat{}   be a linear matroid of rank $p+q=k$ given together with its representation matrix  \repmat{M}  over a field $ \mathbb{F}$. 
Let $ \cS = \{S_1,\ldots, S_t\}$ be a family of independent sets of size $p$. Then 
a $q$-representative family $\widehat{\cal{S}}\subseteq \cal S$  for $ \cS$ with at most  \bnoml{p+q}{p} sets can be found in  \tgem \, operations over $ \mathbb{F}$. Here, $\omega<2.373$ is the matrix multiplication exponent.
 \end{theorem}

 
Actually, we will prove a  variant of  Theorem~\ref{thm:repsetlovasz} which allows sets to have weights. This extension will be used in several applications. This theorem uses  the notion of weighted representative families and computes  
a weighted  $q$-representative family of size  at most  \bnoml{p+q}{p} within the running time claimed in  Theorem~\ref{thm:repsetlovasz}.  The proof of Theorem~\ref{thm:repsetlovasz} relies on the exterior algebra based proof of Lov{\'a}sz \cite{Lovasz77}  and exploits the multi-linearity of  the 
 determinant function. 

 For the case of uniform matroids, we provide the following theorem
 
\begin{theorem}
\label{thm:fastRepUniform}
Let $ \cS = \{S_1,\ldots, S_t\}$ be a family of sets of size $p$ over a universe of size $n$ and let $0<x<1$.
For a given $q$, a $q$-representative family $\widehat{\cal{S}}\subseteq \cal S$  for $ \cS$  with at most  
$ {x^{-p}(1-x)^{-q}} \cdot 2^{o(p+q)} $ sets an be computed in time  
$\cO((1-x)^{-q} \cdot 2^{o(p+q)}\cdot t \cdot \log{n})$.
\end{theorem}


As in the case of Theorem~\ref{thm:repsetlovasz}, we prove a more general version of  Theorem~\ref{thm:fastRepUniform}  for weighted sets.
The proof of Theorem~\ref{thm:fastRepUniform} is essentially an algorithmic variant of the ``random permutation'' proof of Bollob\'{a}s Lemma (see~\cite[Theorem 8.7]{jukna2011extremal}). A slightly weaker variant of Bollob\'{a}s Lemma can be proved using random partitions instead of random permutations, the advantage of the random partitions proof being that it can be de-randomized using efficient constructions of {\em universal sets}~\cite{NaorSS95}. To obtain our results we define {\em separating collections} and give efficient constructions of them.

Separating collections can be seen as a variant of universal sets. In its simplest form, an $n$-$p$-$q$-{\em separating collection} ${\cal C}$ is a pair $({\cal F}, \chi)$, where ${\cal F}$ is a family of sets over a universe $U$ of size $n$ and $\chi$ is a function from ${U \choose p}$ to $2^{\cal F}$ such that the following two properties are satisfied;
(a) for every $A \in {U \choose p}$ and every $F \in \chi(A)$, $A \subseteq F$, 
(b) for every $A \in {U \choose p}$ and $B \in {U\setminus A \choose q}$, there is an $F \in \chi(A)$  such that $A \subseteq F$ and $F \cap B = \emptyset$.
The {\em size} of  $({\cal F},\chi)$ is $|{\cal F}|$, whereas the {\em max degree} of $({\cal F},\chi)$ is $\max_{A \in {U \choose p}} |\chi(A)|$.  Here $2^S$ for a set $S$ is the family of all subsets of $S$ while ${S \choose p}$ is the family of all subsets of $S$ of size $p$.

An efficient construction of separating collections is an algorithm that given $n$, $p$ and $q$ outputs the family ${\cal F}$ of a separating collection $({\cal F},\chi)$ and then allows queries $\chi(A)$ for $A \in {U \choose p}$. We give constructions of separating collections of optimal (up to subexponential factors in $p+q$) size and degree, and construction and query time which is linear (up to subexponential factors in $p+q$) in the size of the output. 

In the conference version of the paper~\cite{FominLS14}, we only proved Theorem~\ref{thm:fastRepUniform}  for $x=\frac{p}{p+q}$. That is, let $ \cS = \{S_1,\ldots, S_t\}$ be a family of sets of size $p$ over a universe of size $n$. Then, for a given $q$, a $q$-representative family $\widehat{\cal{S}}\subseteq \cal S$  for $ \cS$  with at most  
${p+q \choose p} \cdot 2^{o(p+q)} \cdot \log n$ sets can be computed in time
 $\cO((\frac{p+q}{q})^q \cdot 2^{o(p+q)}\cdot t \cdot \log{n})$. Later we observed that our proof works for every $0<x<1$ and allows an interesting trade-off between the size of the computed representative families and the time taken to compute them~\cite{FominLPS14}, and that this trade-off can be exploited algorithmically to speed up ``representative families based'' algorithms. Theorem~\ref{thm:fastRepUniform} improves over the one in~\cite{FominLS14} by shaving off a multiplicative factor of $\log n$ from the upper bound on the output family size. 
Independently, at the same time,  Shachnai and Zehavi~\cite{ShachnaiZ14} also observed that our initial proof could be generalized in essentially the same way as what is stated in Theorem~\ref{thm:fastRepUniform}, and that this generalization used to speed up some of the  algorithms given in the preliminary version of the paper~\cite{FominLS14}. In particular they obtain the same dependence on $k$ in the running time bounds as in this paper for $k$-{\sc Path} and {\sc Long Directed Cycle}. 


\medskip\noindent\textbf{Applications.} Here we provide the list of main applications that can be derived from our algorithms that 
compute representative families together with a short overview of previous work on each application.

\begin{table}[htp]
\begin{center}
\begin{tabular}{|c|c|c|}
\hline
Reference & Randomized & Deterministic \\ \hline
Monien  \cite{Monien85} &- & $\cO(k! n m)$ \\
Bodlaender \cite{Bodlaender93a} &- & $\cO(k!2^k n)$ \\
Alon et al.~\cite{AlonYZ}  &   $\cO(5.44^k n)$ &   $\cO(c^k n\log{n})$ for a large $c$\\
  Kneis at al.~\cite{KneisMRR06}  & $\cO^*(4^k)$  & $\cO^*(16^k)$ \\
Chen et al. \cite{ChenKLMR09} &   $\cO(4^{k} k^{2.7} m)$  & $4^{k+\cO(\log^3{k})} nm$ \\ 
Koutis~\cite{Koutis08} & $\cO^*( 2.83^k)$ &  - \\
Williams~\cite{Williams09} & $\cO^*(2^{k})$ & -\\
Bj{\"o}rklund et al. \cite{BjHuKK10} & $\cO^*(1.66^k)$ & - \\
Conference version &-&  $\cO(2.851^{k} n \log^2{n})$  \\
This paper &-&  $\cO(2.619^{k} n \log {n})$  \\
\hline
\end{tabular}
\end{center}
\caption{Results for {\sc $k$-Path}. 
We use $\cO^*()$ notation that hides factors polynomial in  
the number of vertices $n$ and the parameter $k$ in cases when the authors   do not specify the power of polynomials.} 
\label{table:kpath}
\end{table}%

\medskip\noindent
\textbf{$k$-Path.}  In the \textsc{$k$-Path} problem we are given an undirected $n$-vertex graph $G$ and integer $k$. The question is if $G$ contains a path of length $k$. 
 \textsc{$k$-Path}  was  studied intensively within the parameterized complexity paradigm \cite{DowneyF99}. For $n$-vertex graphs the problem is trivially solvable in time $\cO(n^{k})$.  Monien~\cite{Monien85} and Bodlaender showed  that the problem is fixed parameter tractable.   
Monien  used representative families for set systems for his \textsc{$k$-Path} algorithm~\cite{Monien85} and Plehn and Voigt  extended this algorithm to \textsc{Subgraph Isomorphism}  in \cite{Plehn:1991fk}. This led Papadimitriou and Yannakakis \cite{PapadimitriouY96} to conjecture that the problem is solvable in polynomial time for $k=\log{n}$. This conjecture was resolved in
  a seminal paper of Alon et al.~\cite{AlonYZ}, who introduced the method of {color-coding} and  obtained the first single exponential algorithm for the problem.  Actually, the method of  Alon et al. can be applied for more general problems, like finding a $k$-path in directed graphs, or to solve the {\sc Subgraph Isomorphism} problem in  time $2^{\cO(k)}n^{\cO(t)}$, when the treewidth of the pattern graph is bounded by $t$. 
  There has been 
a lot of efforts in parameterized algorithms to reduce the base of the exponent of both deterministic as well as the randomized 
algorithms for the {\sc $k$-Path} problem, see Table~\ref{table:kpath}. 
After the work of Alon et al.~\cite{AlonYZ}, there were several breakthrough ideas leading to faster and faster \emph{randomized} algorithms. Concerning deterministic algorithms, no improvements occurred since 2007, when 
Chen et al.~\cite{ChenLSZ07}  showed a clever way of applying  universal sets to reduce the running time of color-coding algorithm to  
 $\cO^*(4^{k+o(k)})$.

\textsc{$k$-Path}  is a special case of the {\sc $k$-Subgraph Isomorphism} problem, where for given $n$-vertex graph $G$ and $k$-vertex graph 
$F$, the question is whether $G$ contains a subgraph isomorphic to $F$.
In addition to  \textsc{$k$-Path}, parameterized algorithms for two other variants of   {\sc $k$-Subgraph Isomorphism}, when $F$ is a tree, and more generally, a graph of treewidth at most $t$, were studied in the literature.
Alon et al.~\cite{AlonYZ} showed that   {\sc $k$-Subgraph Isomorphism}, when the treewidth of the pattern graph is bounded by $t$,  is solvable 
 in time $2^{\cO(k)}n^{\cO(t)}$.   
Cohen et al. gave a randomized  algorithm that for  an input digraph $D$ decides in time  $5.704^k n^{\cO(1)}$ if $D$ contains  a given out-tree with $k$ vertices
 \cite{CohenFG10}.  They also showed how to derandomize the  algorithm  in time $ 6.14^k n^{\cO(1)}.$ 
Amini et al.~\cite{AminiFS12} introduced an inclusion-exclusion based approach in the 
classical {color-coding} and   gave a randomized $5.4^kn^{\cO(t)}$ time algorithm and a 
deterministic $5.4^{k+o(k)} n^{\cO(t)}$  time algorithm for the case when $F$ has treewidth at most 
$t$.  Koutis and Williams~\cite{KW09} generalized their algebraic approach for {\sc $k$-Path} to {\sc $k$-Tree} 
and obtained a randomized algorithm running in time 
$2^kn^{\cO(1)}$ for {\sc $k$-Tree}. 
A superset of the authors in \cite{FominLRS12}, extended this result by providing 
 a randomized algorithm for {\sc $k$-Subgraph Isomorphism}  running in time 
$2^k(nt)^{\cO(t)}$, when the treewidth of $F$ is at most $t$. However,  the fastest known deterministic algorithm for this problem  
prior to this paper, was the time $5.4^{k+o(k)} n^{\cO(t)}$ algorithm from \cite{AminiFS12}. In this paper we give deterministic algorithms for 
{\sc $k$-Path} and {\sc $k$-Tree} that run in time $\cO(2.619^{k} n \log n)$ and $\cO(2.619^{k} n^{\cO{(1)}} )$. The algorithm for {\sc $k$-Tree}  
can be generalized  to {\sc $k$-Subgraph Isomorphism} for the case when the pattern graph $F$ has treewidth at most $t$. This  algorithm will run in time $\cO(2.619^{k} n^{\cO{(t)}})$. Our approach can also be applied to find directed paths and cycles of length $k$ in time 
$\cO(2.619^{k} m \log n  )$ and $\cO(2.619^{k} n^{\cO{(1)}} )$ respectively.

 Another interesting feature of our approach is that due to using weighted representative families, we can handle the weighted version of the problem as well.   The  weighted version of \textsc{$k$-Path} is known as \textsc{ Short Cheap Tour}.  Let $G$ be a graph with maximum edge cost  $W$, then the problem is to find a path of length at least $k$ where the total sum of costs on the edges is minimized. The algorithm of Bj{\"o}rklund et al. \cite{BjHuKK10} can be adapted to solve \textsc{ Short Cheap Tour} in time $\cO( 1.66^{k}  n^{\cO(1)} {W} )$, however, their approach does not seem to be applicable to obtain  
  algorithms with polylogarithmic  dependence  on $W$.
 Williams in \cite{Williams09} observed that  a divide-and-color approach from \cite{ChenKLMR09} can be used to solve   \textsc{ Short Cheap Tour}    
in time $\cO( 4^{k}  n^{\cO(1)} \log{W} )$. No better algorithm for \textsc{ Short Cheap Tour} was known prior to our work.  As it was noted by  Williams, the $\cO( 2^{k} n^{\cO(1)})$  algorithm of his paper does not appear to extend to weighted graphs. Our approach provides deterministic  $\cO(2.619^{k} n^{\cO{(1)}} \log{W})$  time algorithm for 
\textsc{ Short Cheap Tour} and partially resolves an open question asked  by Williams. 


\medskip\noindent
\textbf{Long  Directed Cycle.} In the {\sc Long Directed Cycle} problem we are interested in finding a cycle of length at least $k$ in a directed graph. For this problem we give an algorithm of running time  $\cO(6.75^{k+o(k)} mn^2\log n).$

While at the first glance the problem is similar to the problem of finding a cycle or a path of length exactly $k$, it is more tricky. The reason is that the problem of finding a cycle of length $\geq  k$ may entail finding a much longer, potentially even a Hamiltonian cycle. This is why color-coding, and other techniques applicable to \textsc{$k$-Path}  do not seem to work here. Even for undirected graphs color-coding alone is not sufficient, and one needs an additional clever trick to make it work.
The first fixed-parameter tractable algorithm for {\sc Long Directed Cycle}  is due to Gabow and Nie \cite{GabowN08}, who gave algorithms with expected  running time $k^{2k}2^{\cO(k)}nm$ and worst-case times $\cO(k^{2k}2^{\cO(k)}nm\log{n})$ or $\cO(k^{3k}nm)$. These running times allow them to   
 find a directed cycle of length at least $\log{n}/\log{\log{n}}$ in expected polynomial time, if it exists. Let us note, that our algorithm implies that 
 one can find in polynomial time a directed cycle of length at least  $\log{n}$ if there is such a cycle.
 On the other hand, 
Bj{\"o}rklund et al. \cite{BjorklundHK04}  have shown that 
assuming Exponential Time Hypothesis (ETH) of Impagliazzo et al. \cite{ImpagliazzoPZ01},   there is no polynomial time algorithm that finds a  directed cycle of length $\Omega(f(n)\log{n})$, for any nondecreasing, unbounded, polynomial time computable function $f$ that tends to infinity. Thus, our work closes the gap between the upper and lower bounds for this problem. 


\medskip\noindent
\textbf{Minimum Equivalent Graph.} Our next application is from exact exponential time algorithms, we refer to  \cite{Fomin:2011rr} for an introduction to the area of exact algorithms. 
In the \textsc{Minimum Equivalent Graph} (MEG) problem we are 
 seeking a spanning subdigraph $D'$  of a given  digraph $D$ with as few arcs as possible in which the reachability relation is the same as in the original digraph $D$. In other words, for every pair of vertices $u,v$, there is a path from $u$ to $v$ in $D'$ if and   only if the original digraph $D$ has such a path.
 We show that this problem is solvable in time $\cO(2^{4\omega n}m n)$, where $n$ is the number of vertices and $m$ is the number of arcs in $D$.

MEG is a classical NP-hard problem generalizing the  \textsc{Hamiltonian Cycle} problem, see Chapter~12 of the book \cite{BangG089_book} for an  overview of combinatorial and algorithmic results on MEG. The algorithmic studies of MEG can be traced to the  work of Moyles and Thompson  \cite{MoylesT69} from 1969, who gave a (non-trivial)  branching algorithm solving MEG in time  $\cO(n!)$. In 1975, Hsu in  \cite{Hsu75} discovered a mistake in the algorithm of Moyles and Thompson, and designed a different branching algorithm for this problem. Martello \cite{Martello78} and Martello and Toth \cite{MartelloT82}  gave another branching based algorithm with running time $\cO(2^{m})$. No single-exponential exact algorithm, i.e. of running time $2^{\cO(n)}$, for MEG  was known prior to our work. 

As it was already observed by Moyles and Thompson \cite{MoylesT69} the hardest instances of MEG are strong digraphs.
A digraph is strong if for every pair of  vertices $u\neq v$, there are   directed paths from $u$ to $v$ and from $v$ to $u$.  MEG restricted to strong digraphs is known as the 
\textsc{Minimum SCSS} (strongly connected spanning subgraph) problem. It is known that the MEG problem reduces in linear time to  \textsc{Minimum SCSS}, see e.g. \cite{CormenLRS01}.


\medskip\noindent
\textbf{Treewidth algorithms.}  We show that efficient computation of representative families can be used to obtain 
algorithms solving  ``connectivity'' problems like  \textsc{Hamiltonian Cycle} or \textsc{Steiner Tree} in  time $2^{\cO(t)}n$, where $t$ is the treewidth of the input $n$-vertex graph.
 It is well known that many intractable problems can be solved efficiently when   the input graph has bounded treewidth. Moreover, many fundamental problems like 
\textsc{Maximum Independent Set} or \textsc{Minimum Dominating Set} can be solved in  time $2^{\cO(t)}n$.
  On the other hand, it was believed until very recently  that for some ``connectivity''  problems such as  \textsc{Hamiltonian Cycle} or \textsc{Steiner Tree}  no such algorithm exists. 
In their  breakthrough paper, Cygan et al.  
\cite{Cygan11} introduced a new algorithmic framework called Cut\&Count and used it to obtain  $2^{\cO(t)}n^{\cO(1)}$   time  Monte Carlo algorithms for a number of   connectivity problems.  
  Very recently, Bodlaender et al.  \cite{BodlaenderCK12} obtained the first  deterministic single exponential algorithms for these problems.  Bodlaender et al. presented two approaches, one  based on     rank estimations in specific matrices and the second based on matrix-tree theorem and computation of determinants.
 Our approach, based on representative families in  matroids, can be seen as an alternate path to obtain  similar results. The main idea behind our approach is that all the relevant information about ``partial solutions'' in bags of  the tree decomposition, can be encoded as an independent set of a specific matroid. Here efficient computation of representative families comes into play.

 \medskip
 In all our applications we first define a specific matroid and then show a combinatorial relation between solution to the problem and independent sets of the matroid. Then we compute representative families using Theorem~\ref{thm:repsetlovasz} or Theorem~\ref{thm:fastRepUniform} and use them to obtain a 
 solution to the problem. 
 We believe that expressing graph problems in ``matroid language''  is a generic technique explaining  why certain 
   problems admit  single-exponential parameterized and exact exponential  algorithms.  Finally, for completeness we would like to add that in the conference version of the paper, the running time for  {\sc $k$-Path} and  {\sc $k$-Tree} were   $\cO(2.815^{k} n^{\cO{(1)}} )$; for {\sc $k$-Subgraph Isomorphism} for the case when the pattern graph $F$ has treewidth at most $t$ was 
   $\cO(2.815^{k} n^{\cO{(t)}})$ and for {\sc Long Directed Cycle} was $8^{k+o(k)} n^{\cO{(1)}} $.

 
\medskip\noindent
\textbf{Organization of the paper.} In Section~\ref{sec:prelim} we give the necessary definitions and state some of the known results that we will use. In Section~\ref{sec:linmat} we prove Theorem~\ref{thm:repsetlovasz} by giving an efficient algorithm for the computation of representative families for linear matroids. In Section~\ref{sec:unimat} we prove Theorem~\ref{thm:fastRepUniform} by giving an efficient algorithm for the computation of representative families for uniform matroids. In Section~\ref{section:application} we give all our applications of Theorems~\ref{thm:repsetlovasz} and~\ref{thm:fastRepUniform}. Concluding remarks and new developments can be found in Section~\ref{sec:conclusion}. {The proofs of Theorem~\ref{thm:repsetlovasz} and Theorem~\ref{thm:fastRepUniform} are independent of each other and may be read independently. All of our applications use Theorems~\ref{thm:repsetlovasz} and~\ref{thm:fastRepUniform} as black boxes, and thus may be read independently of the sections describing the efficient computation of representative families.}


\section{Preliminaries}\label{sec:prelim}

In this section we give various definitions which we make use of in the paper. 

\medskip

\noindent 
{\bf Graphs.} Let~$G$ be a graph with vertex set $V(G)$ and edge set $E(G)$. A graph~$G'$ is a  \emph{subgraph} of~$G$ if~$V(G') \subseteq V(G)$ and~$E(G') \subseteq E(G)$. 
The subgraph~$G'$ is called an \emph{induced subgraph} of~$G$ if~$E(G')
 = \{ uv \in E(G) \mid u,v \in V(G')\}$, in this case, $G'$~is also called the subgraph \emph{induced by~$V(G')$} and denoted by~$G[V(G')]$. For a vertex set $S$, by $G \setminus S$ we denote $G[V(G) \setminus S]$.  By $N(u)$ we denote (open) neighborhood of $u$, that is, the set of all vertices adjacent to $u$. Similarly, by $N[u]=N(u) \cup \{u\}$ we define the closed neighborhood.  The degree of a vertex $v$ in $G$ is $|N_G(v)|$ and is denoted by $d(v)$.     For a subset $S \subseteq V(G)$, we define $N[S]=\cup_{v\in S} N[v]$ and $N(S) = N[S] \setminus S$.  
 By the length of the path we mean the number of edges in it. 
 
 \medskip 
 
 \noindent
 {\bf Digraphs.} Let $D$ be a digraph. By $V(D)$ and $A(D)$ we represent the vertex set and arc set of $D$, respectively. Given a
subset $V'\subseteq V(D)$ of a digraph $D$, let $D[V']$ denote the
digraph induced by $V'$.  A digraph $D$ is
{\em strong} if for every pair $x,y$ of vertices there
are directed paths from $x$ to $y$ and from $y$ to $x.$ A maximal
strongly connected subdigraph of $D$ is called a {\em strong
component}. A vertex $u$ of $D$ is an {\em in-neighbor} ({\em
out-neighbor}) of a vertex $v$ if $uv\in A(D)$ ($vu\in A(D)$,
respectively). The {\em in-degree} $d^-(v)$ ({\em out-degree}
$d^+(v)$) of a vertex $v$ is the number of its in-neighbors
(out-neighbors). We denote the set of in-neighbors and out-neighbors of a vertex $v$  by $N^{-}(v)$  and $N^{+}(v)$ correspondingly. 
 A {\em closed directed walk} in a digraph $D$ is a sequence $v_0v_1\cdots v_\ell$   of vertices of $D$, not necessarily distinct, 
such that 
$v_0=v_\ell$ and for every $0\leq i\leq \ell-1$,   $v_iv_{i+1}\in A(D)$.


\medskip
\noindent
{\bf Sets, Functions and Constants.}  We use the following notations: 
 $[n]=\{1,\ldots,n\}$ and ${[n] \choose i}=\{X~|~X\subseteq [n],~|X|=i\}$. 

 We use the following operations on families of sets. 
  \begin{definition}
  Given two families of sets ${\cal A}$ and ${\cal B}$, we define 
\begin{itemize}
\item[$(\bullet)$] 
${\cal A}  \bullet {\cal B}=\{X \cup Y~|~X \in {\cal A}  \mbox{ and } Y \in {\cal B}  \mbox{ and } X \cap Y = \emptyset\}.$ 
Let ${\cal A}_1, \ldots ,{\cal A}_r$ be $r$ families. Then 
\[\prod^{\bullet}_{i\in [r]} {\cal A}_i = {\cal A}_1\bullet \cdots \bullet {\cal A}_r.\]

\item[$( \circ)$]   
${\cal A} \circ {\cal B} = \{A \cup B ~:~A \in {\cal A}  \mbox{ and } B \in {\cal B}\}.$ 

   \item[$( +)$]  For a    set $X$, we define  
 ${\cal A} + X = \{A \cup X ~:~A \in {\cal A}\}.$ 
\end{itemize}
\end{definition}



The first and second derivatives of a function $f(x)$ of a variable $x$ is denoted by $f'(x)$ and $f''(x)$ respectively.  Throughout the paper we use $\omega$ to denote the exponent in the running time of matrix multiplication, the current best known bound for which is $\omega<2.373$~\cite{Williams12}. We use $e$ to denote the base of natural logarithm. 


\subsection{Randomized Algorithms}
We follow the same notion of randomized algorithms as described in~\cite[Section~2.3]{Marx09}. That is, 
some of the algorithms presented in this paper are randomized, which means that they can produce incorrect answer, but the probability of doing 
so is small. We assume that the algorithm has an integer parameter $P$ given in unary, and the probability of incorrect answer is 
$2^{-P}$. 
 

%
%

\subsection{Matroids}
In the next few subsections we give definitions related to matroids. For a broader overview on matroids we refer to~\cite{oxley2006matroid}. 
\begin{definition}
A pair \mat, where $E$ is a ground set and $\cal I$ is a family of subsets (called independent sets) of $E$, is a {\em matroid} if it satisfies the following conditions:
 \begin{enumerate}
 \item[\rm (I1)]  $\phi \in \cal I$. 
 \item[\rm (I2)]  If $A' \subseteq A $ and $A\in \cal I$ then $A' \in  \cal I$. 
 \item[\rm (I3)] If $A, B  \in \cal I$  and $ |A| < |B| $, then there is $ e \in  (B \setminus A) $  such that $A\cup\{e\} \in \cal I$.
 \end{enumerate}
\end{definition}
The axiom (I2) is also called the hereditary property and a pair $(E,\cal I)$  satisfying  only (I2) is called hereditary family.  
An inclusion wise maximal set of $\cal I$ is called a {\em basis} of the matroid. Using axiom (I3) it is easy to show that all the bases of a matroid  
have the same size. This size is called the {\em rank} of the matroid $M$, and is denoted by \rank{M}. 
 \subsection{Linear Matroids and Representable Matroids} 
 Let $A$ be a matrix over an arbitrary field $\mathbb F$ and let $E$ be the set of columns of $A$. For $A$, we define  matroid 
 \mat{} as follows. A set $X \subseteq E$ is independent (that is $X\in \cal I$) if the corresponding columns are linearly independent over $\mathbb F$.  
%
%
%
The matroids that can be defined by such a construction are called {\em linear matroids}, and if a matroid can be defined by a matrix $A$ over a 
field $\mathbb F$, then we say that the matroid is representable over $\mathbb F$. That is, a matroid \mat{} of rank $d$ is representable over a field 
$\mathbb F$ if there exist vectors in $\mathbb{F}^d$  corresponding to the elements such that  linearly independent sets of vectors 
 correspond to independent sets of the matroid. 
   A matroid \mat{}  is called {\em representable} or {\em linear} if it is representable over some field $\mathbb F$. 

\subsection{Direct Sum of Matroids.}  Let \matl{1}, \matl{2}, \dots, \matl{t} be $t$ matroids with $E_i\cap  E_j =\emptyset$ for all $1\leq i\neq j \leq t$. The direct sum $M_1\oplus \cdots \oplus M_t$  is a matroid \mat{} with  $E := \bigcup_{i=1}^t E_i$ and $X\subseteq E$ is independent if and only if  $X\cap E_i\in {\cal I}_i$ for all $i\leq t$.  Let $A_i$  be the representation matrix of \matl{i}. Then,  
    \begin{displaymath}
    A_M=\left(
     \begin{array}{ccccc} 
        A_1 & 0 & 0 & \cdots & 0\\
        0 &  A_2 &  0 & \cdots & 0 \\
        \vdots &   \vdots  &  \vdots   &  \vdots  &  \vdots \\
        0 & 0 & 0 & \cdots & A_t 
     \end{array} \right)
     \end{displaymath}
     is a representation matrix of $M_1\oplus \cdots \oplus M_t$.  The correctness of this construction is proved in~\cite{Marx09}.
\begin{proposition}[{\cite[Proposition 3.4]{Marx09}}]
\label{prop:disjointsumrep}
Given representations	of matroids $M_1, \ldots, M_t$ over	the same field $\mathbb{F}$, a representation of their direct sum 
can be found in polynomial time.
\end{proposition}

\subsection{Uniform and Partition Matroids}  
A pair  \mat{} over an $n$-element ground set $E$, is called a uniform matroid if the family of independent sets  is given by ${\cal I}=\{A\subseteq E~|~|A|\leq k\}$, 
where $k$ is some constant.   This matroid is also denoted as $U_{n,k}$. Every uniform matroid is linear and can be represented 
over a finite field by a $k \times n$ matrix \repmat{M} where the $A_M[i,j]=j^{i-1}$.  
  \begin{displaymath}
    A_M=\left(
     \begin{array}{ccccc} 
        1 & 1 & 1 & \cdots & 1\\
        1 &  2 &  3 & \cdots & n \\
          1 &  2^2 &  3^2 & \cdots & n^2 \\
        \vdots &   \vdots  &  \vdots   &  \vdots  &  \vdots \\
        1 & 2^{k-1} & 3^{k-1} & \cdots & n^{k-1}
     \end{array} \right)
     \end{displaymath}

Matrix $A_M$ is called Vandermonde matrix. 
Observe that for  $U_{n,k}$  to be representable over a finite field $\mathbb F$, we need that the determinant of  each $k\times k$ submatrix of $A_M$  must not vanish over  $\mathbb F$. Observe that any $k$ columns corresponding to $x_{i_1},\ldots,x_{i_k}$ itself form a Vandermonde matrix, whose determinant is given by 
\[\prod_{1\leq j <\ell \leq k} (x_{i_j}-x_{i_\ell}).\]
Combining this with the fact that $x_1,\ldots,x_n$ are $n$ distinct elements of $\mathbb F$,  we conclude that every subset of size at most $k$ of the ground set is independent, while clearly each larger subset is  dependent. Thus, choosing a field $\mathbb F$ of size larger than $n$ suffices. 
 Note that this means that a representation of the uniform matroid $U_{n,k}$ can be stored using $\cO(\log n)$ bits.

%

 A partition matroid  \mat{} is defined by a ground set $E$ being partitioned into (disjoint) sets $E_1,\ldots,E_\ell$ and by $\ell$ non-negative integers $k_1,\ldots,k_\ell$.  A set $X\subseteq E$ is independent if and only if  $|X\cap E_i| \leq k_i$ for all $i\in \{1,\ldots, \ell\}$. Observe that a partition matroid is a direct sum of uniform matroids $U_{|E_1|,k_1}, \cdots, U_{|E_\ell|,k_\ell}$.  Thus, by Proposition~\ref{prop:disjointsumrep} and the fact that 
 a uniform matroid $U_{n,k}$ is representable over a field $\mathbb F$ of size larger than $n$, we have that.
  
\begin{proposition}[{\cite[Proposition 3.5]{Marx09}}]
\label{prop:uniformandpartitionrep}
A representation over a field of size $\cO( |E|)$ of a partition matroid can be constructed in polynomial time.
\end{proposition}
 
 
 \subsection{Graphic Matroids}  
 Given a graph $G$, a graphic matroid \mat{} is defined by taking elements as edge of $G$ (that is $E=E(G)$) and $F\subseteq E(G)$ is in $\cal I$ if it forms a spanning forest in the  graph $G$.  The graphic matroid is representable over any field of size at least $2$. 
 Consider the matrix $A_M$  with a row for each vertex $i \in V(G)$ and a column for each edge $e = ij\in E(G)$. In the column corresponding to 
 $e=ij$, all entries are $0$, except for a $1$ in $i$ or $j$ (arbitrarily) and a $-1$ in the other. This is a representation over reals. To obtain a representation over a field $\mathbb F$, one simply needs to take the representation given above over reals and simply replace all $-1$ 
 by the additive inverse of $1$ 

 \begin{proposition}[\cite{oxley2006matroid}]
\label{prop:graphicrep}
Graphic matroids are representable over any field of size at least $2$.
\end{proposition}
 


\subsection{Truncation of a Matroid.} The {\em $t$-truncation} of a matroid \mat{}  is a matroid $M'=(E,{\cal I}')$ such that $S\subseteq  E$ is independent in $M'$  if and only if $|S| \leq t$ and $S$ is independent in $M$ (that is $S\in \cal I$).  

\begin{proposition}[{\cite[Proposition 3.7]{Marx09}}]
\label{prop:truncationrep}
Given a matroid $M$ with a representation $A$ over a finite field $\mathbb F$ and an integer $t$,  a representation of the $t$-truncation $M'$ 
can be found in randomized polynomial time.
\end{proposition}

\section{Fast Computation for Representative Sets  for Linear Matroids}\label{sec:linmat}


In this section we give an algorithm to find a $q$-representative family of a given family.  
We start with the definition of a {\em $q$-representative family}.  
\begin{definition}[{\bf $q$-Representative Family}]
Given a matroid  \mat{} and a family $\cal S$ of subsets of $E$, we say that a subfamily $\widehat{\cal{S}}\subseteq \cal S$ 
is {\em $q$-representative} for $\cal S$ 
if the following holds: for every set $Y\subseteq  E$ of size at most $q$, if there is a set $X \in \cal S$ disjoint from $Y$ with $X\cup Y \in \I$, then there is a set $\whnd{X} \in \whnd{\cal S}$ disjoint from $Y$ with $\whnd{X} \cup  Y \in \I$.  If $\hat{\cal S} \subseteq {\cal S}$ is $q$-representative for ${\cal S}$ we write \rep{S}{q}. 
\end{definition}

In other words if some independent set in $\cal S$ can be extended to a larger independent set by $q$ new elements, then there is a set in 
$\widehat{\cal S}$ that can be extended by the same $q$ elements.  A weighted variant of $q$-representative families  is defined as follows. It 
is useful for solving problems where we are looking for objects of maximum or minimum weight.
\begin{definition}[{\bf Min/Max $q$-Representative Family}]
Given a matroid  \mat{}, a family $\cal S$ of subsets of $E$ and a non-negative weight function \wf, we say that a subfamily $\widehat{\cal{S}}\subseteq \cal S$ 
is {\em min $q$-representative}  ({\em max $q$-representative}) for $\cal S$ 
if the following holds: for every set $Y\subseteq  E$ of size at most $q$, if there is a set $X \in \cal S$ disjoint from $Y$ with $X\cup Y \in \I$, then 
there is a set $\whnd{X} \in \whnd{\cal S}$ disjoint from $Y$ with 
\begin{enumerate}
\item  $\whnd{X} \cup  Y \in \I$; and 
\item $w(\whnd{X} )\leq w(X)$ ($w(\whnd{X} )\geq w(X)$).  
\end{enumerate}
We use  \minrep{S}{q} (\maxrep{S}{q}) to denote a  min $q$-representative (max $q$-representative) family for $\cal S$. 
\end{definition}

We say that a family  $ \cS = \{S_1,\ldots, S_t\}$ of 
sets is a {\em $p$-family} if each set in $\cal S $ is of size $p$. 

We start by three lemmata providing basic results about representative sets. These lemmata will be used in Section~\ref{section:application}, where we provide algorithmic applications of representative families.  We prove them for unweighted representative families but they can be easily   modified to work for weighted variant. 
\begin{lemma}
\label{lem:reptransitive}
Let \mat{} be a matroid and $\cal S$ be a family of subsets of $E$. If  ${\cal S}' \subseteq_{rep}^{q} {\cal S}$  and $\widehat{{\cal S}}\subseteq_{rep}^{q} {\cal S}'$, then   \rep{S}{q}.   
\end{lemma}
\begin{proof}
Let $Y\subseteq  E$ of size at most $q$ such that there is a set $X \in \cal S$ disjoint from $Y$ with $X\cup Y \in \I$. By the definition of $q$-representative family we have that  there is a set $X'\in {\cal S}'$ disjoint from $Y$ with $X' \cup  Y \in \I$. Now the fact that 
$\widehat{{\cal S}}\subseteq_{rep}^{q} {\cal S}'$ yields that there exists a $\whnd{X} \in \whnd{\cal S}$ disjoint from $Y$ with $\whnd{X} \cup  Y \in \I$. 
\end{proof}
\begin{lemma}
\label{lem:repunion}
Let \mat{} be a matroid and $\cal S$ be a family of subsets of $E$. If  ${\cal S}={\cal S}_1\cup \cdots \cup {\cal S}_\ell$ and 
$\widehat{\cal{S}}_i\subseteq_{rep}^q {\cal S}_i$,  then $\cup_{i=1}^\ell \widehat{\cal{S}}_i \subseteq_{rep}^{q} {\cal S}$.   
\end{lemma}
\begin{proof}
Let $Y\subseteq  E$ of size at most $q$ such that there is a set $X \in \cal S$ disjoint from $Y$ with $X\cup Y \in \I$. Since  ${\cal S}={\cal S}_1\cup \cdots \cup {\cal S}_\ell$, there exists an $i$ such that $X\in {\cal S}_i$. This implies that there exists a $\whnd{X} \in \whnd{\cal S}_i \subseteq \cup_{i=1}^\ell \widehat{\cal{S}}_i $ disjoint from $Y$ with $\whnd{X} \cup  Y \in \I$. 
\end{proof}
\begin{lemma}
\label{lem:repconvolution}
Let \mat{} be a matroid of rank $k$ and ${\cal S}_1$ be a $p_1$-family of independent sets,  ${\cal S}_2$ be a $p_2$-family of independent sets,  $\widehat{\cal{S}}_1 \subseteq_{rep}^{k-p_1} {\cal S}_1$ and  $\widehat{\cal{S}}_2 \subseteq_{rep}^{k-p_2} {\cal S}_2$. Then $\widehat{\cal{S}}_1 \bullet \widehat{\cal{S}}_2 \subseteq_{rep}^{k-p_1-p_2} {\cal S}_1 \bullet {\cal S}_2$.
\end{lemma}
\begin{proof} 
Let $Y\subseteq  E$ of size at most $q = k-p_1-p_2$ such that there is a set $X \in {\cal S}_1\bullet {\cal S}_2$ disjoint from $Y$ with $X\cup Y \in \I$.  This implies that there exist $X_1\in {\cal S}_1$ and $X_2 \in {\cal S}_2$ such that $X_1\cup X_2=X$ and $X_1 \cap X_2=\emptyset$.  Since 
$\widehat{\cal{S}}_1 \subseteq_{rep}^{k-p_1} {\cal S}_1$, we have that there exists a $\widehat{X}_1\in \widehat{\cal{S}}_1$ such that 
$\widehat{X}_1 \cup X_2 \cup Y\in \cal I$ and $\widehat{X}_1 \cap (X_2 \cup Y) = \emptyset$. Now since $\widehat{\cal{S}}_2 \subseteq_{rep}^{k-p_2} {\cal S}_2$, we have that there exists a $\widehat{X}_2\in \widehat{\cal{S}}_2$ such that $\widehat{X}_1 \cup \widehat{X}_2 \cup Y\in \cal I$ and $\widehat{X}_2 \cap (\widehat{X}_1 \cup Y) = \emptyset$.  This shows that $\widehat{X}_1 \cup \widehat{X}_2  \in 
\widehat{\cal{S}}_1 \bullet \widehat{\cal{S}}_2 $ and $\widehat{X}_1 \cup \widehat{X}_2 \cup Y\in \cal I$ thus $\widehat{\cal{S}}_1 \bullet \widehat{\cal{S}}_2 \subseteq_{rep}^{k-p_1-p_2} {\cal S}_1 \bullet {\cal S}_2$.
\end{proof}

The main result of this section is that given a representable matroid \mat{} of rank $k=p+q$ with its representation matrix \repmat{M},   a $p$-family of 
independent sets $ \cS $, and a non-negative weight function \wf{}, we can compute \minrep{S}{q} and \maxrep{S}{q} of size ${p+q \choose p}$ deterministically in time \tgem. 
The proof for this result is obtained by making the known exterior algebra based proof of  Lov\' asz~\cite[Theorem 4.8]{Lovasz77} algorithmic.  
Although our proof is based on exterior algebra and is essentially the same as the proof given in~\cite{Lovasz77}, we give a proof here which avoids the terminology from exterior algebra.
 
 For our proof we also need the following well-known generalized Laplace expansion of determinants.  For a matrix $A=(a_{ij})$, the row set and the 
 column set are denoted by $\mathbf{R}(A)$ and $\mathbf{C}(A)$ respectively. For $I \subseteq \mathbf{R}(A)$ and $J\subseteq \mathbf{C}(A)$, 
 $A[I,J]=\big(a_{ij}~|~i \in I,~j\in J\big)$ means the submatrix (or minor) of $A$ with the row set $I$ and the column set $J$.  
 For $I\subseteq [n]$ let $\bar{I}=[n]\setminus I$ and $\sum I=\sum_{i\in I} i$. 
 
 \begin{proposition}[Generalized Laplace expansion]
 \label{prop:genlapaxp}
 For an $n\times n$ matrix $A$ and $J\subseteq \mathbf{C}(A)=[n]$, it holds that
\[\det(A)=\sum_{I\subseteq [n], |I|=|J|} (-1)^{\sum I + \sum J} \det(A[I,J]]) \det(A[\bar{I}, \bar{J}]) \]
 \end{proposition}

We refer to \cite[Proposition 2.1.3]{murota2000matrices}  for a proof of the above identity.  We always assume that 
the number of rows in the representation matrix \repmat{M}  of $M$ over a field $\mathbb F$ is equal to  \rank{M}$=$\rank{A_M}. Otherwise, 
using Gaussian elimination we can obtain a matrix of the desired kind in polynomial time. See~\cite[Proposition 3.1]{Marx09} for details. 
%
%
%
%
%
%
We do not give the proof for Theorem~\ref{thm:repsetlovasz} but rather for the following generalization.
\begin{theorem}
\label{thm:repsetlovaszweighted}
Let \mat{}   be a linear matroid of rank $p+q=k$, $ \cS = \{S_1,\ldots, S_t\}$ be a $p$-family of independent sets and 
\wf{} be a non-negative weight function. Then there exists \minrep{S}{q} (\maxrep{S}{q}) of size \bnoml{p+q}{p}.  
Moreover, 
given a representation \repmat{M}  of $M$ over a field $ \mathbb{F}$, we can find  \minrep{S}{q} (\maxrep{S}{q}) of size at most   \bnoml{p+q}{p} in   \tgem \, operations over $ \mathbb{F}$. 
 \end{theorem}

\begin{proof}
We only show how to  find  \minrep{S}{q} in the claimed running time. The proof for \maxrep{S}{q} is analogous, and for that case we only  point out the   places   where the proof differs.  
 If $t\leq \bnomlwd{k}{p}$,  then we can take  $\widehat{\cal S}=\cal S$. Clearly, in this case \minrep{S}{q}. So from now onwards we always assume that $t> \bnomlwd{k}{p}$. For the proof we view the representation matrix \repmat{M} as a vector space over $\mathbb F$ and each set $S_i\in \cS$  as a subspace of this vector space. For every element $e \in E$, let  $x_e$ be the corresponding $k$-dimensional column  in \repmat{M}. Observe that each  $x_e \in \mathbb{F}^{k}$.  For each subspace $S_i\in \cS$, $i\in \{1,\dots, t\}$, we associate a vector $\vec{s}_i=\bigwedge_{j\in S_i}x_j$ in $\mathbb{F}^{k \choose p}$  as follows.  In exterior algebra terminology, the vector  $\vec{s}_i$ is a wedge product of the vectors corresponding to elements in $S_i$.  For a set $S\in \cal S$ and $I \in {[k]\choose p}$, we define   $s[I]=\det(A_M[I,S])$. 

We also define 
\[\vec{s}_i=\left(s_i[I]\right)_{I\in {[k]\choose p} }.\]
Thus the entries of the vector $\vec{s}_i$ are the values of $\det(A_M[I,S_i])$, where $I$ runs through all the $p$ sized subsets of rows of $A_M$.  



Let $H_\mathcal{S}=(\vec{s}_1,\ldots, \vec{s}_t)$ be the ${k \choose p} \times t$ matrix obtained by taking $\vec{s}_i$ as columns.  
Now we define a weight function $w': \mathbf{C}(H_\mathcal{S}) \rightarrow \mathbb{R}^+$ on the set of columns of $H_\mathcal{S}$. For the column $\vec{s}_i$ corresponding to $S_i \in {\cal S}$, we define  $w'(\vec{s}_i)=w(S_i)$.  Let $\mathcal{W}$ be a set of columns of  $H_\mathcal{S}$ 
that are linearly independent over $\mathbb{F}$,  the size of $\mathcal{W}$ is equal to the \rank{H_\mathcal{S}} and is of minimum total weight 
with respect to the weight function $w'$.  That is, $\mathcal W$ is a minimum weight column basis of $H_\mathcal{S}$.  
Since   the row-rank of a matrix    is equal to the column-rank, we have that $|\mathcal{W}|=$\rank{H_\mathcal{S}}$\leq {k \choose p}$.  We define $\widehat{\cal S}=\{S_\alpha~|~\vec{s}_\alpha \in \mathcal{W}\}$. Let  $|\widehat{\cal S}|=\ell$.  Because  $ |\mathcal{W}| =|\widehat{\cal S}|$, we have that  $\ell \leq {k \choose p}$.  Without loss of generality,  let $\widehat{\cal S}=\{S_i~|~1\leq i \leq \ell\}$ (else we can rename these sets) and 
$\mathcal{W}=\{\vec{s}_1\ldots,\vec{s}_\ell\}$.  The only thing that remains to show is that indeed \minrep{S}{q}. 

Let $S_\beta \in \cal S$ be  such that $S_\beta \notin \widehat{\cal S}$. We show that if there is a set $Y\subseteq  E$ of size at most $q$ such that $S_\beta \cap Y=\emptyset$ and $S_\beta \cup Y \in \I$, then there exists a set $\whnd{S}_\beta \in \whnd{\cal S}$ disjoint from $Y$ with $\whnd{S}_\beta \cup  Y \in \I$ and $w(\whnd{S}_\beta)\leq  w(S_\beta)$. Let us first consider the case $|Y|=q$. Since $S_\beta \cap Y =\emptyset$, it follows that   $|S_\beta \cup Y|=p+q=k$. Furthermore, since $S_\beta \cup Y \in \I$, we have that the columns corresponding to $S_\beta \cup Y$ in $A_M$ are linearly independent over $\mathbb{F}$; that is, $\det(A_M[\mathbf{R}(A_M),S_\beta \cup Y])\neq 0$.  


Recall that,  
$\vec{s}_\beta=\left(s_\beta[I]\right)_{I\in {[k]\choose p} }, $ where $s_\beta[I]=\det(A_M[I,S_\beta])$.
 Similarly we define  $y[L]=\det(A_M[L,Y])$ and 
\[\vec{y}= \left(y[L]\right)_{L \in {[k]\choose q}}.
\]
 
Let $\sum J=\sum_{j\in S_\beta} j$. Define 
\[\gamma(\vec{s}_\beta,\vec{y})=\sum_{I\in {[k]\choose p}}(-1)^{\sum{I}+\sum J}s_\beta[I]\cdot y[{\bar{I}}].\]
Since ${k \choose p}={k \choose k-p}={k \choose q}$ the  above  formula is well defined. 
Observe that by Proposition~\ref{prop:genlapaxp}, we have that $\gamma(\vec{s}_\beta,\vec{y})=\det(A_M[\mathbf{R} (A_M),S_\beta\cup Y]) \neq 0$.  We  also know that $\vec{s}_\beta$ can be written as a linear combination of vectors in $\mathcal{W}=\{\vec{s}_1, \vec{s}_2,\ldots,\vec{s}_\ell\}$.  
That is,  $\vec{s}_\beta=\sum_{i=1}^\ell \lambda_i \vec{s}_i$, $\lambda_i \in \mathbb{F}$, and for some $i$, $\lambda_i\neq 0$.  
Thus, 
\begin{eqnarray*}
\gamma(\vec{s}_\beta,\vec{y})& = & \sum_{I}(-1)^{\sum{I}+\sum J}s_\beta[I]\cdot y[{\bar{I}}] \\
&=&  \sum_{I}(-1)^{\sum{I}+\sum J}\left (\sum_{i=1}^{\ell} \lambda_i s_{i}[I] \right)y[{\bar{I}}] \\
& = & \sum_{i=1}^{\ell} \lambda_i \left(\sum_{I}(-1)^{\sum{I}+\sum J} {s}_i[I]y[{\bar{I}}]\right) \\
&= & \sum_{i=1}^{\ell} \lambda_i \det(A_M[\mathbf{R}(A_M), S_i \cup Y]) ~~~~~(\mbox{by Proposition~\ref{prop:genlapaxp}})\\
\end{eqnarray*}
Define 
\[ \mathsf{sup}(S_\beta)=\Big\{S_i~\Big|~S_i\in \widehat{\cal S},~\lambda_i \det(A_M[\mathbf{R}(A_M), S_i \cup Y]))\neq 0 \Big\}. \]
Since  $\gamma(\vec{s}_\beta,\vec{y})\neq 0$,  we have that $(\sum_{i=1}^{\ell} \lambda_i \det(A_M[\mathbf{R}(A_M), S_i \cup Y]))\neq 0$ 
and thus $\mathsf{sup}(S_\beta)\neq \emptyset$. Observe that for all $S\in \mathsf{sup}(S_\beta)$ we have that 
$\det(A_M[\mathbf{R}(A_M), S \cup Y])\neq 0$ and thus $S\cup Y \in \cal I$.  
We now show that $w(S)\leq w(S_\beta)$ for all $S\in \mathsf{sup}(S_\beta)$. 
\begin{claim}
\label{claim:supportweighted}
For all $S\in \mathsf{sup}(S_\beta)$,  $w(S)\leq w(S_\beta)$.
\end{claim}
\begin{proof}
For a contradiction assume that there exists a set  $S_j \in \mathsf{sup}(S_\beta)$ such that $w(S_j)> w(S_\beta)$. Let 
$\vec{s_j}$ be the vector corresponding to $S_j$ and ${\cal W}'=({\cal W} \cup \{ \vec{s_j} \} ) \setminus \{\vec{s_\beta} \}$. Since 
$w(S_j)> w(S_\beta)$,  we have  that  $w(\vec{s_j})> w(\vec{s_\beta})$ and thus  $w'({\cal W})>w'({\cal W}')$. Now we show that 
${\cal W}'$ is also a column basis of $H_{\cal S}$. This will contradict our assumption that $\cal W$ is a minimum weight column basis of  
$H_{\cal S}$. Recall that $\vec{s}_\beta=\sum_{i=1}^\ell \lambda_i \vec{s}_i$, $\lambda_i \in \mathbb{F}$. 
Since $S_j \in \mathsf{sup}(S_\beta)$,  we have that $\lambda_j \neq 0$. Thus $\vec{s}_j$ can be written as linear combination of vectors in 
$\mathcal{W}'$. That is, 
\begin{eqnarray}
\label{eqn:lincomb}
\vec{s}_j= \lambda_{\beta} \vec{s}_\beta + \sum_{i=1, i \neq j}^\ell \lambda_i' \vec{s}_i.
\end{eqnarray}
Also every vector $\vec{s}_\gamma \notin {\cal W}$ can be written as a linear combination of vectors in $\cal W$  
\begin{eqnarray}
\label{eqn:lincombtwo}
\vec{s}_\gamma=\sum_{i=1}^\ell \delta_i \vec{s}_i,~~ \delta_i\in \mathbb{F}. 
\end{eqnarray}
By substituting \eqref{eqn:lincomb} into  \eqref{eqn:lincombtwo}, we conclude that every 
vector can be written as linear combination of vectors in ${\cal W}'$. This shows that ${\cal W}'$ is also a column basis of $H_{\cal S}$, a contradiction proving the claim. 
\end{proof}
Claim~\ref{claim:supportweighted} and the discussions preceding  above it show that we could take any set $S\in \mathsf{sup}(S_\beta)$ as the desired  $\whnd{S}_\beta \in \whnd{\cal S}$. Also, since $\det(A_M[\mathbf{R}(A_M), S\cup Y])\neq 0$, we have that $S\cap Y=\emptyset$. This shows that indeed \minrep{S}{q} for each $Y$  of size $q$. This completes  the proof for the case $|Y|=q$.

Suppose that $|Y|=q'< q$. Since $M$ is a matroid of rank $k=p+q$, there exists a superset $Y'\in \mathcal{I}$  of $Y$ of size $q$ such that   $S_\beta \cap Y' =\emptyset$ and $S_\beta \cup Y' \in \mathcal{I}$. This implies that there exists a  
set $\widehat{ S}\in \widehat{\cal S}$ such that $\det(A_M[\mathbf{R}(A_M),\widehat{S} \cup Y'])\neq 0$ and $w(\widehat{S})\leq w(S)$. 
Thus  the columns corresponding to $\widehat{S}\cup Y$ are linearly independent. 

We now consider the running time of the algorithm. 
To make the above proof algorithmic we need to 
\begin{itemize}
\setlength{\itemsep}{-2pt}
\item[(a)] compute determinants and 
\item[(b)] apply fast Gaussian elimination to find a 
 minimum weight column basis. 
 \end{itemize}
 It is well known that one can 
 compute the determinant of a $n \times n$ matrix in time $\cO(n^{\omega})$~\cite{bunch1974triangular}.  For a rectangular matrix 
 $A$ of size $d \times n$ (with $d\leq n$), Bodlaender et al.~\cite{BodlaenderCK12} outline an algorithm  computing a  minimum weight column basis 
 in time $\cO(nd^{\omega-1})$.   
Thus given a $p$-family of independent sets $ \cS$ we can 
construct the matrix $H_\mathcal{S}$ as follows. For every set $S_i$, we first 
compute $\vec{s}_i$. To do  this we compute  $\det(A_M[I,S_i])$ for every $I \in {[k]\choose p}$. This can be done in time 
$\cO(\bnomlwd{p+q}{p}p^\omega)$. Thus, we can obtain the matrix $H_\mathcal{S}$ in time $\cO(\bnomlwd{p+q}{p}tp^\omega)$. Given 
matrix $H_\mathcal{S}$ we  can find a minimum weight column basis $\mathcal{W}$ 
of $H_\mathcal{S}$ 
in time $\cO(t \bnomlwd{p+q}{p}^{\omega-1})$.  Given $\mathcal{W}$, we  can  easily recover $\widehat{\cal S}$. Thus, we can compute 
\minrep{S}{q} in  \tgem \, field operations. This concludes the proof  for finding \minrep{S}{q}. To find  \maxrep{S}{q}, the only change we need to 
do in the algorithm for finding \minrep{S}{q} is to  find a {\em maximum weight column basis} $\mathcal{W}$ of $H_\mathcal{S}$. This concludes the proof.
\end{proof}

In Theorem~\ref{thm:repsetlovaszweighted}  we assumed that \rank{M}$=p+q$. However, one can obtain a similar result even when 
\rank{M}$>p+q$ in lieu of randomness.  To do this we first need to compute the representation matrix of a $k$-restriction of \mat. 
For that we make use of  Proposition~\ref{prop:truncationrep}. This step returns a representation of a $k$-restriction of \mat{} with high probability.  Given this matrix, we apply Theorem~\ref{thm:repsetlovaszweighted} and arrive at  the following result. 


\begin{theorem}
\label{thm:repsetlovaszrandomized}
Let \mat{}   be a linear matroid, $ \cS = \{S_1,\ldots, S_t\}$ be a $p$-family of independent sets and 
\wf{} be a non-negative weight function. Then there exists \minrep{S}{q} (\maxrep{S}{q}) of size \bnoml{p+q}{p}.  
Furthermore, given a representation \repmat{M}  of $M$ over a field $ \mathbb{F}$, there is a randomized algorithm  computing  \minrep{S}{q} (\maxrep{S}{q}) of size at most   \bnoml{p+q}{p} in   \tgem \, operations over $ \mathbb{F}$. 
 \end{theorem}

\section{Fast Computation for Representative Sets  for Uniform Matroids}\label{sec:unimat}

In this section we show that for uniform matroids one can avoid matrix multiplication computations in order to  compute representative families. 
 The section is organized as follows. We start (Section~\ref{subsec:lopsided}, Theorem~\ref{thm:naiveRepUniform}) from a  relatively simple algorithm  computating   representative families  over a uniform matroid. This algorithm is already faster than the algorithm of Theorem~\ref{thm:repsetlovasz} for general matroids. In Section~\ref{subsec:sepCollect},  Theorem~\ref{thm:repset uniform general3},   we give an even faster, but more complicated algorithm. 
 Throughout this section a subfamily ${\cal A}' \subseteq {\cal A}$ of the family ${\cal A}$ is said to $q$-{\em represent} ${\cal A}$ if for every set $B$ of size $q$ such that there is an $A \in {\cal A}$ and $A \cap B = \emptyset$, there is a set $A' \in {\cal A}'$ such that $A' \cap B = \emptyset$.

\subsection{Representative Sets using Lopsided Universal Sets}\label{subsec:lopsided}
Our aim in this subsection is to prove the following theorem.

\begin{theorem}\label{thm:naiveRepUniform}
There is an algorithm that given a family ${\cal A}$ of $p$-sets   over a universe $U$ of size $n$ and an integer $q$, computes in time $|{\cal A}|\cdot {p+q \choose p} \cdot 2^{o(p+q)} \cdot \log n$ a subfamily ${\cal A}' \subseteq A$ such that $|{\cal A}'| \leq {p+q \choose p} \cdot 2^{o(p+q)} \cdot \log n$ and ${\cal A'}$ $q$-represents ${\cal A}$.
\end{theorem}

The main tool in our proof of Theorem~\ref{thm:naiveRepUniform} is a generalization of the notion of $n$-$k$-{\em universal} families.  A family ${\cal F}$ of sets over a universe $U$ is an $n$-$k$-{\em universal} family if for every set $A \in  {U \choose k}$ and every subset $A' \subseteq A$ there is some set $F \in {\cal F}$ whose intersection $F \cap A$ is exactly   $A'$. Naor et al.~\cite{NaorSS95} show that given $n$ and $k$ one can construct an $n$-$k$-universal family ${\cal F}$ of size $2^{k+o(k)} \cdot \log n$ in time $2^{k+o(k)} \cdot n \log n$. 

We tweak the notion of universal families as follows. We will say that a  family ${\cal F}$ of sets over a universe $U$ of size $n$ is an $n$-$p$-$q$-{\em lopsided-universal} family if for every $A \in {U \choose p}$ and $B \in {U \setminus A \choose q}$ there is an $F \in {\cal F}$ such that $A \subseteq F$ and $B \cap F = \emptyset$. An alternative definition that is easily seen to be equivalent is that ${\cal F}$ is $n$-$p$-$q$-lopsided-universal if for every subset $A \in  {U \choose p+q}$ and every subset $A' \in {A \choose p}$, there is an $F \in {\cal F}$ such that $F \cap A = A'$. From the second definition it follows that a $n$-$(p+q)$-universal family is also $n$-$p$-$q$-lopsided-universal. Thus the construction of Naor et al.~\cite{NaorSS95} of universal set families also gives an construction of $n$-$p$-$q$-lopsided universal family of size $2^{p+q+o(p+q)} \cdot \log n$, running in time $2^{p+q+o(p+q)} \cdot n \log n$.  It turns out that by slightly changing the construction of Naor et al.~\cite{NaorSS95}, one can prove the following result.
\begin{lemma}\label{lem:lopsidedUniversal}
There is an algorithm that given $n$, $p$ and $q$ constructs an  $n$-$p$-$q$-lopsided-universal family ${\cal F}$ of size ${p+q \choose p} \cdot 2^{o(p+q)} \cdot \log n$ in time $\cO({p+q \choose p} \cdot 2^{o(p+q)} \cdot n \log n)$.
\end{lemma}
We do not give a stand-alone proof of Lemma~\ref{lem:lopsidedUniversal}, however Lemma~\ref{lem:lopsidedUniversal} is a direct corollary of Lemma~\ref{lem:twin_sep_coll_construction} proved in Section~\ref{section:npqsepcollection}. 
We will now show how to use the lemma to prove Theorem~\ref{thm:naiveRepUniform}.

\begin{proof}[Proof of Theorem~\ref{thm:naiveRepUniform}]
The algorithm starts by constructing an $n$-$p$-$q$-lopsided universal family ${\cal F}$ as guaranteed by Lemma~\ref{lem:lopsidedUniversal}. If $|{\cal A}| \leq |{\cal F}|$ the algorithm outputs ${\cal A}$ and halts. Otherwise it builds the set ${\cal A}'$ as follows. 
Initially ${\cal A}'$ is equal to $\emptyset$ and all sets in ${\cal F}$ are marked as unused. The algorithm goes through every $A \in {\cal A}$ and unused sets $F \in {\cal F}$. If an unused set $F \in {\cal F}$ is found such that $A \subseteq F$, the algorithm marks $F$ as used, inserts $A$ into ${\cal A}'$ and proceeds to the next set in ${\cal A}$. If no such set $F$ is found the algorithm proceeds to the next set in ${\cal A}$ without inserting $A$ into ${\cal A}'$.

The size of ${\cal A}'$ is upper bounded by $|{\cal F}| \leq {p+q \choose p} \cdot 2^{o(p+q)} \cdot \log n$ since every time a set is added to ${\cal A}'$ an unused set in ${\cal F}$ is marked as used. For the running time analysis, constructing ${\cal F}$ takes time ${p+q \choose p} \cdot 2^{\cO(\frac{p+q}{\log \log (p+q)})}  \cdot n \log n$. Then we run through all of ${\cal F}$ for each set $A \in {\cal A}$, spending time $|{\cal A}|\cdot|{\cal F}|\cdot (p+q)^{\cO(1)}$, which is at most $|{\cal A}|\cdot {p+q \choose p} \cdot 2^{o(p+q)} \cdot \log n$. Thus in total the running time is bounded by $|{\cal A}|\cdot {p+q \choose p} \cdot 2^{o(p+q)} \cdot \log n$. 

Finally we need to argue that ${\cal A}'$ $q$-represents ${\cal A}$. Consider any set $A \in {\cal A}$ and $B$ such that $|B|=q$ and $A \cap B = \emptyset$. If $A \in {\cal A}'$ we are done, so assume that $A \notin {\cal A}'$. Since ${\cal F}$ is $n$-$p$-$q$-lopsided universal there is a set $F \in {\cal F}$ such that $A \subseteq F$ and $F \cap B = \emptyset$. Since  $A \notin {\cal A}'$ we know that $F$ was already marked as used when $A$ was considered by the algorithm. When the algorithm marked $F$ as used it also inserted a set $A'$ into ${\cal A}'$. For the insertion to be made, $F$ must satisfy $A' \subseteq F$. But then $A' \cap B = \emptyset$, completing the proof.
\end{proof}

One of the factors that drive up the running time of the algorithm in Theorem~\ref{thm:naiveRepUniform} is that one needs to consider all of ${\cal F}$ for each set $A \in {\cal A}$. Doing some computations it is possible to convince oneself that in an  $n$-$p$-$q$-lopsided universal family ${\cal F}$ the number of sets $F \in {\cal F}$ containing a fixed set $A$ of size $p$ should be approximately $|{\cal F}| \cdot \big(\frac{p}{p+q}\big)^p$. Thus, if we could only make sure that this estimation is in fact correct for every $A \in {\cal A}$, {\em and} we could make sure that for a given $A \in {\cal A}$ we can list all of the sets in ${\cal F}$ that contain $A$ without having to go through the sets that don't, then we could speed up our algorithm by a factor $\big(\frac{p+q}{p}\big)^p$. This is exactly the strategy behind the main theorem of Section~\ref{subsec:sepCollect}.


\subsection{Representative Sets using Separating Collections}\label{subsec:sepCollect}

\label{section:npqsepcollection}
In this section we design a faster algorithm to find $q$-representative family. Our main technical tool is 
a construction of {\em $n$-$p$-$q$-separating collection}.  
We start with the formal definition of {\em  $n$-$p$-$q$-separating collection}.

\begin{definition}
\label{def:twincollection}
An  $n$-$p$-$q$-separating collection ${\cal C}$ is a tuple $({\cal F}, \chi, \chi')$, where ${\cal F}$ is a family of sets over a universe $U$ of size $n$, $\chi$ is a function from  $\underset{p'\leq p}{\bigcup} {U\choose p'}$ 
to $2^{\cal F}$ and $\chi'$ is a function from $ \underset{q' \leq q}{\bigcup} {U\choose q'}$ 
to $2^{\cal F}$ such that the following properties are satisfied
\begin{enumerate}
 \item for every $A\in  \underset{p'\leq p}{\bigcup} {U\choose p'}$ 
 and $F \in \chi(A)$, $A \subseteq F$,
 \item for every $B\in  \underset{q' \leq q}{\bigcup} {U\choose q'}$  
 and $F \in \chi'(B)$, $F\cap B=\emptyset$, 
 \item for every pairwise disjoint sets $A_1\in {U \choose p_1},A_2\in {U \choose p_2},\cdots, A_r\in {U \choose p_r}$ and $B \in {U \choose q}$ such that $p_1+\cdots+p_r=p$, 
$\exists F\in \chi(A_1)\cap\chi(A_2)\ldots\chi(A_r)\cap \chi'(B).$
\end{enumerate}
The size of  $({\cal F},\chi, \chi')$ is $|{\cal F}|$, the $(\chi,p')$-degree of $({\cal F},\chi,\chi')$ for $p'\leq p$ is 
\[\max_{A \in {U \choose p'}} |\chi(A)|,\] 
and the $(\chi',q')$-degree of $({\cal F},\chi, \chi')$ for $q'\leq q$ is $$\max_{B \in {U \choose q'}} |\chi'(B)|.$$ 
\end{definition}

We must remark that the definition of an $n$-$p$-$q$-separating collection in the preliminary version of this paper ~\cite{FominLS14} was slightly more restricted than the one given here.  This new definition has already been used recently to obtain faster algorithms for computing representative sets for product families~\cite{FominLPS14}. 

A {\em construction} of  separating collections is a data structure, that given $n$, $p$ and $q$ initializes and outputs a family ${\cal F}$ of sets over the universe $U$ of size $n$. 
After the initialization one can query the data structure by giving it a set  $A \in \bigcup_{p'\leq p}{U \choose p'}$ or $B\in \bigcup_{q'\leq q}{U \choose q'}$, the data structure 
then outputs a family  $\chi(A) \subseteq 2^{\cal F}$ or $\chi'(B)\subseteq 2^{\cal F}$ respectively. Together the tuple ${\cal C}= ({\cal F},\chi, \chi')$ computed by the data structure 
should form a {\em } $n$-$p$-$q$-{\em separating collection}.

We call the time the data structure takes to initialize and output ${\cal F}$ the {\em initialization time}. 
The {\em $(\chi,p')$-query time}, $p'\leq p$, of the data structure is the maximum time the data structure uses to compute $\chi(A)$ over all $A \in {U \choose p'}$. Similarly, the 
{\em $(\chi',q')$-query time}, $q'\leq q$, of the data structure is the maximum time the data structure uses to compute $\chi'(B)$ over all $B \in {U \choose q'}$.
The initialization time of the data structure and the size of ${\cal C}$ are functions of  $n$, $p$ and $q$. The initialization time is denoted by 
$\tau_I(n,p,q)$, size of ${\cal C}$ is denoted by $\zeta(n,p,q)$. The $(\chi,p')$-query time and $(\chi,p')$-degree of 
$\cal C$, $p'\leq p$, are functions of 
$n,p',p,q$ and is denoted by ${Q_{(\chi,p')}}(n,p,q)$ and $\Delta_{(\chi,p')}(n,p,q)$ respectively. Similarly, the $(\chi',q')$-query time and $(\chi',q')$-degree of ${\cal C}$, $q'\leq q$,  are functions of 
$n,q',p,q$ and are denoted by ${Q_{(\chi',q')}}(n,p,q)$ and $\Delta_{(\chi',q')}(n,p,q)$ respectively.  
We are now ready to state the main technical  tool of this subsection.

\begin{lemma}
\label{lem:twin_sep_coll_construction}
Given $0<x<1$, there is a construction of  $n$-$p$-$q$- separating collection with the following parameters
\begin{itemize} 
\setlength\itemsep{-.7mm}
\item size, $\zeta(n,p,q) \leq 2^{\cO(\frac{p+q}{\log\log(p+q)})}\cdot \frac{1}{x^p(1-x)^q}\cdot (p+q)^{\cO(1)} \cdot \log n$
\item initialization time, $\tau_I(n,p,q) \leq  2^{\cO(\frac{p+q}{\log\log(p+q)})}\cdot \frac{1}{x^p(1-x)^q}\cdot (p+q)^{\cO(1)} \cdot n\log n$
\item $(\chi,p')$-degree, $\Delta_{(\chi,p')}(n,p,q) \leq  2^{\cO(\frac{p+q}{\log\log(p+q)})}\cdot \frac{1}{x^{p-p'}(1-x)^q}\cdot (p+q)^{\cO(1)} \cdot \log n$
\item $(\chi,p')$-query time, $Q_{(\chi,p')}(n,p,q) \leq  2^{\cO(\frac{p+q}{\log\log(p+q)})}\cdot \frac{1}{x^{p-p'}(1-x)^q}\cdot (p+q)^{\cO(1)} \cdot \log n$
\item $(\chi',q')$-degree, $\Delta_{(\chi',q')}(n,p,q) \leq  2^{\cO(\frac{p+q}{\log\log(p+q)})}\cdot \frac{1}{x^{p}(1-x)^{q-q'}}\cdot (p+q)^{\cO(1)} \cdot \log n$
\item $(\chi',q')$-query time, $Q_{(\chi',q')}(n,p,q) \leq  2^{\cO(\frac{p+q}{\log\log(p+q)})}\cdot \frac{1}{x^{p}(1-x)^{q-q'}}\cdot (p+q)^{\cO(1)} \cdot \log n$
\end{itemize}
\end{lemma}

We first give the road map that we take to prove Lemma~\ref{lem:twin_sep_coll_construction}. 
The proof of  Lemma~\ref{lem:twin_sep_coll_construction} uses three auxiliary lemmata. 
\begin{enumerate}
\item[(a.)] {\bf Existential Proof (Lemma~\ref{lem:twin_sep_coll_brute_force}}). This lemma shows that there is indeed a 
 $n$-$p$-$q$-separating collection with the required sizes, degrees and query time. Essentially, it shows that if we form a family  ${\cal F}=\{F_1,\ldots,F_t\}$ of sets of $U$ such that each $F_i$ is a random subset of $U$ where each element  is inserted into $F_i$ with probability $x$, then ${\cal F}$ has the desired sizes, degrees and query time. Thus, this also gives a brute force algorithm to design the family $\cal F$ by just guessing the family of desired size and then checking whether it is indeed  a  $n$-$p$-$q$-separating collection. 
\item[(b.)]  {\bf Universe Reduction (Lemma~\ref{lem:twinreduceUniverse}).} The construction obtained in Lemma~\ref{lem:twin_sep_coll_brute_force} 
has only one drawback that the initialization time is much larger than claimed in Lemma~\ref{lem:twin_sep_coll_construction}. To overcome this lacuna, we do not apply the construction in Lemma~\ref{lem:twin_sep_coll_brute_force} directly. 
We first prove a Lemma~\ref{lem:twinreduceUniverse} which helps us in reducing the universe size to $(p+q)^2$. This is done using the  known construction of $k$-perfect hash families of size $(p+q)^{\cO(1)} \log n$. However, 
Lemma~\ref{lem:twinreduceUniverse} alone  can not reduce the universe size sufficiently, that we can apply the construction of Lemma~\ref{lem:twin_sep_coll_brute_force}. 
\item[(c.)] {\bf Splitting Lemma (Lemma~\ref{lem:splitSolution}).} We give a splitter type construction in Lemma~\ref{lem:splitSolution} that when applied with 
Lemma~\ref{lem:twinreduceUniverse} makes the universe and other parameters small enough that we can apply the construction given in 
Lemma~\ref{lem:twin_sep_coll_brute_force}. In this construction we consider all the ``consecutive partitions''  of the universe into $t$ parts, assume that the sets $A\cup B$, $A=\cup_{i=1}^r A_i$, are distributed uniformly into $t$ parts and then use this information to obtain a construction of  separating collections in each part and then take the product of these collections to obtain a collection for the original instance.
\end{enumerate}

We start with the existential proof.

\begin{lemma}\label{lem:twin_sep_coll_brute_force}
Given $0<x<1$, there is a construction of  $n$-$p$-$q$-separating collections with 
\begin{itemize}\setlength\itemsep{-.7mm}
\item size $\zeta(n,p,q) =\cO \left(\frac{1}{x^{p}(1-x)^q} \cdot (p^2+q^2+1)\log n \right)$ 
\item initialization time $\tau_I(n,p,q) = \cO({2^n \choose \zeta(n,p,q)} \cdot \frac{1}{x^p(1-x)^q} \cdot n^{\cO(p+q)})$
\item $(\chi,p')$-degree for $p'\leq p$, $\Delta_{(\chi,p')}(n,p,q) = \cO\left(\frac{1}{x^{p-p'}}\cdot\frac{(p^2+q^2+1)}{(1-x)^q} \cdot \log n\right)$
\item $(\chi,p')$-query time ${Q_{(\chi,p')}}(n,p,q) = \cO(\frac{1}{x^{p}(1-x)^q} \cdot n^{\cO(1)})$
\item $(\chi',q')$-degree $\Delta_{(\chi',q')}(n,p,q)=\cO\left(\frac{1}{x^{p}(1-x)^{q-q'}}\cdot(p^2+q^2+1) \cdot \log n\right)$
\item $(\chi',q')$-query time ${Q_{(\chi',q')}}(n,p,q)=\cO(\frac{1}{x^{p}(1-x)^q} \cdot n^{\cO(1)})$
\end{itemize}
\end{lemma}

\begin{proof}
We start by giving a randomized algorithm that with positive probability constructs a  
$n$-$p$-$q$-separating collection ${\cal C} = ({\cal F},\chi, \chi')$ with the desired size and degree parameters. 
We will then discuss how to deterministically compute such a ${\cal C}$ within the required time bound. Set 
$t = \frac{1}{x^p(1-x)^q} \cdot (p^2+q^2+1)\log n$ and construct the family ${\cal F} = \{F_1, \ldots, F_t\}$ as follows. 
Each set $F_i$ is a random subset of $U$, where each element of $U$ is inserted into $F_i$ with probability $x$. 
Distinct elements are inserted (or not) into $F_i$ independently, and the construction of the different sets in 
${\cal F}$ is also independent. For each  $A \in \bigcup_{p'\leq p}{U\choose p'}$ we set 
$\chi(A) = \{F \in {\cal F}~:~A \subseteq F\}$ and for each $B\in \bigcup_{q'\leq q}{U\choose q'}$ we set 
$\chi'(B)=\{F\in {\cal F}~:~F\cap B=\emptyset\}$.

The size of ${\cal F}$ is within the required bound by construction. We now argue that with positive probability 
$({\cal F},\chi, \chi')$ is indeed a  $n$-$p$-$q$-separating collection, and that the degrees of ${\cal C}$ 
is within the required bounds as well. For fixed sets $A \in {U \choose p}$, $B \in {U\setminus A \choose q}$, and 
integer $i \leq t$, we consider the probability that $A \subseteq F_i$ and $B \cap F_i = \emptyset$. 
This probability is $x^{p}(1-x)^q$. Since each $F_i$ is constructed independently from the other sets in ${\cal F}$, 
the probability that {\em no} $F_i$ satisfies $A \subseteq F_i$ and $B \cap F_i = \emptyset$ is
\begin{align*} \left(1 - x^p(1-x)^q\right)^t \leq e^{-(p^2+q^2+1)\log n} = \frac{1}{n^{p^2+q^2+1}}.\end{align*}
For a fixed $A_1,\ldots,A_r$ and $B$ (choices in condition $3$), the probability that no $F_i$ in 
$\chi(A_1)\cap\chi(A_2)\cap \cdots \cap\chi(A_r)\cap \chi'(B)$ is equal to the probability that no $F_i$ is in 
$\chi(A_1\cup A_2\cdots \cup A_r)\cap \chi'(B)$  (since $\chi(A')$ contains all the sets in ${\cal F}$ that contains 
$A'$ and $\chi'(B)$ contains all the sets in ${\cal F}$ that are disjoint from $B$). Hence the probability that 
condition $3$ fails is upper bounded by 
$$Y\cdot\frac{1}{n^{p^2+q^2+1}}$$
where $Y$ is the number of choices for $A_1,\ldots,A_r$ and $B$ in condition $3$. We upper bound $Y$ as follows.
There are ${n \choose p}$ choices for $A_1\cup\cdots\cup A_r$ and ${n \choose q}$ choices for $B$. 
For each choice of $A_1\cup\cdots\cup A_r$ there are at most $r^p$ choices of making $A_1,\ldots,A_r$ with some of 
them being empty as well. Note that $r\leq p$. Therefore the number of possible choices of sets $A_1,A_2,\ldots,A_r$ 
and $B$ in condition $3$ is upper bounded by ${n\choose p}{n\choose q}p^p\leq n^{2p+q}\leq n^{p^2+q^2}$. 
Hence the probability that condition $3$ in Definition~\ref{def:twincollection} fails is at most $\frac{1}{n}$.

We also need to upper bound the maximum degree of ${\cal C}$. For every $A \in  {U \choose p'}$, $|\chi(A)|$ is a random variable. For a fixed $A \in  {U \choose p'}$ and $i \leq t$ the probability 
that $A \subseteq F_i$ is exactly $x^{p'}$. Hence $|\chi(A)|$ is the sum of $t$ independent  $0/1$-random variables that each take value $1$ with probability $x^{p'}$. Hence the expected value of $|\chi(A)|$ is 
$$E[|\chi(A)|] = t \cdot x^{p'} = \frac{1}{x^{p-p'}(1-x)^q}\cdot (p^2+q^2+1)\log n$$ 
For every $B\in {U\choose q'}$, $|\chi'(B)|$ is also a random variable. For a fixed $B \in  {U \choose q'}$ and $i \leq t$ the probability that $A \cap F_i=\emptyset$ is exactly $(1-x)^{q'}$. 
Hence the expected value of $|\chi'(B)|$ is,
$$E[|\chi'(B)|] = t \cdot (1-x)^{q'} = \frac{1}{x^{p}(1-x)^{q-q'}}\cdot (p^2+q^2+1)\log n.$$
Standard Chernoff bounds~\cite[Theorem 4.4]{mitzenmacher2005probability} show that the probability that for any $A\in{U\choose p'}$, $|\chi(A)|$ 
is at least $6E[|\chi(A)|]$ is upper bounded by $2^{-6E[|\chi(A)|]} \leq \frac{1}{n^{p^2+q^2+1}}$. 
Similarly the probability that for any $B\in{U \choose q'}$, $|\chi'(B)|$ is at least $6E[|\chi'(B)|]$ is upper bounded by $2^{-6E[|\chi'(B)|]} \leq \frac{1}{n^{p^2+q^2+1}}$.  
There are  $\sum_{p'\leq p}{n \choose p'}\leq{n^{p^2}}$ choices for $A\in\bigcup_{p'\leq p}{U\choose p'}$ and $\sum_{q'\leq q}{n \choose q'}\leq{n^{q^2}}$ choices for $B\in\bigcup_{q'\leq q}{U\choose q'}$. 
Hence the union bound yields that the probability that there exists an $A\in\bigcup_{p'\leq p}{U\choose p'}$ such that $|\chi(A)|  > 6E[|\chi(A)|]$ or there exists $B\in \bigcup_{q'\leq q}{U\choose q'}$ such 
that $|\chi'(B)|  > 6E[|\chi'(B)|]$ is upper bounded by $\frac {1}{n}$. Thus ${\cal C}$ is a family of $n$-$p$-$q$-separating collections with the desired size and degree parameters with probability at least 
$1 - \frac{2}{n} > 0$. The degenerate case that  $1 - \frac{2}{n} \leq 0$ is handled 
by the family ${\cal F}$ containing all (at most four) subsets of $U$. 

To construct ${\cal F}$ within the stated initialization time bound, it is sufficient to try all families ${\cal F}$ of size $t$ and for each of the ${2^n \choose \zeta(n,p,q)}$ 
guesses, test whether it is indeed a family of $n$-$p$-$q$-separating collections in time $\cO(t \cdot n^{\cO(p+q)}) = \cO(\frac{1}{x^p(1-x)^q} \cdot n^{\cO(p+q)})$.

For the queries, we need to give an algorithm that given $A$, computes $\chi(A)$ (or $\chi'(A)$), under the assumption that ${\cal F}$ has already has been computed in the initialization step. 
This is easily done within the stated running time bound by going through every set $F \in {\cal F}$, checking whether $A \subseteq F$ (or $A\cap F=\emptyset$), and if so, inserting $F$ into $\chi(A)$ ($\chi'(A)$). 
This concludes the proof.
\end{proof}
We will now work towards improving the time bounds of Lemma~\ref{lem:twin_sep_coll_brute_force}. 
 To that end we will need a construction of {\em $k$-perfect hash functions} by Alon et al.~\cite{AlonYZ}
\begin{definition} 
A family of functions $f_1, \ldots, f_t$ from a universe $U$ of size $n$ to a universe of size $r$ is a $k$-perfect family of hash functions if for every set $S \subseteq U$ such that $|S|=k$ 
there exists an $i$ such that the restriction of $f_i$ to $S$ is injective.
\end{definition}
Alon et al.~\cite{AlonYZ} give very efficient constructions of $k$-perfect families of hash functions from a universe of size $n$ to a universe of size $k^2$.
\begin{proposition}[\cite{AlonYZ}]\label{prop:hashFun} 
For any universe $U$ of size $n$ there is a $k$-perfect family $f_1, \ldots, f_t$ of hash functions from $U$ to 
$[k^2]$ 
with $t = \cO(k^{\cO(1)} \cdot \log n)$. 
Such a family of hash functions can be constructed in time $\cO(k^{\cO(1)}n \log n)$. 
\end{proposition}

\begin{lemma}\label{lem:twinreduceUniverse} If there is a construction of  $n$-$p$-$q$-separating collections $(\hat{\cal F},\hat{\chi},\hat{\chi}')$ with initialization time $\tau_I(n,p,q)$, size $\zeta(n,p,q)$, 
$(\hat{\chi},p')$-query time ${Q_{(\hat{\chi},p')}}(n,p,q)$, $(\hat{\chi}',q')$-query time ${Q_{(\hat{\chi}',q')}}(n,p,q)$, 
$(\hat{\chi},p')$-degree $\Delta_{(\hat{\chi},p')}(n,p,q)$, and $(\hat{\chi}',q')$-degree $\Delta_{(\hat{\chi}',q')}(n,p,q)$  
then there is a construction of  $n$-$p$-$q$-separating collections 
with following parameters.
\begin{itemize} 
\item 
$\zeta'(n,p,q) \leq \zeta\left((p+q)^2,p,q\right) \cdot  (p+q)^{\cO(1)} \cdot \log n$,
\item 
$\tau_I'(n,p,q) = \cO\left(\tau_I\left((p+q)^2,p,q\right) + \zeta\left((p+q)^2,p,q\right) \cdot (p+q)^{\cO(1)} \cdot n \log n\right)$,
\item 
$\Delta'_{(\chi,p')}(n,p,q) \leq \Delta_{(\hat{\chi},p')}\left((p+q)^2,p,q\right) \cdot  (p+q)^{\cO(1)} \cdot \log n$, 
\item 
${Q'_{(\chi,p')}}(n,p,q) = \cO\left(\left({Q_{(\hat{\chi},p')}}\left((p+q)^2,p,q\right) + \Delta_{(\hat{\chi},p')}\left((p+q)^2,p,q\right) \right) \cdot (p+q)^{\cO(1)} \cdot \log n\right)$,
\item 
$\Delta'_{(\chi',q')}(n,p,q)\leq \Delta_{(\hat{\chi}',q')}\left((p+q)^2,p,q\right) \cdot  (p+q)^{\cO(1)} \cdot \log n$,
\item 
${Q'_{(\chi',q')}}(n,p,q)=\cO\left(\left({Q_{(\hat{\chi}',q')}}\left((p+q)^2,p,q\right) + \Delta_{(\hat{\chi}',q')}\left((p+q)^2,p,q\right) \right) \cdot (p+q)^{\cO(1)} \cdot \log n\right)$
\end{itemize}
\end{lemma}
\begin{proof}
We give a construction of  $n$-$p$-$q$-separating collections with initialization time, query time, size and degree $\tau_I'$, ${Q}'$, $\zeta'$ and $\Delta'$ respectively using the 
construction with initialization time, query time, size and degree $\tau_I$, ${Q}$, $\zeta$ and $\Delta$ as a black box. 
 
We first describe the initialization of the data structure. Given $n$, $p$, and $q$, we construct using Proposition~\ref{prop:hashFun} a $(p+q)$-perfect family $f_1, \ldots f_t$ 
of hash functions from the universe $U$ to $[(p+q)^2]$. The construction takes time $\cO((p+q)^{\cO(1)}n \log n)$ and $t \leq  (p+q)^{\cO(1)} \cdot \log n$. 
We will store these hash functions in memory. We use the following notations.
\begin{itemize}
 \item For a set $S \subseteq U$ and $T\subseteq [(p+q)^2]$,  \\ $f_i(S)=\{f_i(s) ~:~ s \in S\}$ and $f_i^{-1}(T)=\{s \in U ~:~ f(s) \in T\}$. 
 \item For a family ${\cal Z}$ of sets over $U$ and family ${\cal W}$ of sets over $[(p+q)^2]$,\\ $f_i({\cal Z})=\{f_i(S) ~:~ S \in {\cal Z}\}$ and $f_i^{-1}({\cal W})=\{f_i^{-1}(T) ~:~ T \in {\cal W}\}$.
\end{itemize}
 
 We first use the given black box construction for $(p+q)^2$-$p$-$q$-separating collections $(\hat{\cal F}, \hat{\chi},\hat{\chi}')$ over the universe $[(p+q)^2]$. 
We run the initialization algorithm of this construction and store the family $\hat{\cal F}$ in memory. We then set
\begin{align*} {\cal F} = \bigcup_{i\leq t} f_i^{-1}(\hat{\cal F}). \end{align*}

 We spent $\cO((p+q)^{\cO(1)}n \log n)$ time to construct a $(p+q)$-perfect family of hash functions,  $\cO(\tau_I((p+q)^2,p,q))$ to construct $\hat{\cal F}$ of size $\zeta((p+q)^2,p,q)$, 
and $\cO(\zeta((p+q)^2,p,q) \cdot (p+q)^{\cO(1)} \cdot n \log n)$ time to construct ${\cal F}$ from $\hat{\cal F}$ and the family of perfect hash functions. 
Thus the upper bound on $\tau_I'(n,p,q)$ follows. Furthermore,  $|{\cal F}| \leq |\hat{\cal F}| \cdot  (p+q)^{\cO(1)} \cdot \log n$, yielding the claimed bound for $\zeta'$.
 
 We now define $\chi(A)$ for every $A \in \bigcup_{p'\leq p}{U\choose p'}$ and describe the query algorithm. For every $A \in \bigcup_{p'\leq p}{U\choose p'}$ we let
\begin{align*} \chi(A) = \bigcup_{\substack{i\leq t \\ |f_i(A)|=|A|}} f_i^{-1}(\hat{\chi}(f_i(A))). \end{align*}
Since for every $\hat{F} \in \hat{\chi}(f_i(A))$, $f_i(A) \subseteq \hat{F}$, it follows that $A \subseteq F$ for every $F \in \chi(A)$. Furthermore we can bound $|\chi(A)|$ for any $A\in \bigcup_{p'\leq p}{U\choose p'}$, 
as follows 
\begin{align*} |\chi(A)| \leq \sum_{\substack{i\leq t \\ |f_i(A)|=|A|}} |\hat{\chi}(f_i(A))| \leq \Delta_{(\hat{\chi},p')}((p+q)^2,p,q) \cdot  (p+q)^{\cO(1)} \cdot \log n.\end{align*}
Thus the claimed bound for $\Delta'_{(\chi,p')}$ follows. 
Similarly,  way can define $\chi'(B)$ for every $B\in\bigcup_{q'\leq q}{U\choose q'}$ as 
\begin{align*} \chi'(B) = \bigcup_{\substack{i\leq t \\ |f_i(A)|=|A|}} f_i^{-1}(\hat{\chi}'(f_i(A))). \end{align*}
\begin{align*} |\chi'(B)| \leq \sum_{\substack{i\leq t \\ |f_i(A)|=|A|}} |\hat{\chi}'(f_i(A))| \leq \Delta_{(\hat{\chi}',q')}((p+q)^2,p,q) \cdot  (p+q)^{\cO(1)} \cdot \log n.\end{align*}
To compute $\chi(A)$ for any $A\in\bigcup_{p'\leq p} {U \choose p'}$, we go over every $i \leq t$ and check whether $f_i$ is injective on $A$. This takes time $\cO((p+q)^{\cO(1)} \cdot \log n)$. 
For each $i$ such that $f_i$ is injective on $A$, we compute $f_i(A)$ and then $\hat{\chi}(f_i(A))$ in time 
$\cO({Q_{(\hat{\chi},p')}}((p+q)^2,p,q))$. Then we compute $f_i^{-1}(\hat{\chi}(f_i(A)))$  in time 
$\cO(|\hat{\hat{\chi}}(f_i(A))|\cdot (p+q)^{\cO(1)}) =\cO(\Delta_{(\hat{\chi},p')}((p+q)^2,p,q)\cdot (p+q)^{\cO(1)})$ and add this set to $\chi(A)$. As we need to do this $\cO((p+q)^{\cO(1)} \cdot \log n)$ times, the total time 
to compute $\chi(A)$ is upper bounded by $\cO(({Q_{(\hat{\chi},p')}}((p+q)^2,p,q) + \Delta_{(\hat{\chi},p')}((p+q)^2,p,q)) \cdot (p+q)^{\cO(1)} \cdot \log n)$, 
yielding the claimed upper bound on ${Q'_{(\chi,p')}}$. Similar way we can bound ${Q'_{(\chi',q')}}$.

It remains to argue that $({\cal F},\chi,\chi')$ is in fact a  $n$-$p$-$q$-separating collection. For any $r$, consider pairwise disjoint sets $A_1 \in {U \choose p_1},\ldots,A_r \in {U \choose p_r}$, and 
$B \in {U \choose q}$ such that $p_1+\ldots+p_r=p$. We need to show that there is $  F\in \chi(A_1)\cap\cdots\cap\chi(A_r)\cap\chi'(B)$. 
Since $f_1, \ldots, f_t$ is a $(p+q)$-perfect family of hash functions, there is an $i$ such that $f_i$ is injective on $A_1\cup\cdots \cup A_r \cup B$. 
Since $(\hat{\cal F}, \hat{\chi},\hat{\chi}')$ is a $(p+q)^2$-$p$-$q$-separating collection,  
$\exists \hat{F}\in \hat{\chi}(f_i(A_1))\cap\cdots\hat{\chi}(f_i(A_r))\cap \hat{\chi}'(f_i(B))$. Since $f_i$ is injective on $A_1,\ldots,A_r$ and $B$, $f_i^{-1}(\hat{F})\in \chi(A_1)\cap\cdots \chi(A_r) \cap\chi'(B)$. 
This concludes the proof.
\end{proof}
We now give a {\em splitting lemma}, which allows us to reduce the problem of finding  $n$-$p$-$q$-separating collections to the same problem, but with much smaller values for $p$ and $q$. 

A {\em partition} of $U$ is a family ${\cal U}_P = \{U_1, U_2, \ldots U_t\}$ of sets over $U$ such that $U_i \cap U_j = \emptyset$ for every $i \neq j$ and $U = \bigcup_{i \leq t} U_i$. Each of the sets $U_i$ are called the {\em parts} of the partition. A {\em consecutive partition} of $\{1,\ldots,n\}$ is a partition ${\cal U}_P = \{U_1, U_2, \ldots U_t\}$ of $\{1,\ldots,n\}$ such that for every integer $i \leq t$ and integers $1 \leq x \leq y \leq z$, if $x \in U_i$ and $z \in U_i$ then $y \in U_i$ as well. In other words, in a consecutive partition each part is a consecutive interval of integers. For every integer $t$, let $\mathscr{P}_t^n$ denote the collection of all consecutive partitions of $\{1,\ldots,n\}$ with exaclty $t$ parts. We do not demand that all of the parts in a partition in $\mathscr{P}_t$ are non-empty. Simple counting arguments show that for every $t$, 
$|\mathscr{P}_t^n| = {n+t-1 \choose t-1}$.

We will denote by ${\cal Z}_{s,t}^p$ the set of all $t$-tuples $(p_1,p_2, \ldots, p_t)$ of integers such that $\sum_{i \leq t} p_i = p$ and $0 \leq p_i \leq s$ for all $i$. Clearly  $|{\cal Z}_{s,t}^p| \leq {p+t-1 \choose t-1}$, since this counts all the ways of writing $p$ as a sum of $t$ non-negative integers, without considering the upper bound on each one.  For an ease of convenience we 
summarize the above in the next definition and the proposition.


%
%

\begin{definition}
A {\em partition} of $U$ is a family ${\cal U}_P = \{U_1, U_2, \ldots U_t\}$ of sets over $U$ such that $\forall i\neq j,\;U_i \cap U_j = \emptyset$ and $U = \bigcup_{i \leq t} U_i$. 
Each of the sets $U_i$ are called the {\em parts} of the partition. A {\em consecutive partition} of $\{1,\ldots,n\}$ is a partition ${\cal U}_P = \{U_1, U_2, \ldots U_t\}$ of $\{1,\ldots,n\}$ 
such that for every integer $i \leq t$ and integers $1 \leq x \leq y \leq z$, if $x \in U_i$ and $z \in U_i$ then $y \in U_i$ as well. 
\end{definition}
\begin{proposition}
Let $\mathscr{P}_t^{n}$ denote the collection of all consecutive partitions of $\{1,\ldots,n\}$ with exactly $t$ parts. Let 
${\cal Z}_{s,t}^p$ be the set of all $t$-tuples $(p_1,p_2, \ldots, p_t)$ of integers such that $\sum_{i \leq t} p_i = p$ and $0 \leq p_i \leq s$ for all $i$. 
Then for every $t$, $|\mathscr{P}_t^{n}| = {n+t-1 \choose t-1}$ and $|{\cal Z}_{s,t}^p| \leq {p+t-1 \choose t-1}$.
\end{proposition}


\begin{lemma}\label{lem:splitSolution} For any $p$, $q$ let $s = \lfloor (\log (p+q))^2 \rfloor$ and 
$t = \lceil \frac{p+q}{s} \rceil$. 
If there is a construction of  $n$-$p$-$q$-separating collections  $({\cal F}_p,\chi_p,\chi'_p)$
\begin{itemize}
\setlength\itemsep{-.7mm}
\item with size $\zeta(n,p,q)$ and initialization time $\tau_I(n,p,q)$, 
\item $({\chi}_p,p')$-degree 
$\Delta_{({\chi}_p,p')}(n,p,q)$ and $({\chi}'_p,q')$-degree $\Delta_{({\chi}'_p,q')}(n,p,q)$, and  

\item query times ${Q_{({\chi}_p,p')}}(n,p,q)$ and ${Q_{({\chi}'_p,q')}}(n,p,q)$, 
\end{itemize}
then there is a construction 
of  $n$-$p$-$q$-separating collection with following parameters

\begin{itemize}
\setlength\itemsep{-.7mm}
\item 
\[\zeta'(n,p,q) \leq |\mathscr{P}_t^n| \cdot 
\sum_{\substack{(p_1,\ldots,p_t) \in {\cal Z}_{s,t}^{p}}} \prod_{i \leq t} \zeta(n,p_i,s-p_i),\]
\item 
\[ \tau_I'(n,p,q) = \cO\Big(\big(\sum_{\substack{\hat{p} \leq s,p\\ s-\hat{p}\leq q}} \tau_I(n,\hat{p},s-\hat{p})\big) + \zeta'(n,p,q) \cdot n^{\cO(1)}\Big),\]
\item 
$$\Delta_{(\chi,p')}'(n,p,q) \leq \Delta_{(\chi,p')}^*(n,p,q)=|\mathscr{P}_t^n|\cdot |{\cal Z}_{s,t}^p|\cdot
\max_{\substack{(p_1,\ldots,p_t) \in {\cal Z}_{s,t}^{p} \\ p_1'\leq p_1,\ldots, p_t'\leq p_t \\ p_1'+\cdots+p_t'=p' }} \prod_{i \leq t}  \Delta_{({\chi}_{p_i},p_i')}(n,p_i,s-p_i),$$

\item 
$${Q'_{(\chi,p')}}(n,p,q) =\cO \Big(\Delta^*_{(\chi,p')}(n,p,q) \cdot n^{\cO(1)} + |\mathscr{P}_{t}^n|\cdot |{\cal Z}_{s,t}^p| \cdot t\cdot \big( \max_{\substack{ \hat{p}'\leq
\hat{p}\leq s \\ \hat{p}-\hat{p}'\leq p-p' \\ s-\hat{p}\leq q }} Q_{(\chi_{\hat{p}},\hat{p}')}(n,\hat{p},s-\hat{p})\big)  \Big),$$
\item 
$$\Delta'_{(\chi',q')}(n,p,q) \leq \Delta^*_{(\chi',q')}(n,p,q)=|\mathscr{P}_t^n|\cdot |{\cal Z}_{s,t}^p|\cdot
\max_{\substack{(p_1,\ldots,p_t) \in {\cal Z}_{s,t}^{p} \\ q_1'\leq s-p_1,\ldots, q_t'\leq s-p_t \\ q_1'+\ldots+q_t'=q'}} \prod_{i \leq t}  \Delta_{({\chi}'_{p_i},q_i')}(n,p_i,s-p_i),$$ 
\item 
$${Q'_{(\chi',q')}}(n,p,q) =\cO \Big(\Delta^*_{(\chi',q')}(n,p,q) \cdot n^{\cO(1)} + |\mathscr{P}_{t}^n|\cdot |{\cal Z}_{s,t}^p| \cdot t\cdot \big( \max_{\substack{ \hat{q}'\leq
\hat{q}\leq s \\ \hat{q}-\hat{q}'\leq q-q' \\ s-\hat{q} \leq p }} 
Q_{(\chi'_{s-\hat{q}},\hat{q}')}(n,s-\hat{q},\hat{q})\big)  \Big).$$
\end{itemize}
\end{lemma}
\begin{proof}
Set $s = \lfloor (\log (p+q))^2 \rfloor$ and $t = \lceil \frac{p+q}{s} \rceil$.
We will give 
a construction of  
$n$-$p$-${q}$-separating collections
with initialization time, query time, size and degree 
within the claimed bounds above. In this construction we will use the given construction as a black box. 
We may assume without loss of generality that $U = \{1, \ldots, n\}$.
%
%
Our algorithm first runs  for every $\hat{p}$, $0\leq\hat{p}\leq s, \hat{p} \leq p , s-\hat{p}\leq {q}$, and initializes 
$n$-$\hat{p}$-$(s-\hat{p})$-separating collections,  $$({\cal F}_{\hat{p}},\chi_{\hat{p}},\chi'_{\hat{p}}).$$ 
These will be be the building blocks of our construction. 
%
 

We need to define a few operations on families of sets. For families of sets  ${\cal A}$, ${\cal B}$ over $U$ 
and subset $U' \subseteq U$ we define
\begin{eqnarray*}
{\cal A} \sqcap U' &=& \{A \cap U'~:~A \in {\cal A}\} \\ 
{\cal A} \circ {\cal B} &=& \{A \cup B ~:~A \in {\cal A} \wedge B \in {\cal B}\} 
\end{eqnarray*}
We now define ${\cal F}$ as follows.
\begin{align}\label{eqn:defineFSplit} 
{\cal F} = \bigcup_{\substack{\{U_1,\ldots,U_{t}\} \in \mathscr{P}_{t}^n \\ (p_1,\ldots,p_t) \in {\cal Z}_{s,t}^{p}~\mbox{\scriptsize such that} \\ \forall i~:~s-p_i\leq q } } 
({\cal F}_{p_1} \sqcap U_1) \circ ({\cal F}_{p_2} \sqcap U_2) \circ \ldots \circ ({\cal F}_{p_t} \sqcap U_{t}) 
\end{align}
It follows directly from the definition of  ${\cal F}$ that $|{\cal F}|$ is within the claimed bound for 
$\zeta'(n,p,q)$. For the initialization time, the algorithm spends 
$\cO\left(\sum_{\substack{\hat{p} \leq s,p\\ s-\hat{p}\leq q}} \tau_I(n,\hat{p},s-\hat{p})\right)$ time to initialize the constructions 
of the  $n$-$\hat{p}$-$(s-\hat{p})$-separating collections for all $\hat{p}\leq s$ such that $\hat{p}\leq p$ and $s-\hat{p}\leq q$ together. Now the 
algorithm can output the entries of ${\cal F}$ one set at a time by using Equation~\eqref{eqn:defineFSplit}, spending 
$n^{\cO(1)}$ time per output set. Hence the time bound for $\tau'_I(n,p,q)$ follows. 

For every set $A \in \bigcup_{p'\leq p}{U \choose p'}$ we define $\chi(A)$ as follows.
\begin{align}\label{eqn:defineChiSplit} 
\chi(A) = \bigcup_{\substack{\{U_1,\ldots,U_{t}\} \in \mathscr{P}_{t}^n\\(p_1,\ldots,p_t) \in {\cal Z}_{s,t}^{p}~\mbox{\scriptsize such that}\\ \forall U_i~:~|U_i \cap A| \leq p_i, s-p_i\leq q}} 
\Big[({\chi}_{p_1}(A \cap U_1) \sqcap U_1) \circ ({\chi}_{p_2}(A \cap U_2) \sqcap U_2) \circ \ldots \\
\nonumber ... \circ ({\chi}_{p_t}(A \cap U_t) \sqcap U_t)\Big] 
\end{align}

Now we show that $\chi(A)\subseteq {\cal F}$. From the definition of  $n$-${p_i}$-$(s-{p_i})$-separating 
collections $({\cal F}_{{p_i}},\chi_{{p_i}},\chi'_{{p_i}})$, each family $\chi_{p_i}(A\cap U_i)$ in Equation~
\eqref{eqn:defineChiSplit} is a subset of ${\cal F}_{p_i}$. This implies that 
$\chi_{p_i}(A\cap U_i)\sqcap U_i\subseteq {\cal F}_{p_i} \sqcap U_i$. Hence $\chi(A)\subseteq {\cal F}$.   
Similarly we can define $\chi'(B)$ for any $B\in\bigcup_{q'\leq q}{U\choose q'}$ as
\begin{align}\label{eqn:defineChiprimeSplit} 
\chi'(B) = \bigcup_{\substack{\{U_1,\ldots,U_{t}\} \in \mathscr{P}_{t}^n\\(p_1,\ldots,p_t) \in {\cal Z}_{s,t}^{p}~\mbox{\scriptsize such that} \\ \forall U_i~:~|U_i \cap B| \leq s-p_i\leq q }} 
\Big[({\chi}'_{p_1}(B \cap U_1) \sqcap U_1) \circ ({\chi}'_{p_2}(B \cap U_2) \sqcap U_2) \circ \cdots \\
\nonumber \cdots \circ ({\chi}'_{p_t}(B \cap U_t) \sqcap U_t)\Big] 
\end{align}

Similar to the proof of $\chi(A)\subseteq {\cal F}$, we can show that $\chi'(B)\subseteq {\cal F}$. 
It follows directly from the definition of  $\chi(A)$ and $\chi'(B)$ that $|\chi(A)|$ and $|\chi'(B)|$ is within the 
claimed bound for $\Delta_{(\chi,p')}'(n,p,q)$ and $\Delta_{(\chi',q')}'(n,p,q)$ respectively. We now describe how 
queries $\chi(A)$ can be answered, and analyze how much time it takes. Given $A$ we will compute $\chi(A)$ using  Equation~\eqref{eqn:defineChiSplit}. Let $|A|=p'$. 
For each $\{U_1,\ldots,U_{t}\} \in \mathscr{P}_{t}^n$ and $(p_1,\ldots,p_t) \in {\cal Z}_{s,t}^{p}$ such that $p_i' = |U_i \cap A| \leq p_i, s-p_i\leq q$ for all $i \leq t$, 
we proceed as follows. First we compute ${\chi}_{p_i}(A \cap U_i)$ for each $i \leq t$, spending in total 
$\cO(\sum_{i \leq t} Q_{(\chi_{p_i},p_i')}(n,p_i,s-p_i))$ time. Now we add each set in 
$$({\chi}_{p_1}(A \cap U_1) \sqcap U_1) \circ ({\chi}_{p_2}(A \cap U_2) \sqcap U_2) \circ \ldots \circ ({\chi}_{p_t}(A \cap U_t) \sqcap U_t)$$  
to $\chi(A)$, spending $n^{\cO(1)}$ time per set, yielding the bound below, 
\begin{align*}
Q'_{(\chi,p')}(n,p,q) \leq \cO\Big(\Delta^*_{(\chi,p')}(n,p,q) \cdot n^{\cO(1)} + \sum_{\substack{\{U_1,\ldots,U_{t}\} \in \mathscr{P}_{t}\\(p_1,\ldots,p_t) \in {\cal Z}_{s,t}^{p}~\mbox{\scriptsize such that}\\ \forall U_i~:~p_i' = |U_i \cap A| \leq p_i, s-p_i\leq q}} \big[\sum_{i \leq t} Q_{(\chi_{p_i},p_i')}(n,p_i,s-p_i)\big]\Big) \\
\leq \cO\Big(\Delta^*_{(\chi,p')}(n,p,q) \cdot n^{\cO(1)} + |\mathscr{P}_{t}^n|\cdot |{\cal Z}_{s,t}^p| \cdot 
\max_{\substack{(p_1,\ldots,p_t) \in {\cal Z}_{s,t}^p\\ p_1'\leq p_1,\cdots,p_t'\leq p_t~\mbox{\scriptsize such that}\\p_1'+\cdots +p_t'=p',\forall i : s-p_i\leq q}} \big( \sum_{i \leq t} Q_{(\chi_{p_i},p_i')}(n,p_i,s-p_i)\big)  \Big) \\
\leq \cO\Big(\Delta^*_{(\chi,p')}(n,p,q) \cdot n^{\cO(1)} + |\mathscr{P}_{t}^n|\cdot |{\cal Z}_{s,t}^p| \cdot t \cdot 
\max_{\substack{(p_1,\ldots,p_t) \in {\cal Z}_{s,t}^p\\ p_1'\leq p_1,\cdots,p_t'\leq p_t~\mbox{\scriptsize such that}\\p_1'+\cdots +p_t'=p',\forall i : s-p_i\leq q}} \big(Q_{(\chi_{p_i},p_i')}(n,p_i,s-p_i)\big)  \Big) \\
\leq \cO\Big(\Delta^*_{(\chi,p')}(n,p,q) \cdot n^{\cO(1)} + |\mathscr{P}_{t}^n|\cdot |{\cal Z}_{s,t}^p| \cdot t\cdot 
\big( \max_{\substack{ \hat{p}'\leq \hat{p}\leq s \\ \hat{p}-\hat{p}'\leq p-p' \\ s-\hat{p}\leq q }} Q_{(\chi_{\hat{p}},\hat{p}')}(n,\hat{p},s-\hat{p})\big)  \Big) \\
\end{align*} 
For any $(p_1,\ldots,p_t) \in {\cal Z}_{s,t}^p$ and $p_1'\leq p_1,\ldots, p_t'\leq p_t$ such that $\sum_{i=1}^t p_i'=p'$, we have that 
$\sum_{i=1}^t p_i-p_i'=p-p'$ and so $p_i-p_i'\leq p-p'$ for all $i$. This shows the correctness of the last inequality in the above query time analysis.

By doing similar analysis, we get required bound for $Q'_{(\chi',q')}$.  
We now need to argue that $({\cal F}, \chi,\chi')$ is in fact a   $n$-$p$-${q}$-separating collection. For any 
$r$, consider pairwise disjoint sets $A_1\in {U \choose b_1},\ldots,A_r\in {U \choose b_r}$ and 
$B \in {U \choose {q}}$ such that $b_1+\cdots+b_r=p$. Let $A=A_1\cup\cdots\cup A_r$. There exists a 
consecutive partition $\{U_1, \ldots, U_t\} \in \mathscr{P}_t^n$ of $U$ such that for every $i \leq t$ we have that 
$|(A \cup B) \cap U_i| \leq \lceil \frac{p+{q}}{t} \rceil = s$. For each $i \leq t$ set $p_i =  |A \cap U_i|$ and 
$q_i = |B \cap U_i| = s - p_i$. Note that $p_i\leq p$ and $q_i\leq q$ for all $i$. For every $i \leq t$ the tuple $({\cal F}_{p_i}, {\chi}_{p_i},\chi'_{p_i})$ form 
a $n$-$p_i$-$q_i$-separating collection. Hence there exists a 
$ F_i \in \chi_{p_i}(A_1\cap U_i)\cap\cdots\cap\chi_{p_i}(A_r\cap U_i) \cap \chi'_{p_i}(B\cap U_i)$ 
because $|A_1\cap U_i|+\cdots+|A_r\cap U_i|=p_i$, $|B\cap U_i|=q_i$ and 
$({\cal F}_{p_i}, {\chi}_{p_i},\chi'_{p_i})$ is a $n$-$p_i$-$q_i$-separating collection. That is 
$F_i \in \chi_{p_i}(A_j\cap U_i)$ for all $j\leq r$ and $F_i \in \chi'_{p_i}(B\cap U_i)$. Let 
$F=\bigcup_{i\leq t}F_i\cap U_i$. By construction of $\chi$ and $\chi'$, $F \in \chi(A_j)$ for all $j\leq r$  and 
$F \in \chi'(B)$. Hence $F \in \chi(A_1)\cap\cdots\cap \chi(A_r)\cap\chi'(B)$. This completes the proof
\end{proof}

Now we are ready to prove Lemma~\ref{lem:twin_sep_coll_construction}. We restate the lemma for easiness of presentation. 

\medskip 


{{\bf Lemma~\ref{lem:twin_sep_coll_construction}} \em 
Given   $0<x<1$, there is a construction of  $n$-$p$-$q$- separating collection with the following parameters
\begin{itemize} 
\item size: $\zeta(n,p,q) \leq 2^{\cO(\frac{p+q}{\log\log(p+q)})}\cdot \frac{1}{x^p(1-x)^q}\cdot (p+q)^{O(1)} \cdot \log n$
\item initialization time: $\tau_I(n,p,q) \leq  2^{O(\frac{p+q}{\log\log(p+q)})}\cdot \frac{1}{x^p(1-x)^q}\cdot (p+q)^{O(1)} \cdot n\log n$
\item $(\chi,p')$-degree: $\Delta_{(\chi,p')}(n,p,q) \leq  2^{O(\frac{p+q}{\log\log(p+q)})}\cdot \frac{1}{x^{p-p'}(1-x)^q}\cdot (p+q)^{O(1)} \cdot \log n$
\item $(\chi,p')$-query time: $Q_{(\chi,p')}(n,p,q) \leq  2^{O(\frac{p+q}{\log\log(p+q)})}\cdot \frac{1}{x^{p-p'}(1-x)^q}\cdot (p+q)^{O(1)} \cdot \log n$
\item $(\chi',q')$-degree: $\Delta_{(\chi',q')}(n,p,q) \leq  2^{O(\frac{p+q}{\log\log(p+q)})}\cdot \frac{1}{x^{p}(1-x)^{q-q'}}\cdot (p+q)^{O(1)} \cdot \log n$
\item $(\chi',q')$-query time: $Q_{(\chi',q')}(n,p,q) \leq  2^{O(\frac{p+q}{\log\log(p+q)})}\cdot \frac{1}{x^{p}(1-x)^{q-q'}}\cdot (p+q)^{O(1)} \cdot \log n$
\end{itemize}}


\begin{proof}
We first explain a brute force construction of $n$-$p$-$q$-separating collection when  the value of 
$x$ is close to $0$ or close to $1$. These are discussed  in Cases $1$ and $2$  and  the result for all other values of $x$ is explained in Case $3$. Let $U$ be the universe. 

\medskip 

\noindent
\textbf{Case 1: $x\leq \frac{1}{n}$.}
In this case the algorithm will output all subset of size $p$ of the universe as the family ${\cal F}$ of sets in the $n$-$p$-$q$- separating collection. That is 
${\cal F}=\{F\subseteq U~|~|F|=p\}$. We define $\chi$ and $\chi'$ as follows. For any $A\in \bigcup_{p'\leq p} {U \choose p'}$, $\chi(A)=\{F\in {\cal F}~|~A\subseteq F\}$. 
For any $B\in \bigcup_{q'\leq q}{U \choose q'}$, $\chi'(B)=\{F\in {\cal F}~|~B\cap F=\emptyset\}$. It is easy to see that $({\cal F},\chi,\chi')$ is a $n$-$p$-$q$- separating
collection. 
Note that $|{\cal F}|={n \choose p}\leq n^p$. Since $n\leq \frac{1}{x}$, the size of the 
$n$-$p$-$q$- separating collection  is upperbound by the claimed bound.
Since we can list all the elements in ${\cal F}$ in $n^{p}$ time, the initialization time is upper bounded by 
the claimed bound. For any $A\subseteq U$, $|A|=p'$, the cardinality of $\chi(A)$ is exactly equal to ${n \choose p-p'}$ which is upper bounded by $\frac{1}{x^{p-p'}}$. Thus the
$(\chi,p')$-degree  and $(\chi,p')$-query time is bounded by the claimed bound.
For any $B\subseteq U$, $|B|=q'$, the cardinality of $\chi'(B)$ is at most $|{\cal F}|$, which is upper bounded by $\frac{1}{x^{p}}$. Thus the
$(\chi',q')$-degree  and $(\chi',q')$-query time is bounded by the claimed bound.

\medskip 

\noindent
\textbf{Case 2: $1-x\leq \frac{1}{n}$.}
In this case the algorithm will output all subset of size $n-q$ of the universe as the family ${\cal F}$ of sets in the $n$-$p$-$q$- separating collection. That is 
${\cal F}=\{F\subseteq U~|~|F|=n-q\}$. We define $\chi$ and $\chi'$ as follows. For any $A\in \bigcup_{p'\leq p}{U \choose p'}$, $\chi(A)=\{F\in {\cal F}~|~A\subseteq F\}$. 
For any $B\in \bigcup_{q'\leq q}{U \choose q'}$, $\chi'(B)=\{F\in {\cal F}~|~B\cap F=\emptyset\}$. It is easy to see that $({\cal F},\chi,\chi')$ is a $n$-$p$-$q$- separating
collection. 
Note that $|{\cal F}|={n \choose n-q}\leq n^q$. Since $n\leq \frac{1}{1-x}$, the size of the 
$n$-$p$-$q$- separating collection  is upperbound by the claimed bound.
Since we can list all the elements in ${\cal F}$ in $n^{q}$ time, the initialization time is upper bounded by 
the claimed bound. For any $A\subseteq U$, $|A|=p'$, the cardinality of $\chi(A)$ is 
is at most $|{\cal F}|$ which is upper bounded by $\frac{1}{(1-x)^{q}}$. Thus the
$(\chi,p')$-degree  and $(\chi,p')$-query time is bounded by the claimed bound.
For any $B\subseteq U$, $|B|=q'$, the cardinality of $\chi'(B)$ is exactly equal to ${n \choose q-q'}$, which is upper bounded by $\frac{1}{(1-x)^{q-q'}}$. Thus the
$(\chi',q')$-degree  and $(\chi',q')$-query time is bounded by the claimed bound.

\medskip 

\noindent
\textbf{Case 3: $x,1-x>\frac{1}{n}$.}
The structure of the proof in this case is as follows. We first create a collection using 
Lemma~\ref{lem:twin_sep_coll_brute_force}. 
Then we apply Lemma~\ref{lem:twinreduceUniverse} and obtain another construction. From here onwards we keep applying Lemma~\ref{lem:splitSolution} and Lemma~\ref{lem:twinreduceUniverse} in phases until we achieve the 
required bounds on size, degree, query and intializitaion time. 

We first apply Lemma~\ref{lem:twin_sep_coll_brute_force} and get a construction of $n$-$p$-$q$-separating collections with the following parameters.
\begin{itemize}\setlength\itemsep{-.7mm}
\item size, $\zeta^1(n,p,q) =\cO \left(\frac{1}{x^{p}(1-x)^q} \cdot (p^2+q^2+1)\log n \right)$, 
\item initialization time, $\tau_I^1(n,p,q) = \cO({2^n \choose \zeta(n,p,q)} \cdot \frac{1}{x^p(1-x)^q} \cdot n^{\cO(p+q)})$,
\item $(\chi_1,p')$-degree for $p'\leq p$, $\Delta^1_{(\chi_1,p')}(n,p,q) = \cO\left(\frac{1}{x^{p-p'}}\cdot\frac{(p^2+q^2+1)}{(1-x)^q} \cdot \log n\right)$
\item $(\chi_1,p')$-query time ${Q^1_{(\chi_1,p')}}(n,p,q) = \cO(\frac{1}{x^{p}(1-x)^q} \cdot n^{\cO(1)})=\cO(2^n n^{\cO(1)})$
\item $(\chi'_1,q')$-degree for $q'\leq q$, $\Delta^1_{(\chi'_1,q')}(n,p,q)=\cO\left(\frac{1}{x^{p}(1-x)^{q-q'}}\cdot(p^2+q^2+1) \cdot \log n\right)$
\item $(\chi'_1,q')$-query time, ${Q^1_{(\chi_1',q')}}(n,p,q)=\cO(\frac{1}{x^{p}(1-x)^q} \cdot n^{\cO(1)})=\cO(2^n n^{\cO(1)})$
\end{itemize}
We apply Lemma~\ref{lem:twinreduceUniverse} to this construction to get a new construction with the following parameters.
\begin{itemize} 
\item size, $\zeta^2(n,p,q)=\cO\left(\frac{1}{x^{p}(1-x)^q} \cdot  (p+q)^{\cO(1)} \cdot \log n\right)$ 
\item initialization time, 
\begin{eqnarray*}
\tau_I^2(n,p,q) &=& \cO\left(\tau_I^1\left((p+q)^2,p,q\right) + \zeta^1\left((p+q)^2,p,q\right) \cdot (p+q)^{\cO(1)} \cdot n \log n\right)\\
&=& \cO\left(
\frac{2^{2^{(p+q)^2}}}{x^p(1-x)^q} \cdot (p+q)^{\cO(p+q)} + \left(\frac{1}{x^{p}(1-x)^q} \cdot (p+q)^{\cO(1)} \cdot n \log n \right)\right)\\
&=&\cO\left( \frac{(p+q)^{\cO(p+q)}}{x^{p}(1-x)^q}\left({2^{2^{(p+q)^2}}}+n \log n\right)\right)
\end{eqnarray*}
\item $(\chi_2,p')$-degree, $\Delta_{(\chi_2,p')}^2(n,p,q) =\cO\left( \frac{1}{x^{p-p'}{(1-x)^q}} \cdot  (p+q)^{\cO(1)} \cdot \log n\right)$
\item $(\chi_2,p')$-query time, $Q_{(\chi_2,p')}^2(n,p,q)=\cO\left(\left(2^{(p+q)^2} + \frac{1}{x^{p-p'}{(1-x)^q}}\right)   (p+q)^{\cO(1)} \cdot \log n\right)$
\item $(\chi_2',q')$-degree, $\Delta_{(\chi_2',q')}^2(n,p,q) =\cO\left( \frac{1}{x^{p}(1-x)^{q-q'}} \cdot  (p+q)^{\cO(1)} \cdot \log n\right)$
\item $(\chi_2,q')$-query time, $Q_{(\chi'_2,q')}^2(n,p,q)=\cO\left(\left(2^{(p+q)^2} + \frac{1}{x^{p}(1-x)^{q-q'}}\right)   (p+q)^{\cO(1)} \cdot \log n\right)$
\end{itemize}
We apply Lemma~\ref{lem:splitSolution} to this construction. Recall that in Lemma~\ref{lem:splitSolution} 
we set $s=\lfloor(\log (p+q))^2 \rfloor$ and $t = \lceil \frac{p+q}{s} \rceil$.
\begin{eqnarray*}
 \zeta^3(n,p,q) &\leq& |\mathscr{P}_t^{n}| \cdot 
\sum_{(p_1,\ldots,p_t) \in {\cal Z}_{s,t}^{p}} \prod_{i \leq t} \zeta^2(n,p_i,s-p_i)\\
&\leq& n^{\cO(t)}\cdot |{\cal Z}_{s,t}^{p}| \cdot \max_{(p_1,\ldots,p_t) \in {\cal Z}_{s,t}^{p}} \prod_{i \leq t} \zeta^2(n,p_i,s-p_i)\\
&\leq& n^{\cO(t)}\cdot (p+q)^{\cO(t)}\cdot \frac{1}{x^{p}(1-x)^{q+s}}\cdot s^{\cO(t)}\cdot (\log n)^{\cO(t)}\\
&\leq& n^{\cO(\frac{p+q}{\log^2(p+q)})}\cdot \frac{1}{x^{p}(1-x)^{q}} \qquad\qquad\quad \left(\mbox{Because} \left(\frac{1}{1-x}\right)^s\leq n^s\leq n^{\cO(t)}\right)\\
\end{eqnarray*}
\begin{eqnarray*}
 \tau_I^3(n,p,q) &=&\cO\left(\left(\sum_{\substack{\hat{p} \leq s,p \\ s-\hat{p}\leq q}} \tau_I^2(n,\hat{p},s-\hat{p})\right) + \zeta^3(n,p,q) \cdot n^{\cO(1)}\right)\\
&=&\cO\left(\left(\sum_{\substack{\hat{p} \leq s,p \\ s-\hat{p}\leq q}} \frac{s^{\cO(s)}}{x^{\hat{p}}(1-x)^{s-\hat{p}}}\left({2^{2^{s^2}}}+n \log n\right)\right) + \zeta^3(n,p,q)
\cdot n^{\cO(1)}\right)\\
&=& \cO\left(\frac{(\log(p+q))^{\cO(\log^2(p+q))}}{x^{p}(1-x)^{q}}\left({2^{2^{\log^4(p+q)}}}+n \log n\right) + n^{\cO(\frac{p+q}{\log^2(p+q)})}\cdot \frac{1}{x^{p}(1-x)^{q}}\right)\\
\end{eqnarray*}
\begin{eqnarray*}
\Delta_{(\chi_3,p')}^3(n,p,q) &\leq& \Delta_{(\chi_3,p')}^{*3}(n,p,q)  \\
&=&|\mathscr{P}_t^n|\cdot |{\cal Z}_{s,t}^{p}| \cdot 
\max_{\substack{(p_1,\ldots,p_t) \in {\cal Z}_{s,t}^{p} \\ p_1'\leq p_1,\ldots, p_t'\leq p_t \\ p_1'+\ldots+p_t'=p'}} \prod_{i \leq t}  \Delta_{(\chi,p')}^2(n,p_i,s-p_i)\\
&\leq& n^{\cO(t)}\cdot (p+q)^{\cO(t)} \cdot\frac{1}{x^{p-p'}(1-x)^{q+s}}\cdot s^{\cO(t)}\cdot (\log n)^{\cO(t)}\\
&\leq& n^{\cO(\frac{p+q}{\log^2(p+q)})}\cdot \frac{1}{x^{p-p'}(1-x)^{q}}  \qquad\qquad\quad \left(\mbox{Because} \left(\frac{1}{1-x}\right)^s\in n^{\cO(t)}\right)\\
\Delta_{(\chi_3',q')}^{3}(n,p,q) &\leq& \Delta_{(\chi_3',q')}^{*3}(n,p,q) \\ 
&=& |\mathscr{P}_t^n|\cdot |{\cal Z}_{s,t}^{p}| \cdot 
\max_{\substack{(p_1,\ldots,p_t) \in {\cal Z}_{s,t}^{p} \\ q_1'\leq s-p_1,\ldots, q_t'\leq s-q_t \\ q_1'+\ldots+q_t'=q'}} \prod_{i \leq t}  \Delta_{(\chi',q_i')}^2(n,p_i,s-p_i)\\
&\leq& n^{\cO(t)}\cdot (p+q)^{\cO(t)}\cdot \frac{1}{x^{p}(1-x)^{q+s-q'}}\cdot s^{\cO(t)}\cdot (\log n)^{\cO(t)}\\
&\leq& n^{\cO(\frac{p+q}{\log^2(p+q)})}\cdot \frac{1}{x^{p}(1-x)^{q-q'}} \qquad\qquad\quad \left(\mbox{Because} \left(\frac{1}{1-x}\right)^s\in n^{\cO(t)}\right)\\
\end{eqnarray*}
\begin{eqnarray*}
Q_{(\chi_3,p')}^3(n,p,q) &\leq& \cO\left(\Delta_{(\chi_3,p')}^{*3}(n,p,q) \cdot n^{\cO(1)} + 
|\mathscr{P}_{t}^{n}|\cdot |{\cal Z}_{s,t}^{p}| \cdot t \cdot 
\max_{\substack{ \hat{p}'\leq \hat{p}\leq s \\ \hat{p}-\hat{p}'\leq p-p' \\ s-\hat{p}\leq q }} Q_{(\chi_2,\hat{p}')}^2(n,\hat{p},s-\hat{p}) \right)\\
 &\leq& \cO\left(\Delta_{(\chi_3,p')}^{*3}(n,p,q) \cdot n^{\cO(1)} + n^{\cO(t)} \cdot 
\max_{\substack{ \hat{p}'\leq \hat{p}\leq s \\ \hat{p}-\hat{p}'\leq p-p' \\ s-\hat{p}\leq q }} \left(2^{s^2}+\frac{1}{x^{\hat{p}-\hat{p}'}(1-x)^{s-\hat{p}}}\right)s^{\cO(1)} \log n
\right)\\
 &\leq& \cO\left(\frac{n^{\cO(\frac{p+q}{\log^2(p+q)})}}{x^{p-p'}(1-x)^{q}} + n^{\cO(t)} \cdot s^{\cO(1)}\cdot \log n
\left(2^{s^2}+\frac{1}{x^{p-p'}(1-x)^{q}}\right)  \right)\\
 &\leq& \cO\left(\frac{n^{\cO(\frac{p+q}{\log^2(p+q)})}}{x^{p-p'}(1-x)^{q}}  \right)\\
\end{eqnarray*}
Similar way we can bound $Q_{(\chi_3',q')}^3$ as,
\begin{eqnarray*}
 Q_{(\chi_3',q')}^3(n,p,q) &\leq& \cO\left(\frac{n^{\cO(\frac{p+q}{\log^2(p+q)})}}{x^{p}(1-x)^{q-q'}}  \right)
\end{eqnarray*}

%
%
We apply Lemma~\ref{lem:twinreduceUniverse} to this construction to get a new construction with the following parameters.
\begin{itemize} 
\item  size, 
$\zeta^4(n,p,q) \leq 2^{\cO(\frac{p+q}{\log(p+q)})}\cdot \frac{1}{x^{p}(1-x)^{q}} \cdot  (p+q)^{\cO(1)} \cdot \log n$,
 \item  initialization time, 
\begin{eqnarray*}
\tau_I^4(n,p,q) &\leq& \cO\left(\tau_I^3\left((p+q)^2,p,q\right) + \zeta^3\left((p+q)^2,p,q\right) \cdot (p+q)^{\cO(1)} \cdot n \log n\right)\\
&\leq& 2^{2^{\log^4(p+q)}}\cdot\frac{(\log(p+q))^{\cO(\log^2(p+q))}}{x^p(1-x)^q} + \frac{2^{\cO(\frac{p+q}{\log (p+q)})}}{x^p(1-x)^q} \cdot (p+q)^{\cO(1)}n\log n\\
\end{eqnarray*}
\item $(\chi_4,p')$-degree, 
\begin{eqnarray*}
\Delta_{(\chi_4,p')}^4(n,p,q) &\leq& \Delta_{(\chi_3,p')}^3\left((p+q)^2,p,q\right) \cdot  (p+q)^{\cO(1)} \cdot \log n \\
&\leq& \frac{2^{\cO(\frac{p+q}{\log (p+q)})}}{x^{p-p'}(1-x)^q} \cdot  (p+q)^{\cO(1)} \cdot \log n
\end{eqnarray*}
\item $(\chi_4',q')$-degree, 
\begin{eqnarray*}
\Delta_{(\chi_4',q')}^4(n,p,q) &\leq& \Delta_{(\chi_3',q')}^3\left((p+q)^2,p,q\right) \cdot  (p+q)^{\cO(1)} \cdot \log n \\
&\leq& \frac{2^{\cO(\frac{p+q}{\log (p+q)})}}{x^{p}(1-x)^{q-q'}} \cdot  (p+q)^{\cO(1)} \cdot \log n
\end{eqnarray*}
\item $(\chi_4,p')$-query time, 
\begin{eqnarray*}
Q_{(\chi_4,p')}^4(n,p,q) &\leq& \cO\left(\left( Q_{(\chi_3,p')}^3\left((p+q)^2,p,q\right) + \Delta_{(\chi_3,p')}^3\left((p+q)^2,p,q\right) \right) \cdot (p+q)^{\cO(1)} \cdot \log
n\right)\\
&\leq& \frac{2^{\cO(\frac{p+q}{\log (p+q)})}}{x^{p-p'}(1-x)^q}\cdot (p+q)^{\cO(1)}\log n
\end{eqnarray*}
\item $(\chi_4',q')$-query time, 
\begin{eqnarray*}
Q_{(\chi_4',q')}^4(n,p,q) 
&\leq& \frac{2^{\cO(\frac{p+q}{\log (p+q)})}}{x^{p}(1-x)^{q-q'}}\cdot (p+q)^{\cO(1)}\log n 
\end{eqnarray*}
\end{itemize}
We apply Lemma~\ref{lem:splitSolution} to this construction by setting $s=\lfloor(\log (p+q))^2 \rfloor$ and 
$t = \lceil \frac{p+q}{s} \rceil$.
\begin{itemize}
\item size,
\begin{eqnarray*}
\zeta^5(n,p,q) &\leq& |\mathscr{P}_t^n| \cdot 
\sum_{(p_1,\ldots,p_t) \in {\cal Z}_{s,t}^{p} } \prod_{i \leq t} \zeta^4(n,p_i,s-p_i)\\
&\leq& n^{\cO(t)}\cdot(p+q)^{\cO(t)}\cdot s^{\cO(t)} \cdot 2^{\cO(\frac{st}{\log s})}\cdot (\log n)^{\cO(t)}\cdot \frac{1}{x^p(1-x)^{q+s}}\\ 
&\leq& n^{\cO(\frac{p+q}{\log^2(p+q)})} \cdot 2^{\cO(\frac{p+q}{\log\log(p+q)})} \frac{1}{x^p(1-x)^q} \qquad\quad \left(\mbox{Because } \left(\frac{1}{1-x}\right)^s\in
n^{\cO(t)}\right)
\end{eqnarray*}
\item initialization time,
\begin{eqnarray*}
\tau_I^5(n,p,q) &\leq& \cO\left(\left(\sum_{\substack{\hat{p} \leq s,p\\ s-\hat{p}\leq q}} \tau_I^4(n,\hat{p},s-\hat{p})\right) + \zeta^5(n,p,q) \cdot n^{\cO(1)}\right)\\
&\leq&\cO\left(s\frac{2^{2^{\log^4s}}\cdot (\log s)^{\cO(\log^2s)}}{x^p(1-x)^q}+ \frac{2^{\cO(\frac{s}{\log s})}}{x^p(1-x)^q}\cdot n\log n + n^{\cO(\frac{p+q}{\log^2(p+q)})} \cdot  \frac{2^{\cO(\frac{p+q}{\log\log(p+q)})}}{x^p(1-x)^q}
\right)\\
&\leq&\cO\left(s\frac{2^{2^{\log^4s}}\cdot (\log s)^{\cO(\log^2s)}}{x^p(1-x)^q} + n^{\cO(\frac{p+q}{\log^2(p+q)})} \cdot  \frac{2^{\cO(\frac{p+q}{\log\log(p+q)})}}{x^p(1-x)^q}
\right)\\
&\leq&\cO\left(\frac{2^{2^{\log^4s}}\cdot (s)^{\cO(s)}}{x^p(1-x)^q} + n^{\cO(\frac{p+q}{\log^2(p+q)})} \cdot  \frac{2^{\cO(\frac{p+q}{\log\log(p+q)})}}{x^p(1-x)^q}
\right)\\
&\leq&\cO\left(\frac{2^{2^{(2\log\log(p+q))^4}}\cdot (\log(p+q))^{\cO((\log(p+q))^2)}} {x^p(1-x)^q} 
+ n^{\cO(\frac{p+q}{\log^2(p+q)})} \cdot  \frac{2^{\cO(\frac{p+q}{\log\log(p+q)})}}{x^p(1-x)^q}
\right)\\
&\leq&\cO\left(n^{\cO(\frac{p+q}{\log^2(p+q)})} \cdot  \frac{2^{\cO(\frac{p+q}{\log\log(p+q)})}}{x^p(1-x)^q} \right) \\
&&\qquad \left(\mbox{Because }  2^{2^{(2\log \log (p+q))^4}} , (\log (p+q) )^{\cO(\log^2(p+q))} \leq 2^{\cO(\frac{p+q}{\log\log(p+q)})}\right)
\end{eqnarray*}
 \item $(\chi_5,p')$-degree, 
\begin{eqnarray*}
\Delta_{(\chi_5,p')}^5(n,p,q) &\leq& \Delta_{(\chi_5,p')}^{*5}(n,p,q)\\ 
&=& |\mathscr{P}_t^n|\cdot |{\cal Z}_{s,t}^{p}| \cdot  
\max_{\substack{(p_1,\ldots,p_t) \in {\cal Z}_{s,t}^{p} \\ p_1'\leq p_1,\ldots, p_t'\leq p_t \\ p_1'+\ldots+p_t'=p' }} \prod_{i \leq t}  \Delta_{(\chi_4,p_i')}^4(n,p_i,s-p_i)\\
&\leq& n^{\cO(t)}\cdot (p+q)^{\cO(t)} \cdot \frac{2^{\cO(\frac{st}{\log s})}}{x^{p-p'}(1-x)^{q+s}}\cdot s^{\cO(t)}\cdot (\log n)^{\cO(t)}\\
&\leq& n^{\cO(\frac{p+q}{\log^2(p+q)})} \cdot 2^{\cO(\frac{p+q}{\log\log(p+q)})} \cdot \frac{1}{x^{p-p'}(1-x)^q} \qquad\quad \left(\mbox{Because } \left(\frac{1}{1-x}\right)^s\in
n^{\cO(t)}\right)
\end{eqnarray*}
\item $(\chi_5',q')$-degree, 
\begin{eqnarray*}
\Delta_{(\chi_5',q')}^5(n,p,q) &\leq& \Delta_{(\chi_5',q')}^{*5}(n,p,q)  \\ 
&\leq& n^{\cO(\frac{p+q}{\log^2(p+q)})} \cdot 2^{\cO(\frac{p+q}{\log\log(p+q)})} \cdot \frac{1}{x^{p}(1-x)^{q-q'}}
\end{eqnarray*}
\item $(\chi_5,p')$-query time, 
\begin{eqnarray*}
Q_{(\chi_5,p')}^5(n,p,q) &\leq& \cO\left(\Delta_{(\chi_5,p')}^{*5}(n,p,q) \cdot n^{\cO(1)} + |\mathscr{P}_{t}^{n}|\cdot |{\cal Z}_{s,t}^{p}| \cdot   
\max_{\substack{ \hat{p}'\leq \hat{p}\leq s \\ \hat{p}-\hat{p}'\leq p-p' \\ s-\hat{p}\leq q }}
Q_{(\chi_4,\hat{p}')}^4(n,\hat{p},s-\hat{p}) \right)\\
&\leq& n^{\cO(\frac{p+q}{\log^2(p+q)})} \cdot 2^{\cO(\frac{p+q}{\log\log(p+q)})} \cdot \frac{1}{x^{p-p'}(1-x)^q}
\end{eqnarray*}
\item $(\chi_5',q')$-query time,
\begin{eqnarray*}
 Q_{(\chi_5',q')}^5(n,p,q) 
&\leq& n^{\cO(\frac{p+q}{\log^2(p+q)})} \cdot 2^{\cO(\frac{p+q}{\log\log(p+q)})} \cdot \frac{1}{x^{p}(1-x)^{q-q'}}
\end{eqnarray*}
\end{itemize}
We apply Lemma~\ref{lem:twinreduceUniverse} to this construction to get a new construction with the following parameters.
\begin{itemize} 
\item size, 
\begin{eqnarray*}
\zeta(n,p,q) &\leq& \zeta^5\left((p+q)^2,p,q\right) \cdot  (p+q)^{\cO(1)} \cdot \log n\\
&\leq & 2^{\cO(\frac{p+q}{\log\log(p+q)})}\cdot \frac{1}{x^p(1-x)^q}\cdot (p+q)^{\cO(1)} \log n
\end{eqnarray*}
\item initialization time, 
\begin{eqnarray*}
\tau_I(n,p,q) &\leq& \cO\left(\tau_I^5\left((p+q)^2,p,q\right) + \zeta^5\left((p+q)^2,p,q\right) \cdot (p+q)^{\cO(1)} \cdot n \log n\right)\\
&=&\cO\left(2^{\cO(\frac{p+q}{\log\log(p+q)})}\cdot \frac{1}{x^p(1-x)^q}\cdot (p+q)^{\cO(1)} n\log n \right)
\end{eqnarray*}
\item $(\chi,p')$-degree, 
\begin{eqnarray*}
\Delta_{(\chi,p')}(n,p,q) &\leq& \Delta_{(\chi_5,p')}^5\left((p+q)^2,p,q\right) \cdot  (p+q)^{\cO(1)} \cdot \log n \\
&\leq& \cO\left( 2^{\cO(\frac{p+q}{\log\log(p+q)})} \cdot \frac{1}{x^{p-p'}(1-x)^q} \cdot (p+q)^{\cO(1)} \cdot \log n\right)
\end{eqnarray*}
\item $(\chi,p')$-query time,
\begin{eqnarray*}
Q_{(\chi,p')}(n,p,q) &\leq& \cO\left(\left(Q_{(\chi_5,p')}^5\left((p+q)^2,p,q\right) + \Delta_{(\chi_5,p')}^5\left((p+q)^2,p,q\right) \right) \cdot (p+q)^{\cO(1)} \cdot \log
n\right)\\
&\leq& \cO\left( 2^{\cO(\frac{p+q}{\log\log(p+q)})} \cdot \frac{1}{x^{p-p'}(1-x)^q} \cdot (p+q)^{\cO(1)} \cdot \log n\right)
\end{eqnarray*}
\item $(\chi',q')$-degree, 
\begin{eqnarray*}
\Delta_{(\chi',q')}(n,p,q)&=&\Delta_{(\chi_5',q')}^5\left((p+q)^2,p,q\right) \cdot  (p+q)^{\cO(1)} \cdot \log n\\
&\leq& \cO\left( 2^{\cO(\frac{p+q}{\log\log(p+q)})} \cdot \frac{1}{x^{p}(1-x)^{q-q'}} \cdot (p+q)^{\cO(1)} \cdot \log n\right)
\end{eqnarray*}
\item $(\chi',q')$-query time, 
\begin{eqnarray*}
Q_{(\chi',q')}(n,p,q)&=&\cO\left(\left(Q_{(\chi_5',q')}^5\left((p+q)^2,p,q\right) + \Delta_{(\chi_5',q')}^5\left((p+q)^2,p,q\right) \right) \cdot (p+q)^{\cO(1)} \cdot \log
n\right)\\
&\leq& \cO\left( 2^{\cO(\frac{p+q}{\log\log(p+q)})} \cdot \frac{1}{x^{p}(1-x)^{q-q'}} \cdot (p+q)^{\cO(1)} \cdot \log n\right)
\end{eqnarray*}
\end{itemize}
The final construction satisfies all the claimed bounds. This concludes the proof. 
\end{proof}

\begin{lemma}
\label{thm:repset uniform general}
There is an algorithm that given a $p$-family ${\cal A}$ of sets  over a universe $U$ of size $n$,  an integer $q$, a $0<x<1$, 
and a non-negative weight function    \awf{} with maximum value at most $W$,
computes in time 
\[ \cO( {x^{-p}(1-x)^{-q}} \cdot 2^{o(p+q)}\cdot n\log n+|{\cal A}|\cdot \log |{\cal A}| \cdot  \log{W}  + |{\cal A}| \cdot  (1-x)^{-q}\cdot 2^{o(p+q)} \cdot \log n ) \] 
a subfamily 
$\widehat{\cal A}\subseteq \cal A$ such that 
$|\hat{\cal A}| \leq  {x^{-p}(1-x)^{-q}} \cdot 2^{o(p+q)}  \cdot \log n$
and 
 \minrep{A}{q} (\maxrep{A}{q}).
\end{lemma}

\begin{proof}
The algorithm first checks whether  $|{\cal A}| \leq {x^{-p}(1-x)^{-q}} \cdot 2^{o(p+q)}  \cdot \log n$. If yes then it outputs ${\cal A}$ (as  $\widehat{\cal A}$) and halts. So we assume that $|{\cal A}| > {x^{-p}(1-x)^{-q}} \cdot 2^{o(p+q)}  \cdot \log n$. 
The algorithm starts by constructing a generalized $n$-$p$-$q$-separating collection $({\cal F}, {\chi}, {\chi}')$ as guaranteed by Lemma~\ref{lem:twin_sep_coll_construction}. If $|{\cal A}| \leq |{\cal F}|$ the algorithm outputs ${\cal A}$ and halts. Otherwise it builds the set $\hat{\cal A}$ as follows. Initially $\hat{\cal A}$ is equal to $\emptyset$ and all sets in ${\cal F}$ are marked as unused. 
Now we sort the sets in ${\cal A}$ in the increasing order of weights, given by  \awf{}. The algorithm goes through every $A \in {\cal A}$ in the sorted order and queries the separating collection to get the set $\chi(A)$. It then looks for a set $F \in \chi(A)$ that is not yet marked as used. The first time such a set $F$ is found the algorithm marks $F$ as used, inserts $A$ into $\hat{\cal A}$ and proceeds to the next set in ${\cal A}$. If no such set $F$ is found the algorithm proceeds to the next set in ${\cal A}$ without inserting $A$ into $\hat{\cal A}$.

The size of $\hat{\cal A}$ is upper bounded by 
$|{\cal F}| \leq {x^{-p}(1-x)^{-q}} \cdot 2^{o(p+q)}  \cdot \log n$
since every time a set is added to $\hat{\cal A}$ an unused set in ${\cal F}$ is marked as used. For the running time analysis, the initialization of $({\cal F}, {\chi})$ takes time 
$ {x^{-p}(1-x)^{-q}}\cdot (p+q)^{\cO(1)} \cdot 2^{o(p+q)}\cdot n\log n$. 
Sorting  ${\cal A}$ takes $\cO(|{\cal A}|\cdot \log |{\cal A}| \cdot  \log{W})$ time. 
For each element $A \in {\cal A}$ the algorithm first queries $\chi(A)$, using time 
$ {(1-x)^{-q}}\cdot 2^{o(p+q)}\cdot (p+q)^{\cO(1)} \cdot \log n$. Then it goes through all sets in $\chi(A)$ and checks whether they have already been marked as used, taking time $ {(1-x)^{-q}}\cdot (p+q)^{\cO(1)} \cdot 2^{o(p+q)}\cdot \log n$. Thus in total, the running time for these steps  is bounded by 
$\cO(|{\cal A}| \cdot  (1-x)^{-q}\cdot 2^{o(p+q)} \cdot \log n +|{\cal A}|\cdot \log |{\cal A}| \cdot  \log{W})$. Adding the initialization time to this gives the claimed running time.

Finally we need to argue that  \minrep{A}{q}. Consider any set $A \in {\cal A}$ and $B$ such that $|B|=q$ and $A \cap B = \emptyset$. If $A \in  \hat{\cal A}$ we are done, so assume that $A \notin \hat{\cal A}$. 
Since $({\cal F},\chi,\chi')$ is a  $n$-$p$-$q$-separating collection, 
we have that there exists $F\in \chi(A)\cap \chi'(B)$, i.e, $A \subseteq F$ and $F \cap B = \emptyset$.   
Since  $A \notin  \hat{\cal A}$ we know that $F$ was marked as used when $A$ was considered by the algorithm. When the algorithm marked $F$ as used it also inserted a set $A'$ into $ \hat{\cal A}$, with the property that $F \in \chi(A')$. Thus $A' \subseteq F$ and hence $A' \cap B = \emptyset$. Furthermore, $A'$ was considered before $A$ and thus $w(A')\leq w(A)$. 
But $A' \in  \hat{\cal A}$, completing the proof.
\end{proof}

Next we prove a ``faster version of Lemma~\ref{thm:repset uniform general}'', that speeds up the running time to compute the representative families.

\begin{lemma}
\label{thm:repset uniform general-fasterruntime}
There is an algorithm that given a $p$-family ${\cal A}$ of sets  over a universe $U$ of size $n$,  an integer $q$, a $0<x<1$, 
and a non-negative weight function    \awf{} with maximum value at most $W$,
computes in time 
\[ \cO( (p+q)^{\cO(1)}  n\log n + |{\cal A}|\cdot \log |{\cal A}| \cdot  \log{W}  + |{\cal A}| \cdot  (1-x)^{-q}\cdot 2^{o(p+q)} \cdot \log n ) \] 
a subfamily 
$\widehat{\cal A}\subseteq \cal A$ such that 
$|\hat{\cal A}| \leq  {x^{-p}(1-x)^{-q}} \cdot 2^{o(p+q)}  \cdot \log n$
and 
 \minrep{A}{q} (\maxrep{A}{q}).
\end{lemma}

\begin{proof}
The algorithm first checks whether  $|{\cal A}| \leq {x^{-p}(1-x)^{-q}} \cdot 2^{o(p+q)}  \cdot \log n$. If yes then it outputs ${\cal A}$ (as  $\widehat{\cal A}$) and halts. So we assume that $|{\cal A}| > {x^{-p}(1-x)^{-q}} \cdot 2^{o(p+q)}  \cdot \log n$. 

We start by  constructing a $(p+q)$-perfect family $f_1, \ldots, f_t$ of hash functions from $U$ to 
$[(p+q)^2]$ with $t = \cO((p+q)^{\cO(1)} \cdot \log n)$ in time  $\cO(k^{\cO(1)}n \log n)$ using Proposition~\ref{prop:hashFun}. Now we sort the sets in ${\cal A}$ in the increasing order of weights, given by  \awf{}. 
For every $f_j$, $1\leq j\leq t$,  we construct a  family $\hat{\cal A}_j$ as follows. The algorithm starts by constructing a generalized $[(p+q)^2]$-$p$-$q$-separating collection $({\cal F}_j, {\chi}_j, {\chi}'_j)$ as guaranteed by Lemma~\ref{lem:twin_sep_coll_construction}. 
It builds the set $\hat{\cal A}_j$ as follows. Initially $\hat{\cal A}_j$ is equal to $\emptyset$ and all sets in ${\cal F}$ are marked as unused.  The algorithm goes through every $A \in {\cal A}$ in the sorted order and does as follows. 
\begin{itemize}
\item It first check whether every element in $A$ gets mapped to distinct integers by $f_j$. 
That is, $|\{f_j(a)~|~a\in A\}|=|A|$. If $|\{f_j(a)~|~a\in A\}|<|A|$ 
then the algorithm proceeds to the next set in ${\cal A}$ without inserting $A$ into $\hat{\cal A}$. Else, we move to the next step. 
\item It queries the separating collection to get the set $\chi(A)$. 
It looks for a set $F \in \chi_j(A)$ that is not yet marked as used. The first time such a set $F$ is found the algorithm marks $F$ as used, inserts $A$ into $\hat{\cal A}_j$ and proceeds to the next set in ${\cal A}$. If no such set $F$ is found the algorithm proceeds to the next set in ${\cal A}$ without inserting $A$ into $\hat{\cal A}_j$.
 \end{itemize}
Finally, we return $\hat{\cal A}=\bigcup_{j=1}^t \hat{\cal A}_j$. 

The size of $\hat{\cal A}_j$ is upper bounded by 
$|{\cal F}| \leq {x^{-p}(1-x)^{-q}} \cdot 2^{o(p+q)}  \cdot  \log (p+q)$
since every time a set is added to $\hat{\cal A}$ an unused set in ${\cal F}$ is marked as used. 
Thus, the size of $\hat{\cal A}$ is upper bounded by 
$|{\cal F}| \leq {x^{-p}(1-x)^{-q}} \cdot 2^{o(p+q)}  \cdot \log (p+q) \cdot (p+q)^{\cO(1)} \cdot \log n \leq {x^{-p}(1-x)^{-q}} \cdot 2^{o(p+q)}  \cdot \log n $. The running time analysis follows similar to the one given in 
Lemma~\ref{thm:repset uniform general}. 

Finally we need to argue that  \minrep{A}{q}. Consider any set $A \in {\cal A}$ and $B$ such that $|B|=q$ and $A \cap B = \emptyset$. If $A \in  \hat{\cal A}$ we are done, so assume that $A \notin \hat{\cal A}$. By the properties of $(p+q)$-perfect family $f_1, \ldots, f_t$ of hash functions from $U$ to $[(p+q)^2]$, there exists an integer $j \in \{1,\ldots,t\}$ such that $f_j$ is injective on $A\cup B$. We focus now on the construction of $\hat{\cal A}_j$. 
Since $({\cal F}_j,\chi_j,\chi'_j)$ is a  $[(p+q)^2]$-$p$-$q$-separating collection, 
we have that there exists $F\in \chi_j(A)\cap \chi'_j(B)$, i.e, $A \subseteq F$ and $F \cap B = \emptyset$.   
Since   $A \notin  \hat{\cal A}_j$  (as $A \notin  \hat{\cal A}$) we know that $F$ was marked as used when $A$ was considered by the algorithm. When the algorithm marked $F$ as used it also inserted a set $A'$ into $ \hat{\cal A}$, with the property that $F \in \chi(A')$. Thus $A' \subseteq F$ and hence $A' \cap B = \emptyset$. Furthermore, $A'$ was considered before $A$ and thus $w(A')\leq w(A)$. 
But $A' \in  \hat{\cal A}_j\subseteq \hat{\cal A}$, completing the proof.
\end{proof}

While applying Lemma~\ref{thm:repset uniform general-fasterruntime} we can reduce the 
universe size to at most $|{\cal A}|p+q$. The next lemma formalizes this.
\begin{lemma}
\label{cor:repset uniform general2}
There is an algorithm that given a $p$-family ${\cal A}$ of sets  over a universe $U$ of size $n$,  an integer $q$, 
a $0<x<1$ 
and a non-negative weight function    \awf{} with maximum value at most $W$,
computes in time 
\[ \cO(|{\cal A}|\cdot \log |{\cal A}| \cdot  \log{W}  + |{\cal A}| \cdot  (1-x)^{-q}\cdot 2^{o(p+q)} \cdot \log n ) \] 
a subfamily 
$\widehat{\cal A}\subseteq \cal A$ such that 
$|\widehat{\cal A}| \leq  {x^{-p}(1-x)^{-q}} \cdot 2^{o(p+q)}  \cdot \log |{\cal A}|$
and 
 \minrep{A}{q} (\maxrep{A}{q}).
\end{lemma} 
\begin{proof} 
We first construct a new universe $U'$ as follows. If $n\leq |{\cal A}|p+q$, then we set $U'=U$, otherwise $U'$ will consist of elements from $U$,   
which are part of any set in ${\cal A}$ and $q$ new elements. The universe $U'$ can be constructed in $\cO(|{\cal A}|p+q)$ time. Also note that $|U'|\leq |{\cal A}|p+q$ 
and $|U'|\leq n$. Now we claim that a $q$-representative family $\widehat{\cal A}$ of ${\cal A}$ with respect to the universe $U'$ is  also the required representative family over $U$. 
Suppose $X\in {\cal A}$ and $Y\subseteq U$, $|Y|\leq q$ such that $X\cap Y=\emptyset$. Let $Y'=Y\setminus U'$ and let $Y''$ be an arbitrary subset of size $|Y'|$ of $U'\setminus U$. 
Let $Z=(Y\setminus Y')\cup Y''$. It is easy to see that $|Z|=|Y|$ and $X\cap Z=\emptyset$. By the definition of $q$-representative family, there exists $\widehat{X}\in \widehat{A}$ such that 
$\widehat{X}\cap Z=\emptyset$. Since $Y'\cap \widehat{X}=\emptyset$, we have that $\widehat{X}\cap Y=\emptyset$. 

Thus we apply Lemma~\ref{thm:repset uniform general-fasterruntime} to compute  $q$-representative family $\widehat{\cal A}$ of ${\cal A}$ with respect to the universe $U'$ and output it as the desired family. The 
 claimed running time as well as the size bound on the output representative family follow by substituting the upper bound on $|U'|$ in the bounds coming from Lemma~\ref{thm:repset uniform general-fasterruntime}. 
\end{proof}

Finally, we give our main theorem. 

\begin{theorem}
\label{thm:repset uniform general3}
There is an algorithm that given a $p$-family ${\cal A}$ of sets  over a universe $U$ of size $n$,  an integer $q$, 
a $0<x<1$
and a non-negative weight function    \awf{} with maximum value at most $W$,
computes in time 
\[ \cO(|{\cal A}|\cdot \log |{\cal A}| \cdot  \log{W}  + |{\cal A}| \cdot  (1-x)^{-q}\cdot 2^{o(p+q)} \cdot \log n ) \] 
a subfamily 
$\widehat{\cal A}\subseteq \cal A$ such that 
$|\hat{\cal A}| \leq  {x^{-p}(1-x)^{-q}} \cdot 2^{o(p+q)} $
and 
 \minrep{A}{q} (\maxrep{A}{q}).
\end{theorem} 
\begin{proof}
Let ${\cal A}={\cal A}_1$. 
We compute a sequence of representative families 
 $${\cal A}_2\subseteq_{minrep}^q {\cal A}_1, 
\cdots,{\cal A}_m\subseteq_{minrep}^q {\cal A}_{m-1}$$
 using Corollary~\ref{cor:repset uniform general2}, such that $m$ is the least integer with the property that $|{\cal A}_m|\geq |{\cal A}_{m-1}|/2$. In other words, for all $i<m$ we have that 
$|{\cal A}_i|\leq |{\cal A}_{i-1}|/2$ and $|{\cal A}_m|\geq |{\cal A}_{m-1}|/2$. We output ${\cal A}_m$ as the $q$-representative family for ${\cal A}$. The correctness of this following from  Lemma~\ref{lem:reptransitive}.
By Corollary~\ref{cor:repset uniform general2}, 
\begin{eqnarray*}
|{\cal A}_m|&\leq& {x^{-p}(1-x)^{-q}} \cdot 2^{o(p+q)}  \cdot \log |{\cal A}_{m-1}|\\
&\leq& {x^{-p}(1-x)^{-q}} \cdot 2^{o(p+q)}  \cdot \log 2|{\cal A}_{m}|\\
\mbox{ Thus, } \frac{|{\cal A}_m|}{\log |{\cal A}_m| }&\leq& {x^{-p}(1-x)^{-q}} \cdot 2^{o(p+q)}.
\end{eqnarray*}
We know that for some number $a$ and $b$, if $a \leq b $ then $a \log^2 a \leq b \log^2 b$. Applying this identity we get the following. 
\begin{eqnarray*}
\frac{|{\cal A}_m|}{\log |{\cal A}_m| } \log^2 \left(\frac{|{\cal A}_m|}{\log |{\cal A}_m| }\right) &\leq& {x^{-p}(1-x)^{-q}} \cdot 2^{o(p+q)}
\end{eqnarray*}
The above inequality implies that 
\begin{eqnarray*}
|{\cal A}_m| \leq \frac{|{\cal A}_m|}{\log |{\cal A}_m| } \log^2 \left(\frac{|{\cal A}_m|}{\log |{\cal A}_m| }\right) &\leq& {x^{-p}(1-x)^{-q}} \cdot 2^{o(p+q)}
\end{eqnarray*}
and thus $|{\cal A}_m| \leq  {x^{-p}(1-x)^{-q}} \cdot 2^{o(p+q)}$. 
By Lemma~\ref{cor:repset uniform general2}, the total running time $T$ to compute ${\cal A}_m$ is, 
\begin{eqnarray*}
T&=&  \sum_{i=1}^{m-1}   |{\cal A}_i|\cdot \log |{\cal A}_i| \cdot  \log{W}  + |{\cal A}_i| \cdot  (1-x)^{-q}\cdot 2^{o(p+q)} \cdot \log n ) \\
&=&   \sum_{i=1}^{m-1}\cO\Big(\frac{|{\cal A}|}{2^{i-1}}\cdot \log |{\cal A}| \cdot  \log{W} + \frac{|{\cal A}|}{2^{i-1}} \cdot  (1-x)^{-q}\cdot 2^{o(p+q)} \cdot \log n \Big) \quad (\mbox{since } |{\cal A}_i|\leq \frac{|{\cal A}|}{2^{i-1}}) \\ 
&=&\cO(|{\cal A}|\cdot \log |{\cal A}| \cdot  \log{W}  + |{\cal A}| \cdot  (1-x)^{-q}\cdot 2^{o(p+q)} \cdot \log n )
\end{eqnarray*}
This concludes the proof. 
\end{proof}

The size of the output representative family in Theorem~\ref{thm:repset uniform general3} is minimized when $x=\frac{p}{p+q}$. By substituting $x=\frac{p}{p+q}$ in Theorem~\ref{thm:repset uniform general3} we get the following corollary.
\begin{corollary}
\label{coro:uniformmatroidweighted}
There is an algorithm that given a $p$-family ${\cal A}$ of sets  over a universe $U$ of size $n$,  an integer $q$, 
and a non-negative weight function    \awf{} with maximum value at most $W$,
computes in time 
\[ \cO( |{\cal A}|\cdot \log |{\cal A}| \cdot  \log{W}  + |{\cal A}| \cdot  (1-x)^{-q}\cdot 2^{o(p+q)} \cdot \log n ) \] 
a subfamily 
$\widehat{\cal A}\subseteq \cal A$ such that $|\widehat{{\cal A}}| \leq {p+q \choose p} \cdot 2^{o(p+q)} $ and 
 \minrep{A}{q} (\maxrep{A}{q}).
\end{corollary}

\section{Applications}\label{section:application}
In this section we demonstrate how the  efficient construction of representative families can be used to design 
single-exponential  parameterized  and exact exponential time algorithms.  Our applications include best known 
deterministic algorithms for {\sc Long Directed Cycle},  \textsc{Minimum Equivalent Graph}, {\sc $k$-Path} 
and {\sc $k$-Tree}. 

\medskip

Let \mat{}  be a matroid with the ground set of size $n$ and $ \cS = \{S_1,\ldots, S_t\}$ be a $p$-family of independent sets.  Then for specific matroids 
we use the following notations to denote the time required  to compute  the following  $q$-representative families of  $ \cS$:
%

\begin{itemize}
\item
  \tc{rm}{t}{p}{q}  is the time required to compute a family   \rep{S}{q} of size $
  \binom{p+q}{q}$, when $M$ is a linear matroid.  
 
\item  \tc{um}{t}{p}{q} is the time     required to compute a family   \rep{S}{q} of size ${p+q \choose p} \cdot 2^{o(p+q)} \cdot \log n$, when $M$ is a uniform matroid and $x$ is chosen to be $\frac{p}{p+q}$.  
\end{itemize}

Let us remind, that by Theorem~\ref{thm:repsetlovasz}, when rank of $M$ is  $p+q$, \tc{rm}{t}{p}{q} is bounded by  \tgem \, multiplied by  the time required to perform operations over $ \mathbb{F}$.
By Corollary~\ref{coro:uniformmatroidweighted},  \tc{um}{t}{p}{q}$=  \cO(t\cdot (\frac{p+q}{q})^q \cdot \log n)$

\subsection{Long Directed Cycle}
\label{subsection:LDC}
In this section we give our first application of algorithms based on representative families.  We study the following problem.

 \defparproblem{ {\sc  Long Directed Cycle} }{A $n$-vertex and $m$-arc directed graph $D$ and a positive integer $k$.}{$k$ }
 {Does there exist a directed cycle of length  at least $ k$ in $D$?}

\bigskip

\noindent 
Observe that the  {\sc  Long Directed Cycle} problem is different from the well-known problem of finding a directed cycle of length {\em exactly} $k$. It is quite possible that the only directed cycle that has length at least $k$ is much longer than $k$, and possibly even is a Hamiltonian cycle. 
Let $D$ be a directed graph,   $k$ be a positive integer, and  \mat{} be a uniform matroid $U_{n,2k}$ where $E=V(D)$  and 
${\cal I}=\{S\subseteq V(D)~|~|S|\leq 2k\}.$ 
 In this subsection whenever we talk about independent sets, these are independent sets of the uniform matroid $U_{n,2k}$.  
 For a pair of vertices $u,v\in V(D)$, we define 
 \begin{eqnarray*}
 {\cal P}_{uv}^i& = & \Big\{X~\Big|~X\subseteq V(D),~u, v \in X, ~|X|=i, \mbox{ and there is a directed $uv$-path  in $D$  } \\ 
   & & \hspace{1cm} \mbox{ of  length $i-1$      with all the vertices  belonging to $X$}. \Big\}
  \end{eqnarray*}

We start with a structural lemma 
providing the  key insight to our algorithm. 

\begin{figure}[t]
\begin{center}
\scalebox{.7}{
\includegraphics[height=8cm]{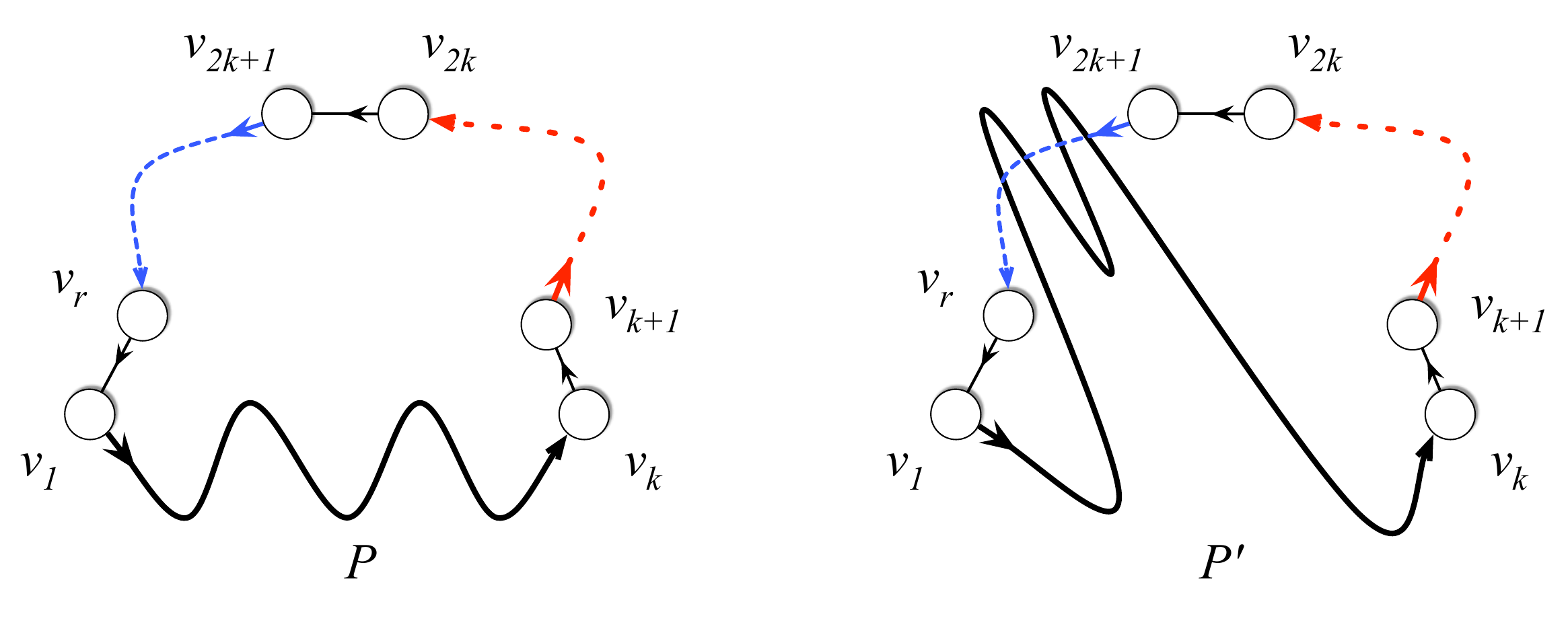}
}

\caption{Illustration to the proof of Lemma~\ref{lem:directedcyle}. }
\label{fig:dircycleproof}
\end{center}
\end{figure}

\begin{lemma}
\label{lem:directedcyle}
Let $D$ be a directed graph. Then $D$ has a directed  cycle of length at least $k$ if and only if  there exists a pair of vertices $u,v\in V(D)$ and 
  $X \in  \widehat{{\cal P}}_{uv}^k \subseteq_{rep}^k  {\cal P}^k_{uv}$ such that $D$ has a directed cycle $C$ and in this cycle vertices of  $X$ induce a directed path (that is, vertices of  $X$ form a consecutive segment in $C$). 
\end{lemma}
\begin{proof}
The reverse direction of the proof is straightforward---if cycle $C$ contains a path  of length $k$, the length of $C$ is at least $k$. We proceed with the proof of the forward direction. Let $C^*=v_1v_2 \cdots  v_{r}v_1$ be a smallest 
directed cycle in $D$ of length at least $k$. That is, $r\geq k$ and there is no directed cycle of length $r'$ where $k \leq r'<r$. We consider two cases.

\medskip\noindent \textbf{Case A: $r\leq 2k$.}
If $r\leq 2k$, then we take $u=v_1$ and $v=v_k$. We define paths  $P=v_1v_2\cdots v_k$  and    $Q=v_{k+1}\cdots v_r$. Because $|Q|\leq k$, by the definition  of $\widehat{{\cal P}}_{uv}^k \subseteq_{rep}^k  {\cal P}^k_{uv}$,  there exists a directed $uv$-path $P'$ such that $X=V(P') \in  \widehat{{\cal P}}_{uv}^k $  and $X\cap Q=\emptyset$. By replacing 
$P$ with $P'$ in $C^*$ we obtain a directed  cycle $C$  of length at least $k$ containing $P'$ as a subpath. 

\medskip\noindent \textbf{Case B: $r\geq 2k+1$.}
In this case we set $u=v_1$, $v=v_k$, and split $C^*$ into three paths $P=v_1\cdots v_k$, $Q=v_{k+1}\cdots v_{2k}$, and $R=v_{2k+1}\cdots v_r$.  Since $|Q|=k$ and $\widehat{{\cal P}}_{uv}^k \subseteq_{rep}^k  {\cal P}^k_{uv}$, it follows that   there exists an $uv$-path $P'$ such that $X=V(P') \in  \widehat{{\cal P}}_{uv}^k $  and $X\cap Q=\emptyset$.  However, $P'$ is not necessarily disjoint with $R$ and by replacing 
$P$ with $P'$ in $C^*$ we can obtain a {\em closed walk} $C'$  containing $P'$ as a subpath.  See Fig.~\ref{fig:dircycleproof} for an illustration.  

If $X\cap R=\emptyset$,  then $C'$ is a simple cycle and we   take $C'$ as  the desired $C$. 
We claim that this is the only possibility.  
Let us assume targeting towards a  contradiction that $X\cap R\neq \emptyset$. We want to show that in this case there is a cycle of length at least $k$ but shorter than $C^*$, contradicting the choice of $C^*$. Let $v_\alpha$ be the last vertex in $X\cap R$ when we walk from $v_1$ to $v_k$ along $P'$. Let $P'[v_\alpha, v_k]$ be the 
subpath of $P'$ starting at $v_\alpha$ and ending at $v_k$. If $v_\alpha=v_{2k+1}$, we set $R'=\emptyset$. Otherwise  we put  
  $R'=R[v_{2k+1}, v_{\alpha -1}]$ to be the subpath of $R$ starting at 
$v_{2k+1}$ and ending at $v_{\alpha-1}$. Observe that since the  arc $v_{\alpha-1}v_\alpha$ is present in $D$ (in fact it is   an arc of  the cycle $C^*$), we have that  
$\overline{C}=P'[v_\alpha ,v_k]QR'$ is a simple cycle in $D$. Clearly, $|\overline{C}|\geq |Q|\geq k$. Furthermore, since $v_1$ is not present in 
$P'[v_\alpha, v_k]$ we have that $|P'[v_\alpha, v_k]|<|P'|=|P|$. Similarly since $v_\alpha$ is not present in $R'$, we have that $|R'|<|R|$. Thus  we have
\[ k\leq | \overline{C}|=|P'[{v_\alpha, v_k}]|+|Q|+|R'|< |P|+|Q|+|R|=|C^*|. \]
This implies that $ \overline{C}$ is a directed simple cycle of length at least $k$ and  strictly smaller than $r$. This is a contradiction. 
Hence  by replacing 
$P$ with $P'$ in $C^*$ we obtain a directed cycle $C$  containing $P'$ as a subpath. This concludes the proof. 
\end{proof}

Next   lemma provides  an efficient computation  of    family   $\widehat{{\cal P}}_{uv}^k \subseteq_{rep}^k  {\cal P}^k_{uv}$. 
The next lemma is provided to give a simple exposition of representative families based dynamic programming algorithm. 
  

\begin{lemma}
\label{lem:pathrepsetfinder}
Let $D$ be a directed/unidrected graph with $n$ vertices and $m$ edges, $u\in V(D)$ and 
 \mat{} be an uniform matroid $U_{n,\ell}$ where $E=V(D)$  and ${\cal I}=\{S\subseteq V(D)~|~|S|\leq \ell\}$. 
 Then for  every $p\leq \ell$ and  $v\in V(D)\setminus \{u\}$, a family $\widehat{{\cal P}}_{uv}^p \subseteq_{rep}^{\ell-p}  {\cal P}^p_{uv}$ 
 of size at most 
 \[{\ell \choose p} \cdot 2^{o(\ell)} 
 \]
 can be found 
 in time   
\[ \cO\left( 2^{o(\ell)}   m \log n \max_{i\in [p]} \left\{{\ell \choose i-1}   \left( \frac{\ell}{\ell-i} \right)^{\ell-i} \right\} \right).\]
 Furthermore, within the same running time every set in $\widehat{{\cal P}}_{uv}^p$ can be ordered in a way that 
 it corresponds to a directed  (undirected) path in $D$. 
\end{lemma}
\begin{proof}
We prove the lemma only for digraphs. The proof for undirected graphs is analogous and we only point out the differences with the proof for the directed case.   We describe a dynamic programming based algorithm. Let $V(D)=\{u,v_1,\ldots,v_{n-1}\}$ and 
${\cal D}$ be a $(p -1) \times (n-1)$ matrix where the rows are indexed from integers in $\{2,\ldots, p\}$ 
and the columns are indexed from vertices in $\{v_1,\ldots,v_{n-1}\}$. The entry ${\cal D}[i,v]$ will store the family  
$\whnd{\cal P}^i_{uv} \subseteq_{rep}^{\ell-i} {\cal P}^i_{uv}$. We 
fill the entries in the matrix $\cal D$ in the increasing order of rows. For $i=2$, ${\cal D}[2,v]=\{\{u,v\}\}$ if $uv\in A(D)$ (for an undirected graph we check whether $u$ and $v$ are adjacent). Assume that 
we have filled all the entries until the row $i$. Let 
\[ {\cal N}_{uv}^{i+1}=\bigcup_{w \in N^{-}(v)} \whnd{\cal P}^i_{uw} \bullet \{v\}. \]
For undirected graphs we use the following definition  
\[{\cal N}_{uv}^{i+1}=\bigcup_{w \in N(v)} \whnd{\cal P}^i_{uw} \bullet \{v\}. \]

\begin{claim}
\label{lem:kpathauxrepset}
${\cal N}_{uv}^{i+1} \subseteq_{rep}^{\ell-(i+1)} {\cal P}_{uv}^{i+1}$.
\end{claim}
\begin{proof}
Let $S \in {\cal P}_{uv}^{i+1}$ and  $Y$ be a set of size $\ell-(i+1)$ (which is essentially an independent set of $U_{n,\ell}$) such that 
$S\cap Y=\emptyset$. We will show that there exists a set $S' \in {\cal N}_{uv}^{i+1}$ such that $S' \cap Y=\emptyset$. This will imply 
the desired result. Since $S \in {\cal P}_{uv}^{i+1}$ there exists a directed path $P=ua_1\cdots a_{i-1}v$  in $D$ such that 
$S=\{u,a_1,\ldots ,a_{i-1},v\}$ and $a_{i-1} \in N^{-}(v)$. The existence of path $P[u,a_{i-1}]$, the subpath of $P$ between $u$ and $a_{i-1}$, implies that  $X^*=S\setminus \{v\} \in {\cal P}_{ua_{i-1}}^{i}$. 
Take $Y^*=Y \cup \{v\}$. Observe that $X^* \cap Y^*=\emptyset$  and $|Y^*|=\ell-i$. Since $\whnd{\cal P}^i_{ua_{i-1}} \subseteq_{rep}^{\ell-i} {\cal P}^i_{ua_{i-1}} $  
there exists a set $\widehat{X}^*\in \whnd{\cal P}^i_{ua_{i-1}} $ such that $\widehat{X}^*\cap Y^*=\emptyset$. However, since $a_{i-1}\in N^{-}(v)$ 
and $\widehat{X}^*\cap \{v\}=\emptyset$ (as $\widehat{X}^*\cap Y^*=\emptyset$), we have  $\widehat{X}^* \bullet \{v\}= \widehat{X}^*\cup \{v\}$ and $\widehat{X}^*\cup \{v\} \in {\cal N}_{v}^{i+1}$. 
Taking $S'=\widehat{X}^*\cup \{v\}$ suffices for our purpose. This completes the proof of the lemma. 
\end{proof}

We fill the entry for ${\cal D}[i+1,v]$ as follows. Observe that 
$${\cal N}_{uv}^{i+1}= \bigcup_{w \in N^{-}(v)} {\cal D}[i,w] \bullet \{v\} .$$
We already have computed the family corresponding to ${\cal D}[i,w] $ for $w\in N^{-}(v)$. By Corollary~\ref{coro:uniformmatroidweighted},   
 $|\whnd{\cal P}^i_{uw}|\leq {\ell \choose i}2^{o(\ell)} $ and thus $|{\cal N}_{uv}^{i+1}|\leq d^-(v) {\ell \choose i} 2^{o(\ell)}$.  Furthermore,  
 we can  compute ${\cal N}_{uv}^{i+1}$ in time $\cO\left(d^-(v) {\ell \choose i} 2^{o(\ell)}  \right)$. Now using Corollary~\ref{coro:uniformmatroidweighted},   
 we  compute  $\widehat{\cal N}_{uv}^{i+1} \subseteq_{rep}^{\ell-i-1}{\cal N}_{uv}^{i+1}$ in time \tc{um}{t}{i+1}{\ell-i-1}, where 
 $t=d(v){\ell \choose i}2^{o(\ell)} $.  
 By Claim~\ref{lem:kpathauxrepset}, we know that ${\cal N}_{uv}^{i+1} \subseteq_{rep}^{\ell-i-1} {\cal P}_{uv}^{i+1}$. Thus Lemma~\ref{lem:reptransitive} implies that $\widehat{\cal N}_{uv}^{i+1} = \whnd{\cal P}^{i+1}_{uv} \subseteq_{rep}^{\ell-i-1} {\cal P}^{i+1}_{uv}$. We assign this family to ${\cal D}[i+1,v]$. This completes the description and the correctness of the algorithm. We give ordering to the vertices of 
 the sets in $\widehat{{\cal P}}_{uv}^p$ in the following way so that it corresponds to a directed  (undirected) path in $D$.  We keep the sets in the order in which they are built using the $\bullet$ operation. That is, we can view these sets as strings and 
 $\bullet$ operation as concatenation. Then every ordered set in our family represents a path in the graph. 
The running time of the algorithm is bounded by
\begin{eqnarray*}
& & \cO\left( \sum_{i=2}^{p} \sum_{j=1}^{n-1} \tcwd{um}{d^-(v_j){\ell \choose i-1}2^{o(\ell)}}{i}{\ell-i} \right) \\
& = & \cO\left( \sum_{i=2}^{p} \sum_{j=1}^{n-1} d^-(v_j) {\ell \choose i-1}   \left( \frac{\ell}{\ell-i} \right)^{\ell-i} 2^{o(\ell)}  \log n \right) \\
& = & \cO\left( 2^{o(\ell)}  \log n \sum_{i=2}^{p}   \sum_{j=1}^{n-1} d^-(v_j) {\ell \choose i-1}   \left( \frac{\ell}{\ell-i} \right)^{\ell-i}  \right)\\
& = & \cO\left( 2^{o(\ell)} m \log n  \max_{i\in [p]} \left\{{\ell \choose i-1}   \left( \frac{\ell}{\ell-i} \right)^{\ell-i} \right\} \right)
\end{eqnarray*}
This completes the proof
\end{proof}

Finally, we are ready to state the main result of this section. 
 \begin{theorem}
 \label{thm:longdirectedcycle}
  {\sc  Long Directed Cycle}  can be solved in time $\cO(8^{k+o(k)}  mn^2).$
 \end{theorem}
\begin{proof} Let $D$ be a directed graph. 
We solve the problem by applying the structural characterization proved in Lemma~\ref{lem:directedcyle}. 
By Lemma~\ref{lem:directedcyle},   $D$ has a directed  cycle of length at least $k$ if and only if  there exists a pair of 
vertices $u,v\in V(D)$ and a path $P'$ with $V(P') \in  \widehat{{\cal P}}_{uv}^k \subseteq_{rep}^k  {\cal P}^k_{uv}$ such that $D$ has a directed cycle 
$C$ containing $P'$ as a subpath. 

We first compute  $ \widehat{{\cal P}}_{uv}^k \subseteq_{rep}^k  {\cal P}^k_{uv}$ for all $u,v\in V(D)$. 
For that we apply Lemma~\ref{lem:pathrepsetfinder} for  each vertex $u\in V(D)$ with $\ell=2k$ and $p=k$.  Thus, we can compute 
$ \widehat{{\cal P}}_{uv}^k \subseteq_{rep}^k  {\cal P}^k_{uv}$ for all $u,v\in V(D)$ in time $\cO\left(8^{k+o(k)} mn\log n\right)$. 
Moreover, for every $X\in \widehat{{\cal P}}_{uv}^k$ we also compute a directed $uv$-path $P_X$ using vertices of $X$.  Let 
\[{\cal Q}=\bigcup_{u,v\in V(D)} \widehat{{\cal P}}_{uv}^k. \]
Now for every 
set $X\in \cal Q$ and the corresponding $uv$-path $P_X$ with endpoint, we check if there is a $uv$-path in $D$ avoiding all vertices of $X$ but $u$ and $v$. 
 This check can be done by a standard graph traversal algorithm like BFS/DFS   in time  $\cO(m+n)$. 
If we succeed in finding a path for at least  one $X\in \cal Q$, we answer 
{\sf YES} and return the corresponding directed cycle obtained by merging $P_X$ and another path. 
Otherwise, if we did not succeed to find such a path for any of the sets $X\in \cal Q$, this means that 
 there is no directed cycle of length at least $k$ in $D$. The correctness of the algorithm follows from 
Lemma~\ref{lem:directedcyle}. By 
Corollary~\ref{coro:uniformmatroidweighted}, the size of $\cal Q$ is upper bounded by 
$n^2 {2k \choose k} 2^{o(k)} 
\leq n^2 4^{k+o(k)}$. Thus the overall running time of the algorithm is upper bounded by 
\[\cO(8^{k+o(k)} mn\log n + 4^{k+o(k)} (n^2m+n^3) ).\]
This concludes the proof. 
\end{proof}

\subsection{Faster Long Directed Cycle}
\label{subsection:fasterLDC}
In this subsection we design a faster algorirthm for {\sc Long Directed Cycle}.
In Subsection~\ref{subsection:LDC} we have seen an algorithm for  
{\sc Long Directed Cycle} where the running time mainly depend on the 
computation of representative families $\widehat{{\cal P}}_{uv}^p \subseteq_{rep}^{q}  {\cal P}^p_{uv}$ for $2\leq p\leq k$ and $q=2k-p$. We used  Theorem~\ref{thm:repset uniform general} with $x=\frac{p}{p+q}$ (i.e, Corollary~\ref{coro:uniformmatroidweighted}) 
to compute representative families. The choice $x=\frac{p}{p+q}$ minimizes the size of representative family. But in fact, we can choose 
$x$ that minimizes the running time instead. 

Now we find out the choice of $x$ which minimizes the computation of 
$\widehat{{\cal P}}_{uv}^p \subseteq_{rep}^{q}  {\cal P}^p_{uv}$ 
for $2\leq p\leq k$ and $q=2k-p$. Let $s_{p,q}$ denote the size of 
$\widehat{{\cal P}}_{uv}^p$. We know that  the computation of 
$\widehat{{\cal P}}_{uv}^p \subseteq_{rep}^{q}    {\cal N}_{uv}^{p}\subseteq_{rep}^{q}  {\cal P}^p_{uv}$ 
depends on  $|{\cal N}_{uv}^{p}|$, which depends on the size of 
the representative families $\widehat{{\cal P}}_{uw}^{p-1}$. That is 
$|{\cal N}_{uv}^{p}|\leq s_{p-1,q+1}\cdot n$.  Thus the value of $s_{p-1,q+1}$ and  $s_{p,q}$ are ``almost equal'' and we denote it by 
$s_{p-1,q+1}\approx s_{p,q}$. 
By Theorem~\ref{thm:repset uniform general3}, the running time to compute  
$\widehat{{\cal P}}_{uv}^p \subseteq_{rep}^{q}    {\cal N}_{uv}^{p}\subseteq_{rep}^{q}  {\cal P}^p_{uv}$ is, 
\begin{eqnarray*}
 &&\cO\left(|{\cal N}_{uv}^{p}|\cdot(1-x)^{-q}\cdot 2^{o(p+q)} \cdot \log n\right)\\
&=&\cO\left(s_{p,q} \cdot(1-x)^{-q}\cdot 2^{o(p+q)} \cdot n \log n \right)\\
&=&\cO\left( x^{-p} \cdot(1-x)^{-2q}\cdot 2^{o(p+q)} \cdot n \log n\right)\\
\end{eqnarray*}
To minimize the above running time it is enough to minimize the function 
$f(x)=x^{-p} \cdot(1-x)^{-2q}$. Using methods from calculus we know that  the value $x^*$ of $x$ for which  $f'(x^*)=0$ 
corresponds to  a minimum value of the function $f(x)$ if 
$f''(x^*)>0$.  The derivative of $f(x)$ is, $f'(x)=-px^{-p-1}(1-x)^{-2q}+2q\cdot x^{-p}(1-x)^{-2q-1}$.  Now consider the value of $x$ for which $f'(x)=0$.
\begin{eqnarray*}
-px^{-p-1}(1-x)^{-2q}+2q\cdot x^{-p}(1-x)^{-2q-1}&=&0\\
-p(1-x)+2q\cdot x &=&0\\
x&=&\frac{p}{p+2q}
\end{eqnarray*}
Set $x^*=\frac{p}{p+2q}$. To prove $f(x)$ is minimized at $x^*$, it is enough to show that $f''(x^*)>0$. 
\begin{eqnarray*}
f'(x)&=&-px^{-p-1}(1-x)^{-2q}+2q\cdot x^{-p}(1-x)^{-2q-1}\\
&=&  x^{-p} (1-x)^{-2q}(-p\cdot x^{-1}+2q\cdot (1-x)^{-1})\\
&=&  f(x)\cdot (-p\cdot x^{-1}+2q\cdot (1-x)^{-1})\\
f''(x) &=& f(x) \cdot (p\cdot x^{-2}+2q\cdot (1-x)^{-2})+ 
f'(x)\cdot (-p\cdot x^{-1}+2q\cdot (1-x)^{-1})\\
f''(x^*) &=& f(x^*) \cdot (p\cdot (x^*)^{-2}+2q\cdot (1-(x^*))^{-2})>0
\end{eqnarray*}
Hence the run time to compute $\widehat{{\cal P}}_{uv}^p \subseteq_{rep}^{q}  {\cal P}^p_{uv}$ is minimized when $x=\frac{p}{p+2q}$. 
\begin{lemma}
\label{lem:FastLDCpathrepset} 
Let $D$ be a directed graph with $n$ vertices and $m$ edges, $u\in V(D)$ and 
 \mat{} be an uniform matroid $U_{n,\ell}$ where $E=V(D)$  and ${\cal I}=\{S\subseteq V(D)~|~|S|\leq \ell\}$.  Then 
for every  $v\in V(D)\setminus\{u\}$ and  
integer $2 \leq p \leq \ell$ 
there is an algorithm that computes a family 
$\widehat{{\cal P}}_{uv}^p \subseteq_{rep}^{\ell-p}  {\cal P}^p_{uv}$ 
of size $\left(\frac{2\ell-p}{p}\right)^p\left(\frac{2\ell-p}{2\ell-2p}\right)^{\ell-p}\cdot 2^{o(\ell)}$ 
in time
$\cO\left(2^{o(\ell)} \cdot m \log n \cdot \max_{i\in[p]}\left\{\left(\frac{2\ell-i}{i}\right)^i \left(\frac{2\ell-i}{2\ell-2i}\right)^{2\ell-2i}\right\}\right)$
\end{lemma}
\begin{proof}
The proof is same as the proof of Lemma~\ref{lem:pathrepsetfinder}, except the choice of $x$ while applying Theorem~\ref{thm:repset uniform general} (instead of Corollary~\ref{coro:uniformmatroidweighted}).  As in the proof of Lemma~\ref{lem:pathrepsetfinder},
we have a dynamic programming table 
${\cal D}$ where the rows are indexed from integers in $\{2,\ldots, p\}$ 
and the columns are indexed from vertices in $\{v_1,\ldots,v_{n-1}\}$. The entry ${\cal D}[i,v]$ will store the family  
$\whnd{\cal P}^i_{uv} \subseteq_{rep}^{\ell-i} {\cal P}^i_{uv}$. We 
fill the entries in the matrix $\cal D$ in the increasing order of rows. For $i=2$, ${\cal D}[2,v]=\{\{u,v\}\}$ if $uv\in A(D)$. Assume that 
we have filled all the entries until the row $i$. Let 
\[ {\cal N}_{uv}^{i+1}=\bigcup_{w \in N^{-}(v)} \whnd{\cal P}^i_{uw} \bullet \{v\}. \]
Due to Claim~\ref{lem:kpathauxrepset}, we have that  
${\cal N}_{uv}^{i+1} \subseteq_{rep}^{\ell-(i+1)} {\cal P}_{uv}^{i+1}$. 
Lemma~\ref{lem:reptransitive} implies that $\widehat{\cal N}_{uv}^{i+1} = \whnd{\cal P}^{i+1}_{uv} \subseteq_{rep}^{\ell-i-1} {\cal P}^{i+1}_{uv}$. We assign this family to ${\cal D}[i+1,v]$.  

Now we explain the computation of  $\widehat{\cal N}_{uv}^{i+1}= \whnd{\cal P}^{i+1}_{uv}$. For any $j$, to compute 
$\widehat{\cal N}_{uv}^{j}= \whnd{\cal P}^{j}_{uv}$, we apply 
Theorem~\ref{thm:repset uniform general3} with  the value $x_j$ for $x$, where 
$$x_j=\frac{j}{j+2(\ell-j)}=\frac{j}{2\ell-j}$$  
Let $s_{j,\ell-j}$ be the size of the representative family   $\widehat{\cal N}_{uv}^{j}= \whnd{\cal P}^{j}_{uv}$ when we apply Theorem~\ref{thm:repset uniform general3} with  the value $x_j$. 
That is $s_{j,\ell-j}=(x_j)^{-j}(1-x_j)^{\ell-j}\cdot 2^{o(\ell)}$.  
Assume that we have computed   $\whnd{\cal P}^{j}_{uw}$ of size $s_{j,\ell-j}$ 
and stored it in ${\cal D}[j,w]$ for all $j\leq i$ and 
$w\in \{v_1,\ldots,v_{n-1}\}$. Now consider the computation of $\widehat{\cal N}_{uv}^{i+1}= \whnd{\cal P}^{i+1}_{uv}$. 
We apply  Theorem~\ref{thm:repset uniform general3} with value $x_{i+1}$ for $x$ to compute $\widehat{\cal N}_{uv}^{i+1} \subseteq_{rep}^{\ell-(i+1)} {\cal N}_{uv}^{i+1}$. 
Since  ${\cal N}_{uv}^{i+1}=\bigcup_{w \in N^{-}(v)} \whnd{\cal P}^i_{uw} \bullet \{v\}$, we have that 
\begin{eqnarray*}
|{\cal N}_{uv}^{i+1}|&\leq& s_{i,\ell-i}\cdot d^{-}(v)\\
&\leq& (x_i)^{-i}(1-x_i)^{\ell-i}\cdot 2^{o(\ell)} d^{-}(v)
\end{eqnarray*}  
By Theorem~\ref{thm:repset uniform general3}, the running time to compute $\widehat{\cal N}_{uv}^{i+1}$ is,
\begin{align}\label{eqn:runtime1}
 &s_{i,\ell-i} \cdot (1-x_{i+1})^{\ell-(i+1)} \cdot 2^{o(\ell)} \cdot d^{-}(v)\cdot \log n 
\end{align}
To analyze the running time further we need the following claim.
\begin{claim}\label{clm:value of s} For any $3<i<p$,  
$s_{i,\ell-i} \leq e^2 \cdot (i+1) \cdot s_{i+1,\ell-i-1}$. 
\end{claim}
\begin{proof}
By applying the definition of $s_{i}$ and $x_{i+1}$ we get he following inequality.
\begin{align*}
\frac{s_{i,\ell-i}}{s_{i+1,\ell-i-1}} & = 
\frac{x_i^{-i}(1-x_i)^{-\ell+i}}{x_{i+1}^{-(i+1)}(1-x_{i+1})^{-\ell+(i+1)}}\\
&=\left(\frac{2\ell-i}{i} \right)^i \left(\frac{2\ell-i}{2\ell-2i} \right)^{\ell-i} 
\left(\frac{i+1}{2\ell-(i+1)} \right)^{i+1} \left(\frac{2\ell-2(i+1)}{2\ell-(i+1)} \right)^{\ell-(i+1)}\\
&=\left(\frac{2\ell-i}{2\ell-(i+1)} \right)^{\ell}\cdot \frac{(i+1)^{i+1}}{i^i} \cdot
 \frac{(2\ell-2(i+1))^{\ell-(i+1)}}{2\ell-2i)^{\ell-i}}\\
&\leq \left(1+\frac{1}{2\ell-(i+1)} \right)^{2\ell-(i+1)}\cdot (i+1)\cdot \left(1+\frac{1}{i}\right)^i\\
& \leq e^2 \cdot (i+1).
\end{align*}
In the last transition we used that $(1+1/x)^x < e$ for every $x > 0$.
\end{proof} 
From Equation~\ref{eqn:runtime1} and Claim~\ref{clm:value of s} we have that the running time for computing $\widehat{\cal P}_{uv}^{p}$ 
is bounded by
\begin{align*}
&\cO\left(\sum_{i=2}^{p}\sum_{j=1}^{n-1}s_{i,\ell-i}\cdot d^{-}(v_j) \cdot (1-x_{i})^{-\ell+i} \cdot 2^{o(\ell)} \cdot \log n\right) \\
=\;&\cO\left(2^{o(\ell)} \cdot m \log n \cdot \max_{i\in[p]}\left\{\left(\frac{2\ell-i}{i}\right)^i \left(\frac{2\ell-i}{2\ell-2i}\right)^{2\ell-2i}\right\}\right)  
\end {align*}
The size of the family $\widehat{{\cal P}}_{uv}^p \subseteq_{rep}^{\ell-p} {\cal N}_{uv}^{p} \subseteq_{rep}^{\ell-p}  {\cal P}^p_{uv}$ is, 
$$s_{p,\ell-p}=(x_p)^{-p}(1-x_p)^{-\ell+p}\cdot 2^{o(\ell)} = \left(\frac{2\ell-p}{p}\right)^p\left(\frac{2\ell-p}{2\ell-2p}\right)^{\ell-p}\cdot 2^{o(\ell)}.$$
This completes the proof.
\end{proof}
We now have a faster algorithm to compute the representative family~$\widehat{{\cal P}}_{uv}^k \subseteq_{rep}^{k}  {\cal P}^p_{uv}$. 
Using  Lemma~\ref{lem:FastLDCpathrepset}, we can compute $\widehat{\cal P}_{uv}^{k}$, for all $v\in V(D)\setminus \{u\}$ 
in time  
$$\cO\left(2^{o(k)} \cdot m \log n \cdot \max_{i\in[p]}\left\{\left(\frac{4k-i}{i}\right)^i \left(\frac{4k-i}{4k-2i}\right)^{4k-2i}\right\}\right).$$ 
Simple calculus shows that the maximum is attained for $i=k$. Hence the running time to compute  $\widehat{\cal P}_{uv}^{k}$ for all $u,v\in V(D)$ is upper bounded by $\cO(6.75^{k+o(k)}nm\log n)$. 
This yields an improved bound for the running time of our algorithm for {\sc Long Directed Cycle}. 

We apply Lemma~\ref{lem:FastLDCpathrepset} for each $u\in V(D)$ with $\ell=2k$ and $p=k$. Thus, we can compute $\widehat{{\cal P}}_{uv}^k \subseteq_{rep}^{k}  {\cal P}^p_{uv}$ for all $u,v\in V(D)$ in time 
$\cO(6.75^{k+o(k)}nm\log n)$.    
The size of the family  $\widehat{\cal P}_{uv}^{k}$ for any $u,v\in V(D)$ is upper bounded by $\cO(4.5^{k+o(k)})$. 
Thus, if we now loop over every set in the representative families and run a breadth first search, just as in the proof of Theorem~\ref{thm:longdirectedcycle}, 
this will take at most $\cO(6.75^{k+o(k)}nm\log^ n + 4.5^{k+o(k)} (n^3+n^2m))$ time. Hence we arrive at the following theorem. 
\begin{theorem}\label{thm:fastLDC}
There is a $\cO(6.75^{k+o(k)}mn^2)$ time algorithm for {\sc Long Directed Cycle}
\end{theorem}

\subsection{Minimum Equivalent Graph}

%
%
%
%

 
 For a given digraph $D$, a subdigraph $D'$ of $D$ is said to be an {\em equivalent} subdigraph of $D$ if for any pair of vertices $u,v\in V(D)$ if there is a directed path in $D$ from $u$ to $v$  then there is also a directed path from $u$ to $v$ in $D'$.  That is, reachability of vertices in $D$ 
 and $D'$ is same. In this section we study a problem  where given a digraph $D$ the objective is to find an equivalent subdigraph of $D'$ of $D$ with as few arcs as possible. Equivalently, the objective is to remove the maximum number of arcs from a digraph $D$  without affecting 
its  reachability.  More precisely the problem we study is as follows. 
 

 \begin{center} 
\fbox{\begin{minipage}{0.96\textwidth}
\noindent{\sc Minimum Equivalent Graph ({\sc MEG})}
 \\
\noindent {\bf Input}: A  directed graph $D$\\
\noindent{\bf Task}:  Find an equivalent subdigraph   of $D$ with the minimum number of arcs. 
\end{minipage}}
\end{center}
\medskip


The following proposition is due to  Moyles and Thompson~\cite{MoylesT69},  see also \cite[Sections 2.3]{BangG089_book},  reduces the problem of finding a minimum equivalent subdigraph of an arbitrary $D$ to a strong digraph. 

\begin{proposition}\label{prop:medredtostrong}
Let $D$ be a digraph  on $n$ vertices  with strongly connected components $C_1,\ldots,C_r$.  Given a minimum equivalent subdigraph $C_i'$ for each $C_i$, $i\in [r]$, one can obtain a minimum equivalent subdigraph $D'$ of $D$ containing each of $C'_i$ in $\cO(n^\omega)$ time. 
\end{proposition}
  Observe that for a strong digraph $D$ any equivalent subdigraph is also  strong. 
  By Proposition~\ref{prop:medredtostrong}, {\sc MEG} reduces to the following problem.


 \begin{center} 
\fbox{\begin{minipage}{0.96\textwidth}
\noindent{\sc Minimum  Strongly Connected Spanning Subgraph ({\sc Minimum SCSS})}
 \\
\noindent {\bf Input}: A strongly connected directed  graph $D$\\
\noindent{\bf Task}: Find a strong spanning subdigraph   of $D$ with the minimum number of arcs. 
\end{minipage}}
\end{center}
\medskip
It seems to be no established agreement in the literature on how to call these problems. 
{\sc MEG} 
sometimes is also referred  as 
{\sc Minimum Equivalent Digraph} and {\sc Minimum Equivalent Subdigraph}, while 
{\sc Minimum SCSS} is also called 
{\sc Minimum  Spanning Strong Subdigraph ({\sc MSSS})}.

A  digraph $T$ is an {\em out-tree} (an {\em in-tree}) if  $T$ is an oriented tree with just one vertex $s$ of in-degree zero (out-degree zero). 
The vertex $s$ is the root of $T$. If an out-tree (in-tree) $T$ is a spanning subdigraph of $D$, $T$ is called an 
{\em out-branching} (an {\em in-branching}).  We use the notation $B_s^+$ ($B_s^-$) to denote an out-branching 
(in-branching) rooted at $s$ of the digraph.    

It is known  that a digraph is strong if and only if it contain an out-branching and an in-branching rooted at some vertex $v\in V(D)$~\cite[Proposition 12.1.1]{BangG089_book}.
 \begin{proposition}\label{prop:two_branching}
Let $D$ be a strong digraph on $n$ vertices, let $v$ be an arbitrary vertex of $V(D)$, and $\ell\leq n-2 $ be a natural number. Then there exists a strong spanning subdigraph of $D$ with at most $2n-2 -\ell$ arcs if and only if $D$ contains an in-branching $B_v^-$ and an out-branching $B_v^+$ with root $v$ so that $|A(B_v^+)\cap A(B_v^-)|\geq \ell$ (that is, they have at least $\ell$ common arcs). 
 \end{proposition}
 
 Proposition~\ref{prop:two_branching} implies that  the  {\sc Minimum SCSS} problem is equivalent to finding, for an arbitrary vertex $v \in  V (D)$, an out-branching  $B_v^+$ and an in-branching $B_v^-$ that maximizes  $|A(B_v^+)\cap A(B_v^-)|$.  For our exact algorithm for {\sc Minimum SCSS} 
 we implement this equivalent version using representative sets.

%
%
%

 Let $D$ be a strong digraph and $s\in V(D)$ be a fixed vertex. For $v\in V(D)$  we use $\In{v}$ and $\Out{v}$ to denote the sets of in-coming and out-going arcs incident with $v$. By $D_s^{-}$ we denote the digraph obtained from $D$ by deleting the arcs in $\Out{s}$. Similarly, 
by  $D_s^{+}$  we denote the digraph obtained from $D$  by deleting the arcs in $\In{s}$. 
 
 We take two copies $E_1, E_2$ of $A(D)$ (that is $E_i=\{e_i~|~e\in A(D)\}$) , a copy $E_3$ of $A(D_s^+)$ and a copy $E_4$ of 
 $A(D_s^-)$ and construct four matroids as follows. Let $U(D)$ denote the underlying undirected graph of $D$. 
  The first two matroids \matl{1}, \matl{2} are the graphic matroids on $U(D)$.
Observe that \[A(D_s^+)=\biguplus_{v\in V(D_s^+)} \In{v}\mbox{ and } A(D_s^-) = \biguplus_{v\in V(D_s^-)} \Out{v}. \]
Thus the arcs of $D_s^+$ can be partitioned into sets  of in-arcs  and similarly the  arcs of $D_s^-$ 
into sets of out-arcs. The other two matroids are the following partition matroids \matl{3}, \matl{4}, where 
  \[
   {\cal I}_3=\{I ~|~ I\subseteq A(D_s^+), \text{ for every } v\in V(D_s^+)=V(D),   |I\cap \In{v}|\leq 1 \text \},
  \]
  and
  \[
   {\cal I}_4=\{I ~|~ I\subseteq A(D_s^-), \text{ for every } v\in V(D_s^-)=V(D), |I\cap \Out{v}|\leq 1 \}.
  \] 
  \noindent 
 We define the matroid \mat{} as the direct sum    $M=M_1  \oplus M_2 \oplus    M_3\oplus    M_4$. Since each of $M_i$ is a  representable matroids over the same field (by Propositions~\ref{prop:uniformandpartitionrep} and \ref{prop:graphicrep}), we have that  $M$ is also representable (Proposition~\ref{prop:disjointsumrep}). The reason we say that $M_i$ is representable over the same field $\mathbb F$ is that the graphic matroid is representable over any field and the partition matroids defined here are  representable over a finite field of size 
$n^{\cO(1)}$. So if we take $\mathbb F$ as a finite field of size  $n^{\cO(1)}$ then  $M$ is representable over $\mathbb F$. The rank of this 
 matroid is $4n-4$.   

 Let us note that for each arc $e\in A(D)$ which is not incident with $s$, we have four elements in the matroid $M$, corresponding to the copies of $e$ in 
 $M_i$,  $i\in \{1,\dots, 4\}$. We denote these elements  by $e_i$, $i\in \{1,\dots, 4\}$. For every edge $e \in A(D)$  incident with $s$, we have three corresponding elements. We denote them by $e_1, e_2, e_3$, or $e_1, e_2, e_4$, depending on the case when $e$ is in- or out-arc for $s$.
 
 For   $i\in \{1,\ldots, n-1\}$, we define 
 \begin{eqnarray*}
 {\cal B}^{4i} =   \Big\{W~\Big|~W\in {\cal I},~|W|=4i,~\forall~e\in A(D)  \mbox{ either $W\cap \{e_1,e_2,e_3,e_4\}=\emptyset$ or 
 $ \{e_1,e_2,e_3,e_4\} \subseteq W$}  \Big\}.  
  \end{eqnarray*}
 \noindent 
 For $W\in \cal I$, by $A_W$ we denote the set of arcs $e\in A(D)$ such that $ \{e_1,e_2,e_3,e_4\} \cap W\neq \emptyset$.   Now we are 
 ready to state the lemma that relates representative sets and the {\sc Minimum SCSS} problem.
 
 \begin{lemma}
 \label{lem:medrelationrepset}
  Let $D$ be a strong digraph on $n$ vertices and $\ell\leq n-2 $ be a natural number. 
  Then there exists a strong spanning subdigraph $D'$ of $D$ with at most $2n-2 -\ell$ arcs if and only if  there exists a set 
  $\widehat{F} \in \widehat{\cal B}^{4\ell}\subseteq_{rep}^{n'-4\ell}  {\cal B}^{4\ell}$ such that $D$ has a strong spanning subdigraph 
  $\bar{D}$ with $A_{\widehat{F}}\subseteq A(\bar{D})$.  Here, $n'=4n-4$. 
 \end{lemma}
 \begin{proof}
We only show the forward	 direction of the proof, the reverse direction is straightforward. Let $D'$ be a strong spanning subdigraph of $D$ with at most  $2n-2 -\ell$ arcs. Thus, by Proposition~\ref{prop:two_branching} we have that for any vertex $v\in V(D') $, there exists  an 
out-branching $B_v^+$ and an in-branching $B_v^-$ in $D'$ such that  $|A(B_v^+)\cap A(B_v^-)|\geq \ell$. Observe that the arcs in $A(B_v^+)\cap A(B_v^-)$ form an out-forest (in-forest). Let  $F'$ be an arbitrary subset of $A(B_v^+)\cap A(B_v^-)$ containing exactly  $\ell$  arcs. Take $X=A(B_v^+)\setminus F'$ and $Y=A(B_v^-)\setminus F'$. Observe that $X$ and $Y$ need not be disjoint. Clearly,  $|X|=|Y|=n-1-\ell$. 
 
   
   In matroid $M$, one can associate with $D'$ an independent set $I_{D'}$ of size $4n-4$ as follows: 
   \[I_{D'}= \bigcup_{e \in F'} \{e_1,e_2,e_3,e_4\}   \bigcup_{e\in X}  \{e_1,e_3\} \bigcup_{e \in Y} \{e_2,e_4\}.\]
   By our construction,  we have that $I_{D'}$ is an independent set in $\cal I$ and  $|I_{D'}|=4\ell+4(n-1-\ell)=n'$. Let  
   $F= \bigcup_{e \in F'} \{e_1,e_2,e_3,e_4\}$, $\bar{X}= \bigcup_{e\in X}  \{e_1,e_3\}$ and $\bar{Y}=\bigcup_{e \in Y} \{e_2,e_4\}$. Then notice that $F  \in  {\cal B}^{4\ell} $ and $F\subset I_{D'}$.  
   This implies that there exists a set  $\widehat{F} \in \widehat{\cal B}^{4\ell}\subseteq_{rep}^{n'-4\ell}  {\cal B}^{4\ell}$ such that 
   $I_{\bar{D}}=\widehat{F} \cup \bar{X}\cup \bar{Y}\in \cal I$.  
   We 
   show that  $D$ has a strong spanning subdigraph $\bar{D}$ with $A_{\widehat{F}}\subseteq A(\bar{D})$. Let 
   $\bar{D}$ be the digraph with the vertex set $V(D)$ and the arc set $A(\bar{D})=X\cup Y \cup A_{\widehat{F}}$. Consider the following four sets. 
   \begin{enumerate}
   \item Let $W_1=\{e_1~|~e \in X \cup  A_{\widehat{F}} \}$ then we have that $W_1\subseteq I_{\bar{D}}$ and thus $W_1\in {\cal I}_1$. This together with     the fact that $|W_1|=n-1$ implies that $ X \cup A_{\widehat{F}}$ forms a spanning tree in $U(D)$. 
  \item Let $W_2=\{e_2~|~e \in Y \cup  A_{\widehat{F}} \}$. Similar to the first case,  then 
   $Y \cup A_{\widehat{F}}$ forms a spanning tree in $U(D)$. 
   \item Let $W_3=\{e_3~|~e \in X \cup  A_{\widehat{F}} \}$   then we have that $W_3\subseteq I_{\bar{D}}$ and thus $W_3\in {\cal I}_3$. 
   This together with the fact that $|W_1|=|W_3|=n-1$ and that $X \cup  A_{\widehat{F}}$ is a a spanning tree in $U(D)$ implies that
  $X \cup  A_{\widehat{F}}$ forms an out-branching rooted at $s$ in $D_s^+$. 
   
  \item Let $W_4=\{e_3~|~e \in Y \cup  A_{\widehat{F}} \}$.  
   Similar to the previous case,    then $Y \cup  A_{\widehat{F}}$ forms an in-branching rooted at $s$ in $D_s^-$. 
\end{enumerate}
 We have shown that $\bar{D}$ contains $A_{\widehat{F}}$ and has an out-branching and in-branching rooted at $s$. This implies that 
   $\bar{D}$ is the desired strong spanning subdigraph of $D$ containing a set from $\widehat{\cal B}^{4\ell}$. This concludes the proof of the lemma.  
 \end{proof}

 \begin{lemma}
 \label{lem:medrepsetcomp}
    Let $D$ be a strong digraph on $n$ vertices and $\ell\leq n-2 $ be a natural number.  Then in time $\cO\left(\max_{i\in [\ell]} {n'  \choose 4i}^{\omega}  mn^2 \log n\right)$ we can compute 
    $\widehat{\cal B}^{4\ell}\subseteq_{rep}^{n'-4\ell}  {\cal B}^{4\ell}$ of size ${n'\choose 4 \ell}$. Here, $n'=4n-4$.
  \end{lemma}
  \begin{proof}
  We describe a dynamic programming based algorithm.   Let 
${\cal D}$ be an array of size $\ell$. The entry ${\cal D}[i]$ will store the family  
$\widehat{\cal B}^{4i}\subseteq_{rep}^{n'-4i}  {\cal B}^{4\ell}$. We 
fill the entries in the array $\cal D$ in the increasing order of its index, that is, from $0,\ldots,\ell$.  For the base case define 
$\widehat{\cal B}^{0}=\{\emptyset\}$ and let $W=\{\{e_1,e_2,e_3,e_4\}|~e \in A(D)\}$.  
Given that  ${\cal D}[i]$ is filled for all $i'\leq i$, we fill  ${\cal D}[i+1]$ as follows. Define  
${\cal N}^{4(i+1)}= \left(\widehat{\cal B}^{4i} \ \bullet W \right) \cap {\cal I}.$


\begin{claim}
\label{claim:medauxrepset}
For all $0\leq i\leq \ell-1$, 
${\cal N}^{4(i+1)} \subseteq_{rep}^{n'-4(i+1)} {\cal B}^{4(i+1)}$.
\end{claim}
\begin{proof}
Let $S \in {\cal B}^{4(i+1)}$ and  $Y$ be a set of size $n'-4(i+1)$ such that 
$S\cap Y=\emptyset$ and $S\cup Y \in {\cal I}$. We will show that there exists a set $\hat{S} \in {\cal N}^{4(i+1)}$ such that 
$\hat{S} \cap Y=\emptyset$ and $\hat{S}\cup Y \in {\cal I}$. This will imply the desired result. 

Let $e\in A(D)$ such that $\{e_1,e_2,e_3,e_4\}\subseteq S$. Define 
$S^*=S\setminus \{e_1,e_2,e_3,e_4\}$ and $Y^*=Y \cup \{ e_1,e_2,e_3,e_4\}$.  Since $S\cup Y \in \cal I$ we have that $S^*\in \cal I$ and 
$Y^*\in  \cal I$. Observe that $S^*\in {\cal B}^{4i}$,  $S^* \cup Y^*\in \cal I$ and the size of $Y^*$ is $n'-4i$. This implies that there exists 
$\widehat{S}^*$ in  $\widehat{\cal B}^{4i}\subseteq_{rep}^{n'-4i}  {\cal B}^{4\ell}$ such that $\widehat{S}^* \cup Y^*\in \cal I$. Thus 
$\widehat{S}^* \cup \{e_1,e_2,e_3,e_4\} \in \cal I$ and also in $\widehat{\cal B}^{4i} \ \bullet W $ and thus in ${\cal N}^{4(i+1)}$. Taking 
$\widehat{S}=\widehat{S}^*\cup \{e_1,e_2,e_3,e_4\}$ suffices for our purpose. This completes the proof of the claim.  
\end{proof}

We fill the entry for ${\cal D}[i+1]$ as follows. Observe that ${\cal N}_{uv}^{4(i+1)}=  ({\cal D}[i,w] \bullet W)\cap \cal I .$ 
We already have computed the family corresponding to ${\cal D}[i] $. By Theorem~\ref{thm:repsetlovasz}, 
 $|\whnd{\cal B}^{4i}|\leq {n' \choose 4i}$ and thus $|{\cal N}^{4(i+1)}|\leq  4m {n' \choose 4i}$. Furthermore,  
 we can  compute ${\cal N}^{4(i+1)}$ in time $\cO\left(mn {n' \choose 4i} \right)$. Now using Theorem~\ref{thm:repsetlovasz},   
 we can compute  $\widehat{\cal N}^{4(i+1)} \subseteq_{rep}^{n'-4(i+1)}{\cal N}^{4(i+1)}$ in time \tc{rm}{t}{4i+4}{n'-4(i+1)}, where 
 $t=4m {n' \choose 4i}$.

 By Claim~\ref{claim:medauxrepset} we know that ${\cal N}^{4(i+1)} \subseteq_{rep}^{n'-4(i+1)} {\cal B}^{4(i+1)}$. Thus Lemma~\ref{lem:reptransitive} implies that $\widehat{\cal N}^{4(i+1)} = \whnd{\cal B}^{4(i+1)}\subseteq_{rep}^{n'-4(i+1)} {\cal B}^{4(i+1)}$. We assign this family to ${\cal D}[i+1]$. This completes the description and the correctness of the dynamic programming.  The field size for uniform matroids 
 are upper bounded by $n^{\cO(1)}$ and thus we can perform all the field operations in time $\cO( \log n)$. Thus, the running time of this algorithm is upper bounded by 
 \[\cO\left( \sum_{i=1}^{\ell}  \tcwd{rm}{4m{n' \choose 4(i-1)}}{4i}{n'-4i} \right)=\cO\left(\max_{i\in [\ell]} {n'  \choose 4i}^{\omega}  m \log n\right).\]
This completes the proof.
\end{proof}
 
 \begin{lemma}
 \label{lem:msssthm}
 {\sc Minimum  SCSS} can be solved in time $\cO(2^{4\omega n}m n)$.
 \end{lemma}
  \begin{proof}
Let us fix $n'=4n-4$. Proposition~\ref{prop:two_branching} implies that  the  {\sc Minimum SCSS} problem is equivalent to finding, for an arbitrary vertex $s \in  V (D)$, an out-branching  $B_v^+$ and an in-branching $B_v^-$ that maximizes  $|A(B_v^+)\cap A(B_v^-)|$.  We guess the value of $|A(B_v^+)\cap A(B_v^-)|$ and let this be $\ell$. By Lemma~\ref{lem:medrelationrepset}, there exists a strong spanning subdigraph $D'$ of $D$ with at most $2n-2 -\ell$ 
arcs if and only if  there exists a set   $\widehat{F} \in \widehat{\cal B}^{4\ell}\subseteq_{rep}^{n'-4\ell}  {\cal B}^{4\ell}$ such that $D$ 
has a strong spanning subdigraph  $\bar{D}$ with $A_{\widehat{F}}\subseteq A(\bar{D})$. Recall that for $X\in \cal I$, 
by $A_X$ we denote the set of arcs $e\in A(D)$ such that $ \{e_1,e_2,e_3,e_4\} \cap X\neq \emptyset$.  Now using 
Lemma~\ref{lem:medrepsetcomp} we compute  $\widehat{\cal B}^{4\ell}\subseteq_{rep}^{n'-4\ell}  {\cal B}^{4\ell}$ in time 
$\cO\left(\max_{i\in [\ell]} {n'  \choose 4i}^{\omega}  m \log n\right)$.  

For every $\widehat{F} \in \widehat{\cal B}^{4\ell}$ we test whether  $A_{\widehat{F}}$ can be extended to an out-branching in $D_s^+$ and to an in-branching in $D_s^-$. We can do it in $\cO(n(n+m))$-time by putting weights $0$ to the arcs of $A_{\widehat{F}}$ and weights $1$ to all remaining arcs and then by running the classical algorithm of Edmonds \cite{Edmonds67}.  Since $\ell \leq n-2$, the running time of this algorithm is upper bounded by $\cO(2^{4\omega n}mn)$. This concludes the proof. 
 \end{proof}
  Finally, we are ready to prove the main result of this section
 \begin{theorem}
 {\sc Minimum Equivalent Graph} can be solved in time $\cO(2^{4\omega n}m n)$.
 \end{theorem}
  \begin{proof}
Given an arbitrary digraph $D$ we first find its strongly connected components $C_1,\ldots,C_s$. Now on each $C_i$, we apply  Lemma~\ref{lem:msssthm} and obtain a minimum equivalent subdigraph $C_i'$. After this we apply Proposition~\ref{prop:medredtostrong} and obtain a minimum equivalent subdigraph of $D$. Since all the steps except Lemma~\ref{lem:msssthm} takes polynomial time we get the desired running time. This completes the proof.  
\end{proof}

A weighted variant of  {\sc Minimum Equivalent Graph} has also been studied in literature. 
More precisely the problem is defined as follows.   
 \begin{center} 
\fbox{\begin{minipage}{0.96\textwidth}
\noindent{\sc Minimum Weight Equivalent Graph ({\sc MWEG})}
 \\
\noindent {\bf Input}: A directed graph $D$ and a weight function $w: A(D)\rightarrow \mathbb{N}. $\\
\noindent{\bf Task}:  Find a minimum weight equivalent subdigraph  of $D$. 
\end{minipage}}
\end{center}
\medskip

{\sc MWEG} can be solved along the same line as {\sc MEG} but to do this we need to use the notion of min $q$-representative family and use 
Theorem~\ref{thm:repsetlovaszweighted} instead of Theorem~\ref{thm:repsetlovasz}. These changes give us the following theorem. 
\begin{theorem}
 {\sc Minimum Weight Equivalent Graph} can be solved in time $\cO(2^{4\omega n}m n \log W)$. Here, $W$ is the maximum value assigned by the weight function $w: A(D)\rightarrow \mathbb{N}$. 
 \end{theorem}
\subsection{Dynamic Programming over graphs of bounded treewidth}

 
In this section we discuss deterministic algorithms for ``connectivity problems'' such as {\sc Hamiltonian Path}, {\sc Steiner Tree}, 
{\sc Feedback Vertex Set}  parameterized by the treewidth of the input graph. The algorithms are based on Theorem~\ref{thm:repsetlovasz} 
and use graphic matroids to take care of connectivity constraints. The approach is  generic and can be used whenever all the relevant information about a ``partial solution'' can be encoded as an independent set of a specific linear matroid. We exemplify the approach on the {\sc Steiner Tree} problem. 

 \medskip
\begin{center} 
\fbox{\begin{minipage}{0.96\textwidth}
\noindent{\sc Steiner Tree} \\ 
\noindent {\bf Input}: An undirected graph $G$ 
with a set of terminals $T\subseteq V(G)$, and  a  weight\\
\noindent{\phantom{{\em Input}:}} 
 function $w:E(G)\rightarrow \mathbb{N}$.\\
\noindent{\bf Task}: Find a subtree  in $G$ of minimum weight spanning all vertices of $T$. 
\end{minipage}}
\end{center}
\medskip

\subsubsection{Treewidth}
\label{subsect:twprelim}
Let $G$ be a graph.  A {\em tree-decomposition} of a graph $G$ is a pair $(\mathbb{T},\mathcal{ X}=\{X_{t}\}_{t\in V({\mathbb T})})$
such that
\begin{itemize}
\setlength\itemsep{-1mm}
\item $\cup_{t\in V(\mathbb{T})}{X_t}=V(G)$,
\item for every edge $xy\in E(G)$ there is a $t\in V(\mathbb{T})$ such that  $\{x,y\}\subseteq X_{t}$, and 
\item for every  vertex $v\in V(G)$ the subgraph of $\mathbb{T}$ induced by the set  $\{t\mid v\in X_{t}\}$ is connected.
\end{itemize}

The {\em width} of a tree decomposition is $\max_{t\in V(\mathbb{T})} |X_t| -1$ and the {\em treewidth} of $G$ 
is the  minimum width over all tree decompositions of $G$ and is denoted by $\tw(G)$. 

A tree
decomposition  $(\mathbb{T},\mathcal{ X})$ is called a {\em nice tree
decomposition} if $\mathbb{T}$ is a tree rooted at some node $r$ where
$X_{r}=\emptyset$, each node of $\mathbb{T}$ has at most two children, and each
node is of one of the following kinds:
\begin{enumerate}
\item {\bf Introduce node}: a node $t$ that has only one child $t'$ where $X_{t}\supset X_{t'}$ and  $|X_{t}|=|X_{t'}|+1$.
\item {\bf  Forget node}: a node $t$ that has only one child $t'$  where $X_{t}\subset X_{t'}$ and  $|X_{t}|=|X_{t'}|-1$.
\item {\bf Join node}:  a node  $t$ with two children $t_{1}$ and $t_{2}$ such that $X_{t}=X_{t_{1}}=X_{t_{2}}$.
\item {\bf Base node}: a node $t$ that is a leaf of $\mathbb T$, is different than the root, and $X_{t}=\emptyset$. 
\end{enumerate}
Notice that, according to the above definition, the root $r$ of $\mathbb{T}$ is
either a forget node or a join node. It is well known that any tree
decomposition of $G$ can be transformed into a nice tree decomposition
maintaining the same
width in linear time~\cite{Kloks94}. We use $G_t$ to denote the graph induced  by the
vertex set  $\cup_{t'}X_{t'}$, where $t'$ ranges over all descendants of $t$,
including $t$. By $E(X_t)$ we denote the edges present in $G[X_t]$.
We use $H_t$ to denote the graph on vertex set $V(G_t)$ and the edge set 
$E(G_t)\setminus E(X_t)$.  For clarity of presentation we use the term nodes to refer to the vertices of the tree 
$\mathbb T$.   

\subsubsection{{\sc Steiner Tree} parameterized by treewidth}
\label{subsect:steinertreetw}


Let $G$ be an input  graph of the {\sc Steiner Tree} problem. Throughout this section, we say that $E'\subseteq E(G)$ is a {\em solution} if the subgraph induced on this edge set is connected and 
it contains all the terminal vertices. We call $E'\subseteq E(G)$ an {\em optimal solution} if $E'$ is a solution of the  minimum weight. 
Let $\mathscr{S}$ be the family of edge subsets such that every edge subset corresponds to an optimal solution. That is, 
$$\mathscr{S}=\{E'\subseteq E(G)~|~E' \mbox{ is an optimal solution}\}.$$ 
We start with few definitions that will be useful in explaining the algorithm. 
Let $(\mathbb{T},\mathcal{ X})$  be a tree decomposition of $G$ of width $\tw$. Let $t$ be a node of $V(\mathbb{T})$. By $\mathcal{S}_t$ we denote the family of edge subsets  of $E(H_t)$, $\{E'\subseteq E(H_t)\}$,  that satisfies the following properties.  
\begin{itemize}
\item Either $E'$ is  a solution (that is, the subgraph formed by  this edge set is connected and  contains all the terminal vertices); or 
\item every vertex of $(T\cap V(G_t))\setminus X_t$ is incident with  some edge from  $E'$,  and every connected component of the graph  induced by $E'$  contains a  vertex from $X_t$.
\end{itemize}

We call $\mathcal{S}_t$ a \emph{family of partial solutions} for $t$. We denote by  $K^t$   a complete graph on the vertex set $X_t$. 
For an edge subset $E^* \subseteq E(G)$ and  a  bag $X_t$ corresponding to a node $t$, we define the following. 
\begin{enumerate}
\item  Set  $\partial^t(E^*)= X_t\cap V(E^*)$,  the set of endpoints  of $E^*$ in $X_t$.
\item Let $G^*$ be the subgraph of $G$ on the vertex set $V(G)$ and the edge set $E^*$. Let $C_1',\ldots ,C_\ell'$ be the connected components of  $G^*$ such that for all $i\in [\ell]$,  $C_i'\cap X_t\neq \emptyset$. Let $C_i=C_i'\cap X_t$. Observe that $C_1,\ldots,C_\ell$ is a partition of $\partial^t(E^*)$. 
By $F(E^*)$ we denote a forest $\{Q_1,\ldots,Q_\ell\}$ where each $Q_i$ is an arbitrary spanning tree of $K^t[C_i]$. For an example, since 
$K^t[C_i]$ is a complete graph we could take $Q_i$ as a star. The purpose of  $F(E^*)$ is to keep track  
for   the vertices in $C_i$ whether they are in the same connected component of $G^*$. 
\item We define $w(F(E^*))=w(E^*)$. 
\end{enumerate}


Our description of the algorithm  slightly deviates from the usual table look-up based expositions of dynamic programming algorithms on graphs of bounded treewidth.
With every node $t$ of  $\mathbb T$,   we associate a subgraph of $G$. In our case it will be $H_t$. For every node $t$, rather than keeping a table, we keep  a family of partial solutions for the graph $H_t$. That is, for every optimal solution $L\in \mathscr{S}$ and its intersection $L_t=E(H_t)\cap L$ with the graph $H_t$, we have some  partial solution in the family that is ``as good as $L_t$''. More precisely, we have some partial solution, say $\hat{L}_t$ in our family such that $\hat{L}_t\cup L_R$ is also an optimum solution for the whole graph. Here, $L_R=L\setminus L_t$. As we move from one node $t$ in the decomposition tree to the next node $t'$ the graph $H_t$ changes to $H_{t'}$, and so does the set of partial solutions. The algorithm updates its set of partial solutions accordingly.
Here matroids come into play: in order to bound the size of the family of partial solutions that the algorithm stores at each node we employ Theorem~\ref{thm:repsetlovaszweighted} for graphic matroids. More details are given in the proof of the following theorem, which is  the main result of this section.


\begin{theorem}\label{thm:steinertree_DP}
Let $G$ be an $n$-vertex graph given together with its tree decomposition of width $\tw$. Then 
{\sc Steiner Tree} on $G$ can be solved in time $\cO((1+2^{\omega+1})^{\tw} \tw^{\cO(1)}n)$.
\end{theorem}
\begin{proof}
We first outline an algorithm with running time $\cO((1+2^{\omega+1})^{\tw} \tw^{\cO(1)}n^2)$ for a simple exposition. Later we point out how we can remove the extra factor of $n$ at the cost of a factor polynomial in $\tw$.

For every node $t$ of  $\mathbb T$ and subset $Z\subseteq X_t$,  we store a family of edge subsets  $\widehat{\mathcal{S}}_t[Z]$ of $H_t$  satisfying the following correctness invariant.
\begin{quote}
{\bf Correctness Invariant:} For every $L\in \mathscr{S}$ we have the following. Let $L_t=E(H_t)\cap L$, $L_R=L\setminus L_t$,  and $Z=\partial^t(L)$. Then there exists $\hat{L}_t\in \widehat{\mathcal{S}}_t[Z]$ such that $w(\hat{L}_t)\leq w(L_t)$, $\hat{L}=\hat{L}_t\cup L_R$ is a solution, and $\partial^t(\hat{L})=Z$.   
Observe that since $w(\hat{L}_t)\leq w(L_t)$ and $L\in \mathscr{S}$, we have that $\hat{L} \in \mathscr{S}$.
\end{quote}

We process the nodes of the tree $\mathbb T$  from base nodes to the root node while doing the dynamic programming. Throughout the process we maintain the correctness invariant, which will prove the correctness of the algorithm. However, our main idea is to use representative sets  to obtain   $\widehat{\mathcal{S}}_t[Z]$ of small size. That is, given the set $\widehat{\mathcal{S}}_t[Z]$ that satisfies the correctness invariant,  we use Theorem~\ref{thm:repsetlovaszweighted} to obtain a subset $\widehat{\mathcal{S}}_t'[Z]$ of  $\widehat{\mathcal{S}}_t[Z]$ that also satisfies the  correctness invariant  and has size upper bounded by $2^{|Z|}$.  Thus, we maintain the following size invariant.

\begin{quote}
{\bf Size Invariant:} After  node $t$ of $\mathbb T$ is processed by the algorithm,  for every $Z\subseteq X_t$ we have that 
 $|\widehat{\mathcal{S}}_t[Z]|\leq2^{|Z|}$. 
\end{quote}
 
 The new ingredient of the dynamic programming algorithm for {\sc Steiner Tree} is the use of Theorem~\ref{thm:repsetlovaszweighted} to compute 
 $\widehat{\mathcal{S}}_t[Z]$ maintaining the size invariant.  The next  lemma shows how to implement  it. 

\begin{lemma}[Shrinking Lemma]
\label{lem:sizeinvariant}
Let $t$ be a node of  $\mathbb T$, and let  $Z\subseteq X_t$ be a set of size $k$. Furthermore, let  $\widehat{\mathcal{S}}_t[Z]$ be a 
family of edge subsets of $H_t$   satisfying the correctness invariant.  If $|\widehat{\mathcal{S}}_t[Z]|=\ell$, then in time 
$\cO\left(2^{k(\omega-1)} k^{\cO(1)} \ell \cdot n\right)$ we can compute $\widehat{\mathcal{S}}_t'[Z] \subseteq \widehat{\mathcal{S}}_t[Z]$ 
satisfying correctness and size invariants. 
\end{lemma}
\begin{proof}
We start by associating a matroid with  node $t$ and the set $Z\subseteq X_t$  as follows. We consider a graphic matroid \mat{} on 
$K^t[Z]$. 
Here, the element set $E$ of the matroid is  the edge set $E(K^t[Z])$ and the family of independent sets  $\cal I$ consists of spanning forests of  
$K^t[Z]$.

Let $\widehat{\mathcal{S}}_t[Z]=\{E_1^t,\ldots,E_\ell^t\}$ and let ${\cal N}=\{ F(E_1^t),\ldots,F(E_\ell^t)\}$ be the set of forests in $K^t[Z]$ corresponding to the edge subsets in $\widehat{\mathcal{S}}_t[Z]$. 
 For $i\in \{1,\ldots,k-1\}$, let ${\cal N}_i$ be the family of forests of ${\cal N}$ with $i$ edges. 
For each family ${\cal N}_i$ we apply Theorem~\ref{thm:repsetlovaszweighted} and compute its min $(k-1-i)$-representative. That is, 
$$\widehat{{\cal N}}_i \subseteq_{minrep}^{k-1-i}{\cal N}_i.$$  Let $\widehat{\mathcal{S}}_t'[Z] \subseteq \widehat{\mathcal{S}}_t[Z]$ be  such that for every $E_j^t\in \widehat{\mathcal{S}}_t'[Z] $ we have that $F(E_j^t)\in \cup_{i=1}^{k-1} \widehat{{\cal N}}_i $. By Theorem~\ref{thm:repsetlovaszweighted},   $|\widehat{\mathcal{S}}_t'[Z]| \leq \sum_{i=1}^{k-1}{k \choose i}\leq 2^k$. Now we show that  $\widehat{\mathcal{S}}_t'[Z]$ maintains the correctness invariant. 

Let $L\in \mathscr{S}$ and let $L_t=E(H_t)\cap L$, $L_R=L\setminus L_t$ and $Z=\partial^t(L)$. Then there exists $E_j^t \in \widehat{\mathcal{S}}_t[Z]$ such that $w(E_j^t)\leq w(L_t)$, $\hat{L}=E_j^t \cup L_R$ is an optimal solution and $\partial^t(\hat{L})=Z$.  Consider the forest $F(E_j^t)$. Suppose its size is $i$, then  $F(E_j^t)\in {\cal N}_i$. Now let $F(L_R)$ be the forest corresponding to $L_R$ with respect to the bag $X_t$. Since $\hat{L}$ is a solution, we have that $F(E_j^t)\cup F(L_R)$ is a spanning tree in  $K^t[Z]$. Since   $\widehat{{\cal N}}_i \subseteq_{minrep}^{k-1-i}{\cal N}_i$, we have that there exists a forest $F(E_h^t) \in \widehat{{\cal N}}_i $ such that $w(F(E_h^t)) \leq w(F(E_i^t)) $ and $F(E_h^t) \cup F(L_R)$ is a spanning tree in  $K^t[Z]$. 
Thus, we know that $E_h^t \cup L_R$ is an optimum solution and $E_h^t \in \widehat{\mathcal{S}}_t'[Z]$. This proves that $\widehat{\mathcal{S}}_t'[Z]$ maintains the invariant.  

The running time to compute $\widehat{\mathcal{S}}_t[Z]$ is dominated by:
\begin{eqnarray*}
\cO\left(\sum_{i=1}^{k-1}   \binom{k-1}{i}^{\omega-1} k^{\cO(1)} \ell \right)= \cO\left( 2^{k(\omega-1)} k^{\cO(1)} \ell \right).
\end{eqnarray*}
For a given edge set we also need to compute the forest and that can take $\cO(n)$ time. 
\end{proof}

In our algorithm the size of $\widehat{\mathcal{S}}_t[Z]$ can grow larger than $2^{|Z|}$ in intermediate steps but it will be at most 
$4^{|Z|}$ and thus we can use Shrinking Lemma (Lemma~\ref{lem:sizeinvariant}) to reduce its size efficiently.

 We now return to the dynamic programming algorithm over the tree-decomposition   $(\mathbb{T},\mathcal{ X})$ of $G$ and prove that it maintains the correctness invariant. We assume that $(\mathbb{T},\mathcal{ X})$ is a nice tree-decomposition of $G$. 
  By $\widehat{\mathcal{S}}_t$ we denote $\bigcup_{Z\subseteq X_t} \widehat{\mathcal{S}}_t[Z]$ (also called a \emph{representative family of partial solutions}). 
 We show how  
 $\widehat{\mathcal{S}}_t$  is obtained by doing dynamic programming from base node to the root node.

\paragraph{Base node $t$.}  Here the graph $H_t$ is empty and thus we take $\widehat{\mathcal{S}}_t=\emptyset$.  

\paragraph{Introduce node $t$ with child $t'$.} Here, we know that $X_{t}\supset X_{t'}$ and  $|X_{t}|=|X_{t'}|+1$. Let $v$ be the vertex in 
 $X_{t}\setminus X_{t'}$. Furthermore observe that $E(H_t)=E(H_{t'})$ and $v$ is degree zero vertex in $H_t$. Thus the graph $H_t$ only differs from $H_{t'}$ at a isolated vertex $v$. Since we have not added any edge to the new graph, the family of solutions, which contains edge-subsets, does not change. Thus, we take $\widehat{\mathcal{S}}_t=\widehat{\mathcal{S}}_{t'}$. Formally,  we take $\widehat{\mathcal{S}}_t[Z]=\widehat{\mathcal{S}}_{t'}[Z\setminus \{v\}]$.   
 Since, $H_t$ and $H_{t'}$ have  same set of edges the invariant is vacuously maintained. 
 
 \paragraph{Forget node $t$ with child $t'$.} Here we know $X_{t}\subset X_{t'}$ and  $|X_{t}|=|X_{t'}|-1$. Let $v$ be the vertex in 
 $X_{t'}\setminus X_{t}$. Let ${\cal E}_v[Z]$ denote the set of edges between $v$ and the vertices in $Z\subseteq X_t$. 
 Let ${\cal P}_v[Z]=\{Y\;|\; \emptyset \neq Y\subseteq {\cal E}_v[Z]\}$. 
 Observe that $E(H_t)=E(H_{t'})\cup {\cal E}_v[X_t]$. Before we define  things formally, observe that in this step the graphs $H_t$ and $H_{t'}$ differ by at most $\tw$ edges - the
edges with one endpoint in $v$ and the other in $X_t$. We go through every possible way an optimal solution can intersect with these newly added edges. 
The idea is that for every edge subset in our family of partial solutions we make several new partial solutions, one each for every  subset of newly added edges.  
More formally the new set of partial solutions is defined as follows. 
 
 \[ \widehat{\mathcal{S}}_t[Z] = \left\{
  \begin{array}{l l}
     \left(\widehat{\mathcal{S}}_{t'}[Z\cup \{v\}]\circ  {\cal P}_v[Z]\right)\cup \left\{A\in \widehat{\mathcal{S}}_{t'}[Z\cup
\{v\}] : A\in {\mathcal S}_t \right \} & \quad \text{if } v\in T \vspace*{2mm}\\
    \left(\widehat{\mathcal{S}}_{t'}[Z\cup \{v\}]\circ  {\cal P}_v[Z]\right)\cup \left\{A\in \widehat{\mathcal{S}}_{t'}[Z\cup
\{v\}] : A\in {\mathcal S}_t \right\}\cup \widehat{\mathcal{S}}_{t'}[Z] & \quad \text{if } v\notin T 
  \end{array} \right.\]

Recall that for two families ${\cal A}$ and  ${\cal B}$, we defined  ${\cal A} \circ {\cal B} = \{A \cup B ~:~A \in {\cal A} \wedge B \in {\cal B}\}.$ 
Now we claim that  $\widehat{\mathcal{S}}_t[Z]\subseteq {\mathcal S}_t$. Towards the proof we first show that 
$\widehat{\mathcal{S}}_{t'}[Z\cup \{v\}]\circ  {\cal P}_v[Z]\subseteq {\mathcal S}_t$. Let $E'\in \widehat{\mathcal{S}}_{t'}[Z\cup \{v\}]\circ  {\cal P}_v[Z]$. 
Note that $E'\cap {\cal E}_v[Z] \neq \emptyset$. If $E'$ is a solution tree then $E'\in {\mathcal S}_t$ and we are done. Since $E'\setminus {\cal E}_v[Z]\in
\widehat{\mathcal{S}}_{t'}[Z\cup \{v\}]\subseteq {\mathcal{S}}_{t'}$, every vertex of $(T\cap V(G_t))\setminus (X_t\cup \{v\})$ is incident with  some edge from  $E'$. Since
$E'\cap {\cal E}_v[Z] \neq \emptyset$, there exists an edge 
in $E'$ which is incident to $v$. This implies that every vertex of $(T\cap V(G_t))\setminus X_t$ is incident with  some edge from  $E'$. Now consider any connected 
component $C$ in $G[E']$. If $v\notin V(C)$, then $C$ contains a vertex from $X_{t'}\setminus \{v\}=X_t$, because $E'\setminus {\cal E}_v[Z]\in \widehat{\mathcal{S}}_{t'}[Z\cup
\{v\}] \subseteq {\mathcal S}_{t'}$. If $v\in V(C)$, then $C$ contains a vertex from $X_{t}$ because $E'\cap {\cal E}_v[Z]\neq \emptyset$. Thus we have shown that 
$E'\in {\mathcal S}_{t}$. It is easy to see that $\{A\in \widehat{\mathcal{S}}_{t'}[Z\cup \{v\}] : A\in {\mathcal S}_t \}\subseteq {\mathcal S}_t$.
If $v\notin T$ then $\widehat{\mathcal{S}}_{t'}[Z]\subseteq {\mathcal S}_t$, because  $\widehat{\mathcal{S}}_{t'}[Z]\subseteq {\mathcal S}_{t'}$ and
$X_t=X_{t'}\setminus \{v\}$.

 Now we show that $\widehat{\mathcal{S}}_t$ maintains the invariant of the algorithm. Let $L\in \mathscr{S}$.
 
  \begin{enumerate}
  \label{eqn:twdpone}
    \item  Let $L_t=E(H_t)\cap L$ and $L_R=L\setminus L_t$. Furthermore, edges of $L_t$ can be partitioned into 
    $L_{t'}=E(H_{t'})\cap L$ and $L_v=  L_t\setminus L_{t'}$.  That is, $L_t= L_{t'} \uplus L_v$. 
    \item Let $Z=\partial^t(L)$ and $Z'=\partial^{t'}(L)$. 
    \end{enumerate}

 
 
 By the property of $\widehat{\mathcal{S}}_{t'}$, there exists a $\hat{L}_{t'}\in \widehat{\mathcal{S}}_{t'}[Z']$ such that 
 \begin{eqnarray}
 \label{eqn:twdptwo}
 L\in \mathscr{S}  & \iff &  L_{t'} \uplus L_v \uplus L_R \in \mathscr{S}  \nonumber\\
                            & \iff &  \hat{L}_{t'} \uplus L_v \uplus L_R \in \mathscr{S} 
 \end{eqnarray}
 and $\partial^{t'}(L)=\partial^{t'}(\hat{L}_{t'} \uplus L_v \uplus L_R )=Z'$.  
 
 \medskip
 
 \noindent
 We put { $\hat{L}_t=\hat{L}_{t'} \cup L_v$ and $\hat{L}=\hat{L}_t \cup L_R $.}  We know show that $\hat{L}_t \in \widehat{\mathcal{S}}_t[Z]$. 
 Towards this just note that since $Z'=Z$ or $Z'=Z\cup\{v\}$, we have that $\widehat{\mathcal{S}}_t[Z]$ contains $\widehat{\mathcal{S}}_{t'}[Z']\circ \{L_v\}$. By
\eqref{eqn:twdptwo}, $\hat{L} \in \mathscr{S} $. Finally, we need to show that   $\partial^{t}(\hat{L})=Z$.  Towards this just note that $\partial^{t}(\hat{L})=Z'\setminus
\{v\}=Z$.  This concludes the proof for the fact that $\widehat{\mathcal{S}}_t$ maintains the correctness invariant. 
 

 

 \paragraph{Join node $t$ with two children $t_{1}$ and $t_{2}$.} Here, we know that  $X_{t}=X_{t_{1}}=X_{t_{2}}$.  Also we know that the edges of $H_t$ is obtained by the union of edges of $H_{t_1}$ and $H_{t_2}$ which are disjoint. Of course they are separated by the vertices in $X_t$. A natural way to obtain a family of partial solutions for $H_t$ is that we take the union of edges subsets of the families stored at nodes $t_1$ and $t_2$. This is exactly what we do. Let 
    $$ \widehat{\mathcal{S}}_t[Z]= \widehat{\mathcal{S}}_{t_1}[Z]\circ  \widehat{\mathcal{S}}_{t_2}[Z].$$ 
   
    Now we show that $\widehat{\mathcal{S}}_t$ maintains the invariant. Let $L\in \mathscr{S}$. 
    \begin{enumerate}
    \item Let $L_t=E(H_t)\cap L$ and $L_R=L\setminus L_t$. Furthermore edges of $L_t$ can be partitioned into those belonging to $H_{t_1}$ and those belonging to  $H_{t_2}$. Let  $L_{t_1}=E(H_{t_1})\cap L$ and $L_{t_2}=E(H_{t_2})\cap L$. Observe that since 
    $E(H_{t_1})\cap E(H_{t_2})=\emptyset$,  we have that $L_{t_1} \cap L_{t_2}=\emptyset$. Also observe that $L_t=L_{t_1}\uplus L_{t_2}$. 
    \item  Let $Z=\partial^t(L)$. Since $X_{t}=X_{t_{1}}=X_{t_{2}}$  this implies  that $Z=\partial^t(L)=\partial^{t_1}(L)=\partial^{t_2}(L)$. 
    \end{enumerate}

 Now observe that 
 \begin{eqnarray*}
 L\in \mathscr{S}  & \iff &  L_{t_1} \uplus  L_{t_2} \uplus L_R \in \mathscr{S}  \\
                            & \iff &  \hat{L}_{t_1}  \uplus L_{t_2} \uplus L_R \in \mathscr{S}~~~~\mbox{(by the property of $\widehat{\mathcal{S}}_{t_1}$ we have 
                                     that  $\hat{L}_{t_1}\in \widehat{\mathcal{S}}_{t_1}[Z]$)}\\
                            & \iff &  \hat{L}_{t_1} \uplus \hat{L}_{t_2} \uplus L_R \in \mathscr{S}~~~~\mbox{(by the property of $\widehat{\mathcal{S}}_{t_2}$ we have  that  $\hat{L}_{t_2}\in \widehat{\mathcal{S}}_{t_2}[Z]$)}
 \end{eqnarray*}
 \noindent 
{We put    $\hat{L}_t=\hat{L}_{t_1} \cup \hat{L}_{t_2}$.}  
By the definition of  $\widehat{\mathcal{S}}_t[Z]$,  we have that $\hat{L}_{t_1} \cup \hat{L}_{t_2}\in \widehat{\mathcal{S}}[Z]$. The above inequalities also show that $\hat{L}=\hat{L}_t\cup L_R \in \mathscr{S}$. It remains to show  that 
$\partial^{t}(\hat{L})=Z$.  
  Since $\partial^{t_1}(L)=Z$,  we have that 
 $\partial^{t_1}(\hat{L}_{t_1}  \uplus L_{t_2} \uplus L_R)=Z$. Now since $X_{t_1}=X_{t_2}$ we have that  $\partial^{t_2}(\hat{L}_{t_1}  \uplus L_{t_2} \uplus L_R)=Z$ and thus $\partial^{t_2}(\hat{L}_{t_1}  \uplus \hat{L}_{t_2} \uplus L_R)=Z$. Finally, because $X_{t_2}=X_t$, we conclude  that 
 $\partial^{t}(\hat{L}_{t_1}  \uplus \hat{L}_{t_2} \uplus L_R)=\partial^{t}(\hat{L})=Z$. This concludes the proof of correctness invariant. 
 
 \paragraph{Root node $r$.} Here, $X_{r}=\emptyset$. We go through all the solution in $\widehat{\mathcal{S}}_r[\emptyset]$ and output the one with the 
 minimum weight.   This concludes the description of the dynamic programming algorithm.  

\paragraph{Computation of $\widehat{\mathcal{S}}_t$.}
Now we show how to implement the algorithm described above in the desired running time by making use of  Lemma~\ref{lem:sizeinvariant}. 
For our discussion let us fix a node $t$ and   $Z\subseteq X_t$ of size $k$.  
While doing dynamic programming algorithm from the base nodes to the root node we always maintain the size invariant. That is, $\widehat{\mathcal{S}}_t[Z]|\leq  2^{k}. $

\paragraph{Base node $t$.} Trivially, in this case we have $|\widehat{\mathcal{S}}_t[Z]|\leq   2^{k}$. 
\paragraph{Introduce node $t$ with child $t'$.}
Here, we have that $\widehat{\mathcal{S}}_t[Z]=\widehat{\mathcal{S}}_{t'}[Z\setminus \{v\}]$ and thus 
$|\widehat{\mathcal{S}}_t[Z]|= | \widehat{\mathcal{S}}_{t'}[Z\setminus \{v\}]| \leq     2^{k-1} \leq  2^{k}$.  
\paragraph{Forget node $t$ with child $t'$.} In this case,
\[ \widehat{\mathcal{S}}_t[Z] = \left\{
  \begin{array}{l l}
     \left(\widehat{\mathcal{S}}_{t'}[Z\cup \{v\}]\circ  {\cal P}_v[Z]\right)\cup \left\{A\in \widehat{\mathcal{S}}_{t'}[Z\cup
\{v\}] : A\in {\mathcal S}_t \right \} & \quad \text{if } v\in T \vspace*{2mm}\\
    \left(\widehat{\mathcal{S}}_{t'}[Z\cup \{v\}]\circ  {\cal P}_v[Z]\right)\cup \left\{A\in \widehat{\mathcal{S}}_{t'}[Z\cup
\{v\}] : A\in {\mathcal S}_t \right\}\cup \widehat{\mathcal{S}}_{t'}[Z] & \quad \text{if } v\notin T 
  \end{array} \right.\]
  
Observe that,  
 \begin{eqnarray*}
  \left|\widehat{\mathcal{S}}_t[Z]\right| &\leq& \left|\widehat{\mathcal{S}}_{t'}[Z\cup \{v\}]\circ  {\cal P}_v[Z]\right|+ \left|\left\{A\in \widehat{\mathcal{S}}_{t'}[Z\cup
\{v\}] : A\in {\mathcal S}_t \right\}\right| + \left| \widehat{\mathcal{S}}_{t'}[Z] \right|\\
&\leq& \left(\sum_{i=1}^{k} {k \choose i} 2^{k+1} \right) +  2^{k+1}+ 2^k=\cO(4^k).
 \end{eqnarray*}

It can happen in this case that  the size of $\widehat{\mathcal{S}}_t[Z]$ is larger than $2^k$ and thus we need to reduce the size of family.  We apply  
  Lemma~\ref{lem:sizeinvariant} and obtain $\widehat{\mathcal{S}}_t'[Z]$ that maintains the correctness and size invariants. 
We update  $\widehat{\mathcal{S}}_t[Z]=\widehat{\mathcal{S}}_t'[Z]$.

 The running time to compute $\widehat{\mathcal{S}}_t$ (that is, across all subsets of $X_t)$ is 
%
\begin{eqnarray*}
\cO\left(\sum_{i=1}^{\tw +1} \binom{\tw+1}{i} 2^{i(\omega-1)} 4^i  \cdot \tw^{\cO(1)} n \right)= \cO\left( (1+2^{\omega +1})^\tw   \cdot \tw^{\cO(1)} n \right). 
\end{eqnarray*}


\paragraph{Join node $t$ with two children $t_{1}$ and $t_{2}$.}   Here we defined  
$$ \widehat{\mathcal{S}}_t[Z]= \widehat{\mathcal{S}}_{t_1}[Z]\circ  \widehat{\mathcal{S}}_{t_2}[Z].$$ 
The size of $\widehat{\mathcal{S}}_t[Z]$ is $2^k \cdot 2^k=4^k$. Now, we apply  
  Lemma~\ref{lem:sizeinvariant} and obtain $\widehat{\mathcal{S}}_t'[Z]$ that maintains the correctness invariant and has size at most $2^k$.  
We {put } $\widehat{\mathcal{S}}_t[Z]=\widehat{\mathcal{S}}_t'[Z]$.

The running time to compute $\widehat{\mathcal{S}}_t$  is 
\begin{eqnarray*}
\cO\left(\sum_{i=1}^{\tw +1} \binom{\tw+1}{i} 4^i 2^{i(\omega -1)}  \cdot \tw^{\cO(1)} n \right)=\cO\left( (1+2^{\omega +1})^\tw   \cdot \tw^{\cO(1)} n\right). 
\end{eqnarray*}

Thus the whole algorithm takes time $\cO\left( (1+2^{\omega +1})^\tw   \cdot \tw^{\cO(1)} \cdot n^2 \right)$ as the number of nodes in a nice tree-decomposition is upper bounded by $\cO(n)$. However, observe that we do not need to compute the forests and the  associated weight at every step of the algorithm. The size of the forest is at most $\tw+1$ and we can maintain these forests across the bags during dynamic programming in time $\tw^{\cO(1)}$. This will lead to an algorithm with the claimed running time. The last remark we would like to make is that one can do better at {\bf forget node} by forgetting a single edge at a time. However, we did not try to optimize this, as the running time to compute the family of 
partial solutions at  
{\bf join node} is the most expensive operation. This completes the proof. 
\end{proof}

The approach of  Theorem~\ref{thm:steinertree_DP} can be used  to obtain single-exponential algorithms parameterized by the treewidth of an input graph for several other connectivity problems such as 
\textsc{Hamiltonian Cycle}, \textsc{Feedback Vertex Set}, and  \textsc{Connected Dominated Set}. For all these problems, checking 
whether two partial solutions can be glued together to form a global solution can be checked by testing independence in a specific graphic 
matroid. We believe that there exist interesting problems where this check corresponds to testing independence in a different class of linear matroids.


%

\subsection{Path, Trees and Subgraph Isomorprhism} 

In this section we outline algorithms for {\sc $k$-Path}, {\sc $k$-Tree} and {\sc $k$-Subgraph Isomorphism} using 
representative families.  All results in this section are based on computing representative families with respect to uniform matroids. 

\subsubsection{\sc  $k$-Path}
 The problem we study in this section is as follows. 
 
 \defparproblem{ {\sc  $k$-Path} }{An  undirected $n$-vertex  and $m$-edge graph $G$ and a  positive integer $k$.}{$k$ }
 {Does there exist a simple path of length $k$ in $G$?}
 
 \medskip
 
 We start by modifying the graph slightly. We add a new vertex, say $s$ not present in $V(G)$, to $G$ by making it adjacent to every vertex in $V(G)$. Let the modified graph be called $G'$. It is clear that $G$ has a path of length $k$ if and only if $G'$ has a path of length $k+1$ starting from $s$. 
 For ease of presentation we rename $G'$ to $G$ and the objective is to find a path of length $k+1$ starting from $s$.  
 Let \mat{} be an uniform matroid $U_{n,k+2}$ where $E=V(G)$ and ${\cal I}=\{S\subseteq V(G)~|~|S|\leq k+2\}.$  
 In this section whenever we speak about independent sets we mean independence with respect to the uniform matroid $U_{n,k+2}$ defined above. 
 For a given pair of vertices $s,v\in V(G)$, recall that we defined 
 \begin{eqnarray*}
 {\cal P}_{sv}^i& = & \Big\{X~\Big|~X\subseteq V(G),~v,s \in X, ~|X|=i \mbox{ and there is a path from $s$ to $v$ of length $i$} \\ 
   & & \hspace{1cm} \mbox{     in $G$    with all the vertices belonging to $X$}. \Big\}
  \end{eqnarray*}
The problem can be reformulated to asking whether there exists  $v\in V(G)$ such that ${\cal P}_{sv}^{k+2}$ is non-empty. Our algorithm will check whether ${\cal P}_{uv}^{k+2}$ is non-empty by computing 
$\widehat{\cal P}_{sv}^{k+2} \subseteq_{rep}^0 {\cal P}_{sv}^{k+2}$ and checking whether $\widehat{\cal P}_{sv}^{k+2}$ is non-empty. 
%
The correctness of this algorithm is as follows. 
If ${\cal P}_{sv}^{k+2}$ is non-empty then ${\cal P}_{sv}^{k+2}$ contains some set $A$ which does not intersect the empty set $\emptyset$. But then $\widehat{\cal P}_{sv}^{k+2}\subseteq_{rep}^0 {\cal P}_{sv}^{k+2}$ must also contain a set which does not intersect  with $\emptyset$, and hence $\widehat{\cal P}_{sv}^{k+2}$ must be non-empty as well. Thus, having computed the representative familes $\widehat{\cal P}_{sv}^{k+2}$ all we need to do is to check whether there is a  vertex $v$ such that  $\widehat{\cal P}_{sv}^{k+2}$ is non-empty. All that remains is an algorithm that computes the representative families  $\widehat{\cal P}_{sv}^{k+2}\subseteq_{rep}^0 {\cal P}_{sv}^{k+2}$ for all $v\in V(G)\setminus\{s\}$.
 \noindent

Now using Lemma~\ref{lem:FastLDCpathrepset} (by setting  $\ell=p=k+2$)  we compute $\widehat{{\cal P}}_{sv}^{k+2} \subseteq_{rep}^0  {\cal P}^{k+2}_{sv}$ for all  $v\in V(G)\setminus \{s\}$  
in time $$2^{o(k)} \cdot m\log n  \cdot \max_{i\in[k+2]}\left\{\left(\frac{2(k+2)-i}{i}\right)^i \left(\frac{2(k+2)-i}{2(k+2)-2i}\right)^{2(k+2)-2i}\right\}.$$ 
Simple calculus shows that the running time is maximized for $i = (1-\frac{1}{\sqrt{5}})(k+2)$, 
and thus the running time to compute $\widehat{{\cal P}}_{sv}^{k+2} \subseteq_{rep}^0  {\cal P}^{k+2}_{sv}$ for all  $v\in V(G)\setminus \{s\}$ together is upper bounded by 
$\phi^{2k+o(k)}m\log^2n = \cO(2.619^{k}m\log n)$, where  
where $\phi$ is the golden ratio $\frac{1+\sqrt{5}}{2}$.
 Furthermore, in the same time every set in 
 $\widehat{{\cal P}}_{sv}^p$ can be ordered in a way that it corresponds to an undirected path in $G$. A graph $G$ has a path of length $k+1$ 
starting from $s$ if and only if for some $v\in V(G)\setminus \{s\}$, we have that $\widehat{{\cal P}}_{sv}^{k+2}\neq \emptyset$. Thus  the running time of this algorithm is upper bounded by 
$\cO(2.619^{k}m\log n)$. 
 Let us remark that almost the same arguments show that the version of the problem on directed graphs is solvable within the same running time. 
 However on undirected graphs we can speed up the algorithm slightly by using the following standard trick. We need the following result. 
 \begin{proposition}[\cite{Bodlaender93a}]
 \label{prop:bodpath}
 There exists an algorithm, that given a graph $G$ and an integer $k$, in time $\cO(k^2n)$ either finds a simple path of length $\geq k$ or computes  a  DFS (depth first search) tree rooted at some vertex of $G$ of depth at most $k$. 
 \end{proposition}
 
 We  first apply Proposition~\ref{prop:bodpath} and  in time $\cO(k^2n)$ either find a simple path of length $\geq k$ in $G$  or compute  a
  DFS tree of $G$ of depth at most $k$. In the former case we simply output the same path. In the later case since all the root to leaf paths are upper bounded by $k$ and there are no cross edges in a DFS tree, we have that  the number of edges in $G$ is upper bounded by $
  \cO(k^2n)$. Now on this $G$ we apply the 
  representative set based algorithm described above. This results in the following theorem. 
  \begin{theorem}
  {\sc  $k$-Path} can be solved in time $\cO(2.619^{k}n\log n)$.
  \end{theorem}
 
Our algorithm for  {\sc  $k$-Path} can be used to solve the weighted version of the problem, i.e,  \textsc{ Short Cheap Tour}. In this problem a  graph $G$  with maximum edge cost
$W$ is given, and the objective is to find a path of length at least $k$ where the total sum of costs on the edges is minimized. 
  \begin{theorem}
\textsc{ Short Cheap Tour} can be solved in time $\cO(2.619^{k} n^{\cO{(1)}} \log{W})$. 
  \end{theorem}
\subsubsection{{\sc $k$-Tree} and {\sc $k$-Subgraph Isomorphism}}
 
In this section we consider the following problem.

\smallskip
\defparproblem{ {\sc  $k$-Tree} }{An undirected $n$-vertex, $m$-edge graph $G$ and a tree $T$ on $k$ vertices.}{$k$ }
 {Does $G$ contains a subgraph  isomorphic to  $T$?}
\smallskip

We design an algorithm for {\sc $k$-Tree} using the method of representative sets.  The algorithm for {\sc  $k$-Tree} is more involved than for 
  {\sc  $k$-Path}. The reason to that is  due to the fact that paths poses perfectly balanced separators of size one while trees not. 
We select a leaf $r$ of $T$ and root the tree at $r$. For vertices $x$,$y \in V(T)$ we say that $y \leq x$ if $x$ lies on the path from $y$ to $r$ in $T$ (if $x = r$ we also say that  $y \leq x$). For a set $C$ of vertices in $T$ we will say that $x \preceq_C y$ if $x \leq y$ and there is no $z \in C$ such that $x \leq z$ and $z \leq y$. For a pair $x$, $y$ of vertices such that $y \leq x$ in $T$ we define 
\begin{equation*}
C^{xy} =
\begin{cases}
\emptyset & \text{if } xy \in E(T),\\
\text{The unique component $C$ of $T \setminus \{x,y\}$ such that $N(C) = \{x,y\}$ } & \text{otherwise}.
\end{cases}
\end{equation*}
We also define $T^{uv} = T[C^{uv} \cup \{u,v\}]$. We start by making a few simple observations about sets of vertices in trees.
\begin{lemma}\label{lem:treeDivide} For any tree $T$, a pair $\{x,y\}$ of vertices in $V(T)$ and integer $c \geq 1$ there exists a set $W$ of vertices such that $\{x, y\} \subseteq W$, $|W| = \cO(c)$ and every connected component $U$ of $T \setminus W$ satisfies $|U| \leq \frac{|V(T)|}{c}$ and $|N(U)| \leq 2$.
\end{lemma}
\begin{proof}
We first find a set $W_1$ of size at most $c$ such that every connected component $U$ of 
$T \setminus W_1$ satisfies $|U| \leq \frac{|V(T)|}{c}$. Start with $W_1 = \emptyset$ and select a lowermost vertex $u \in V(T)$ such that the subtree rooted at $u$ has at least  $\frac{|V(T)|}{c}$ vertices. Add $u$ to $W_1$ and remove the subtree rooted at $u$ from $T$. The process must stop after $c$ iterations since each iteration removes $\frac{|V(T)|}{c}$ vertices of $T$. Each component $U$ of $T \setminus W_1$ satisfies $|U| \leq \frac{|V(T)|}{c}$ because (a) whenever a vertex $u$ is added to $W_1$, all components below $u$ have size strictly less than $\frac{|V(T)|}{c}$ and (b) when the process ends the subtree rooted at $r$ has size at most $|U| \leq \frac{|V(T)|}{c}$. Now, insert $x$ and $y$ into $W_1$ as well.

We build $W$ from $W_1$ by taking the {\em least common ancestor closure} of $W_1$; start with $W = W_1$ and as long as there exist two vertices $u$ and $v$ in $W$ such that their least common ancestor $w$ is not in $W$, add $w$ to $W$. Standard counting arguments on trees imply that this process will never increase the size of $W$ by more than a factor $2$, hence $|W| \leq 2|W_1| = O(c)$.

We claim that every connected component $U$ of $T \setminus W$ satisfies $N(U) \leq 2$. Suppose not and let $u$ be the vertex of $u$ closest to the root. Since $N(U) > 2$ at least two vertices $v$ and $w$ in $N(U)$ are descendants of $u$. Since $U$ is connected $v$ and $w$ can't be descendants of each other, but then the least common ancestor of $v$ and $w$ is in $U$, contradicting the construction of $W$.
\end{proof}

\begin{observation}\label{obs:treeHasLeaf} For any tree $T$, set $W \subseteq V(T)$ and component $U$ of $T \setminus W$ such that $|N(U)| = 1$, $U$ contains a leaf of $T$.
\end{observation}
\begin{proof}
$T[U \cup N(U)]$ is a tree on at least two vertices and hence it has at least two leaves. At most one of these leaves is in $N(U)$, the other one is also a leaf of $T$.
\end{proof}

\begin{lemma}\label{lem:identifyComponents}
Let $W \subseteq V(T)$ be a set of vertices such that for every pair of vertices in $W$ their least common ancestor is also in $W$. Let $X$ be a set containing one leaf of $T$ from each connected component $U$ of $T \setminus W$ such that $|N(U)| = 1$. Then, for every connected component $U$ such that $|N(U)| = 1$ there exist $x \in W$, $y \in X$ such that $U = C^{xy} \cup \{y\}$. For every other connected component $U$ there exist $x$, $y \in W$ such that $U = C^{xy}$.
\end{lemma}

\begin{proof}
It follows from the argument at the end of the proof of Lemma~\ref{lem:treeDivide} that every component $U$ of $T \setminus W$ satisfies $|N(U)| \leq 2$. If $|N(U)| = 2$, let $N(U) = \{x,y\}$. We have that $x \leq y$ or $y \leq x$ since least common ancestor of $x$ and $y$ can not be in $U$ and would therefore be in $N(U)$, contradicting $|N(U)| = 2$. Without loss of generality $y \leq x$. But then $U = C^{xy}$. If $N(U) = 1$, let $N(U) = \{x\}$. By Observation~\ref{obs:treeHasLeaf} $U$ contains a leaf $y$ of $T$. Then $U = C^{xy} \cup \{y\}$.
\end{proof}

Given two graphs $F$ and $H$, a graph {\em homomorphism} from
 $F$ to $H$ is a map $f$ from $V(F)$ to $V(H)$, that is $f:~V(F)\rightarrow V(H)$, such that if $uv\in E(F)$, then $f(u)f(v)\in E(H)$. Furthermore,
when the map $f$ is injective,  
 $f$ is called a {\em subgraph isomorphism}.
 For every $x,y \in V(T)$ such that $y \leq x$, and every $u$,$v$ in $V(G)$ we define
\begin{eqnarray*}
{\cal F}^{xy}_{uv} & = & \Big\{F \in {V(G) \setminus \{u,v\} \choose |C^{xy}|}~:~ \exists \mbox{ subgraph isomorphism $f$} \\ 
& & \mbox{from $T^{xy}$ to $G[F \cup \{u,v\}]$ such that 
$f(x)=u$ and $f(y)=v$}\Big\}
\end{eqnarray*}
Let us remind that for a    set $X$ and a family  ${\cal A}$, we use   
 ${\cal A} + X$ to denote $\{A \cup X ~:~A \in {\cal A}\}.$ 
For every $x,y \in V(T)$ such that $y \leq x$, and every $u$ in $V(G)$ we define 
\begin{align}\label{eqn:fastdefn}
{\cal F}^{xy}_{u\ast} = \bigcup_{v \in V(G) \setminus \{u\}} {\cal F}^{xy}_{uv} + \{v\}
\end{align}

In order to solve the problem it is sufficient to select an arbitrary leaf $\ell$ of $T$ and determine whether there exists a $u \in V(G)$ such that the family ${\cal
F}^{r\ell}_{u\ast}$ is non-empty. We show that the collections of families $\{{\cal F}^{xy}_{uv}\}$ and  $\{{\cal F}^{xy}_{u\ast}\}$ satisfy a recurrence relation. We will then
exploit this recurrence relation to get a fast algorithm for  {\sc  $k$-Tree}.

\begin{lemma}\label{lem:ktreeRecurrence} For every $x$,$y \in V(T)$ such that $y \leq x$, every $\widehat{W} = W \cup \{x,y\}$ where $W \subseteq C^{xy}$, such that  for every pair of vertices in $\widehat{W}$ their least common ancestor is also in $\widehat{W}$, every $X \subseteq C^{xy} \setminus W$ such that $X$ contains exactly one leaf of $T$ in each connected component $U$ of $T^{xy} \setminus \widehat{W}$ with $|N(U)|=1$, the following recurrence holds.
\begin{align}\label{eqn:ktreeRecurrence}
{\cal F}^{xy}_{uv} = \bigcup_{\substack{g : \widehat{W} \rightarrow V(G) \\ g(x)=u \wedge g(y)=v}} \left[\left(\prod^\bullet_{\substack{x',y' \in \widehat{W} \\ y' \preceq_{\widehat{W}} x'}} {\cal F}^{x'y'}_{g(x')g(y')} \bullet \prod^\bullet_{\substack{x' \in\widehat{W} \mbox{ , } y' \in X \\ y' \preceq_{\widehat{W}} x'}} {\cal F}^{x'y'}_{g(x')\ast}\right) + g(W) \right]
\end{align}
Here the union goes over all $O(n^{|W|})$ injective maps $g$ from $\widehat{W}$ to $V(G)$ such that $g(x)=u$ and $g(y)=v$, and by $g(W)$ we mean $\{g(c)~:~c \in W\}$.
\end{lemma}

\begin{proof}
For the $\subseteq$ direction of the equality consider any subgraph isomorphism $f$ from $T^{xy}$ to $V(G)$ such that $f(x) = u$ and $f(y)=v$. Let $g$ be the restriction of $f$ to $W$. The map $f$ can be considered as a collection of subgraph isomorphisms with one isomorphism for each $x', y' \in \widehat{W}$ such that $y' \preceq_{\widehat{W}} x$ from $T^{x'y'}$ to $G$ such that $f(x') = g(x')$ and $f(y') = g(y')$, and  one isomorphism for each $x' \in \widehat{W}, y' \in X$ such that $y' \preceq_{\widehat{W}} x$ from $T^{x'y'}$ to $G$ such that $f(x') = g(x')$. Taking the union of the ranges of each of the small subgraph isomorphisms clearly give the range of $f$. Here we used Lemma~\ref{lem:identifyComponents} to argue that for every connected component $U$ of $T^{xy} \setminus \widehat{W}$ we have that $T[U \cup N(U)]$ is in fact on the form $T^{x'y'}$ for some $x', y'$.

For the reverse direction take any collection of subgraph isomorphisms with one isomorphism $f$ for each $x', y' \in \widehat{W}$ such that $y' \preceq_{\widehat{W}} x$ from $T^{x'y'}$ to $G$ such that $f(x') = g(x')$ and $f(y') = g(y')$, and  one isomorphism for each $x' \in \widehat{W}, y' \in X$ such that $y' \preceq_{\widehat{W}} x$ from $T^{x'y'}$ to $G$ such that $f(x') = g(x')$, such that the range of all of these subgraph isomorphisms are pairwise disjoint (except on vertices in $\widehat{W}$). Since all of these subgraph isomorphisms agree on the set $W$ they can be glued together to a subgraph isomorphism from $T^{xy}$ to $G$.
\end{proof}

Our goal is to compute for every $x, y \in V(T)$ such that $y \leq x$ and $u, v \in V(G)$ a family $\hat{\cal F}^{xy}_{uv}$ such that $\hat{\cal F}^{xy}_{uv} \subseteq_{rep}^{k-|C^{xy}|} {\cal F}^{xy}_{uv}$ and for every $x, y \in V(T)$ such that  $y \leq x$ and $u \in V(G)$ a family $\hat{\cal F}^{xy}_{u\ast}$ such that $\hat{\cal F}^{xy}_{u\ast} \subseteq_{rep}^{k-|C^{xy}|-1} {\cal F}^{xy}_{uv}$. 
We will also maintain the following size invariants.
\begin{align}
 |\hat{\cal F}^{xy}_{uv}| &\leq 
{\left(\frac{2k-|C^{xy}|}{|C^{xy}|}\right)^{|C^{xy}|} \left(\frac{2k-|C^{xy}|}{2k-2|C^{xy}|}\right)^{k-|C^{xy}|}}2^{o(k)}
\label{eqn:sxy}\\
 |\hat{\cal F}^{xy}_{u\ast}| &\leq 
{\left(\frac{2k-|C^{xy}|-1}{|C^{xy}|+1}\right)^{|C^{xy}|+1} \left(\frac{2k-|C^{xy}|-1}{2k-2|C^{xy}|-2}\right)^{k-|C^{xy}|-1}}2^{o(k)}
\label{eqn:sxystar}
\end{align}
Let the right hand side of equation~\ref{eqn:sxy} be $s_{xy}$ and the right had side of equation~\ref{eqn:sxystar} be $s_{xy}^{*}$.
%
We first compute such families $\hat{\cal F}^{xy}_{uv}$ for all $x, y \in V(T)$ such that $y \leq x$ and $xy \in E(T)$. Observe that in this case we have
\begin{equation*}
{\cal F}^{xy}_{uv} =
\begin{cases}
\{\emptyset\} & \text{if } uv \in E(G),\\
\emptyset & \text{if } uv \notin E(G).
\end{cases}
\end{equation*}
For each  $x, y \in V(T)$ such that $y \leq x$ and $xy \in E(T)$ and every $u,v \in V(G)$ we set  $\hat{\cal F}^{xy}_{uv} =  {\cal F}^{xy}_{uv}$. We can now for compute  $\hat{\cal F}^{xy}_{u\ast}$ for every $x, y \in V(T)$ such that $y \leq x$ and $xy \in E(T)$ and every $u \in V(G)$ by applying Equation~\ref{eqn:fastdefn}.  Clearly the computed families are within the required size bounds.

We now show how to compute a family $\hat{\cal F}^{xy}_{uv}$ of size 
$s_{xy}$ for every $x, y \in V(T)$ such that $y \leq x$ and $u, v \in V(G)$ and $|C^{xy}| = t$, assuming that the families $\hat{\cal F}^{xy}_{uv}$ and $\hat{\cal F}^{xy}_{u\ast}$ have been computed for every $x, y \in V(T)$ such that $y \leq x$ and $u, v \in V(G)$ and $|C^{xy}| < t$. We also assume that for each family $\hat{\cal F}^{xy}_{uv}$  that has been computed, 
$|\hat{\cal F}^{xy}_{uv}| \leq s_{xy}$. 
Similarly we assume that for each family $\hat{\cal F}^{xy}_{u\ast}$  that has been computed,  
$|\hat{\cal F}^{xy}_{u\ast}| \leq s_{xy}^{*}$.

We fix a constant $c$ whose value will be decided later. First apply Lemma~\ref{lem:treeDivide} on $T^{xy}$, vertex pair $\{x,y\}$ and constant $c$ and obtain a set $\widehat{W}$  such that $\{x, y\} \subseteq \widehat{W}$ and every connected component $U$ of $T \setminus \widehat{W}$ satisfies $|U| \leq \frac{|V(T)|}{c}$ and $|N(U)| \leq 2$. Select a set $X \subseteq V(T^{x,y}) \setminus \widehat{W}$ such that each connected component $U$ of $T \setminus\widehat{W}$ with $|N(U)| = 1$ contains exactly one leaf which is in $X$. Now, set $W = \widehat{W} \setminus \{x,y\}$ and consider Equation~\ref{eqn:ktreeRecurrence} for $\hat{\cal F}^{xy}_{uv}$ for this choice of $x$,$y$,$W$ and $X$. Define
\begin{align}\label{eqn:ktreeHatRecurrence}
\tilde{\cal F}^{xy}_{uv} = \bigcup_{\substack{g : \widehat{W} \rightarrow V(G) \\ g(x)=u \wedge g(y)=v}} \left[\left(\prod^\bullet_{\substack{x',y' \in \widehat{W} \\ y' \preceq_{\widehat{W}} x'}} \hat{\cal F}^{x'y'}_{g(x')g(y')} \bullet \prod^\bullet_{\substack{x' \in\widehat{W} \mbox{ , } y' \in X \\ y' \preceq_{\widehat{W}} x'}} \hat{\cal F}^{x'y'}_{g(x')\ast}\right) + g(W) \right]
\end{align}
Lemma~\ref{lem:ktreeRecurrence} together with Lemmata~\ref{lem:repunion} and \ref{lem:repconvolution} directly imply that $\tilde{\cal F}^{xy}_{uv} \subseteq_{rep}^{k-|C^{xy}|} {\cal F}^{xy}_{uv}$. Furthermore, each family on the right hand side of Equation~\ref{eqn:ktreeHatRecurrence} has already been computed, since $C^{x'y'} \subset C^{xy}$ and so $|C^{x'y'}| < t$.
For a fixed injective map $g : W \rightarrow V(G)$ we define 
\begin{align}\label{eqn:fixedFunction}
\tilde{\cal F}^{xy}_{g} = \left(\prod^\bullet_{\substack{x',y' \in \widehat{W} \\ y' \preceq_{\widehat{W}} x'}} \hat{\cal F}^{x'y'}_{g(x')g(y')} \bullet \prod^\bullet_{\substack{x' \in\widehat{W} \mbox{ , } y' \in X \\ y' \preceq_{\widehat{W}} x'}} \hat{\cal F}^{x'y'}_{g(x')\ast}\right) + g(W) 
\end{align}
It follows directly from the definition of $\tilde{\cal F}^{xy}_{uv}$ and $\tilde{\cal F}^{xy}_{g}$ that
\begin{align*}
\tilde{\cal F}^{xy}_{uv} = \bigcup_{\substack{g : \widehat{W} \rightarrow V(G) \\ g(x)=u \wedge g(y)=v}} \tilde{\cal F}^{xy}_{g}.
\end{align*}

Our goal is to compute a family $\hat{\cal F}^{xy}_{uv} \subseteq_{rep}^{k-|C^{xy}|} \tilde{\cal F}^{xy}_{uv}$ such that  
$|\hat{\cal F}^{xy}_{uv}| \leq s_{xy}$. 
 Lemma~\ref{lem:reptransitive} then implies that $\hat{\cal F}^{xy}_{uv} \subseteq_{rep}^{k-|C^{xy}|} {\cal F}^{xy}_{uv}$. To that end, we define the function {\sf reduce}. Given a family ${\cal F}$ of sets of size $p$, the function {\sf reduce} will run the algorithm of 
Theorem~\ref{thm:repset uniform general3}
 on ${\cal F}$ with $x=\frac{p}{2k-p}$ and produce a family of size 
$\left(\frac{2k-p}{p}\right)^p \left(\frac{2k-p}{2k-2p}\right)^{k-p} 2^{o(k)}$ 
that $k-p$ represents ${\cal F}$ .

We will compute for each $g : \widehat{W} \rightarrow V(G)$ such that $g(x)=u$ and $g(y)=v$ a family $\hat{\cal F}^{xy}_{g}$ of size at most 
$s_{xy}$
such that $\hat{\cal F}^{xy}_{g} \subseteq_{rep}^{k-|C^{xy}|} \tilde{\cal F}^{xy}_{g}$. We will then set
\begin{align}
\hat{\cal F}^{xy}_{uv} = {\sf reduce}\left(\bigcup_{\substack{g : \widehat{W} \rightarrow V(G) \\ g(x)=u \wedge g(y)=v}} \hat{\cal F}^{xy}_{g}\right).
\end{align}
To compute $\hat{\cal F}^{xy}_{g}$, inspect Equation~\ref{eqn:fixedFunction}. Equation~\ref{eqn:fixedFunction} shows that $\tilde{\cal F}^{xy}_{g}$ basically is a long chain of $\bullet$ operations, specifically
\begin{align}\label{eqn:fixedFunction2}
\tilde{\cal F}^{xy}_{g} = \left(\hat{F}_1 \bullet \hat{F}_2 \bullet \hat{F}_3 \ldots \bullet \hat{F}_\ell\right) + g(W) 
\end{align}
We define (and compute) $\hat{\cal F}^{xy}_{g}$ as follows
\begin{align}\label{eqn:computeHatG}
\hat{\cal F}^{xy}_{g} = {\sf reduce}\left({\sf reduce}\left(\ldots {\sf reduce}\left({\sf reduce}\left(\hat{F}_1 \bullet \hat{F}_2\right)  \bullet \hat{F}_3\right) \bullet \ldots \right) \bullet \hat{F}_\ell\right) + g(W) 
\end{align}
$\hat{\cal F}^{xy}_{g} \subseteq_{rep}^{k-|C^{xy}|} \tilde{\cal F}^{xy}_{g}$ and thus also $\hat{\cal F}^{xy}_{uv} \subseteq_{rep}^{k-|C^{xy}|} \tilde{\cal F}^{xy}_{uv} \subseteq_{rep}^{k-|C^{xy}|} {\cal F}^{xy}_{uv}$ follows from Lemma~\ref{lem:repconvolution} and 
Theorem~\ref{thm:repset uniform general3}. 
Since the last operation we do in 
the construction of $\hat{\cal F}^{xy}_{uv}$ is a call to {\sf reduce}, 
$|\hat{\cal F}^{xy}_{uv}| \leq s_{xy}$  
follows from Theorem~\ref{thm:repset uniform general3}.  
To conclude the computation we set 
\begin{align}\label{eqn:computeHatAst}
\tilde{\cal F}^{xy}_{u\ast} = {\sf reduce}\left(\bigcup_{v \in V(G) \setminus \{u\}} \hat{\cal F}^{xy}_{uv} + \{v\}\right)
\end{align}
Lemma~\ref{lem:repconvolution} and 
Theorem~\ref{thm:repset uniform general3} 
imply that $\tilde{\cal F}^{xy}_{u\ast}  \subseteq_{rep}^{k-|C^{xy}|-1} {\cal F}^{xy}_{u\ast}$ and that 
$|\hat{\cal F}^{xy}_{u\ast}| \leq s_{xy}^*$.

The algorithm computes the families $\hat{\cal F}^{xy}_{u\ast}$ and $\hat{\cal F}^{xy}_{uv}$ for every $x,y \in V(T)$ such that $y \leq x$. It then selects an arbitrary leaf $\ell$ of $T$ and checks whether there exists a $u \in V(G)$ such that the family $\hat{\cal F}^{r\ell}_{u\ast}$ is non-empty. Since $\hat{\cal F}^{r\ell}_{u\ast} \subseteq_{rep}^{0} {\cal F}^{r\ell}_{u\ast}$ there is a non-empty ${\cal F}^{r\ell}_{u\ast}$ if and only if there is a non empty  $\hat{\cal F}^{r\ell}_{u\ast}$. Thus the algorithm can answer that there is a subgraph isomorphism from $T$ to $G$ if some $\hat{\cal F}^{r\ell}_{u\ast}$ is non-empty, and that no such subgraph isomorphism exists otherwise.

It remains to bound the running time of the algorithm. Up to polynomial factors, the running time of the algorithm is dominated by the computation of $\hat{\cal F}^{xy}_{uv}$. This computation consists of $n^{\cO(\widehat{|W|})}$ independent computations of the families $\hat{\cal F}^{xy}_{g}$. Each computation of the family $\hat{\cal F}^{xy}_{g}$ consists of at most $k$ repeated applications of the operation
$$\hat{\cal F}^{i+1} = {\sf reduce}(\hat{\cal F}^i \bullet \hat{\cal F}_{i+1}).$$
Here ${\cal F}^i$ is a family of sets of size $p$, and so 
$|{\cal F}^i| \leq {\left(\frac{2k-p}{p}\right)^{p} \left(\frac{2k-p}{2k-2p}\right)^{k-p}}2^{o(k)}\log n $.
On the other hand $\hat{\cal F}_{i+1}$ is a family of sets of size $p' \leq \frac{k}{c}$ since we used Lemma~\ref{lem:treeDivide} to construct $\widehat{W}$. Thus,
\begin{align*}
|\hat{\cal F}_{i+1}| &\leq {\left(\frac{2k-p'}{p'}\right)^{p'} \left(\frac{2k-p'}{2k-2p'}\right)^{k-p'}}2^{o(k)}\\
&\leq {\left(\frac{2k}{p'}\right)^{p'} \left(\frac{2k}{2k-2p'}\right)^{k-p'}}2^{o(k)}\\ 
&\leq {k \choose p'}\cdot 2^{p'}\cdot 2^{o(k)}\\
&\leq {k \choose k/c}\cdot 2^{k/c}\cdot 2^{o(k)}\\
&\leq 2^{(\varepsilon+1/c) k} \cdot 2^{o(k)}
\end{align*}
Thus $|\hat{\cal F}^i \bullet \hat{\cal F}_{i+1}| \leq 
{\left(\frac{2k-p}{p}\right)^{p} \left(\frac{2k-p}{2k-2p}\right)^{k-p}}2^{(\varepsilon+1/c) k + o(k)}$. Hence, when we apply Theorem~\ref{thm:repset uniform general3} with $x=\frac{p+p'}{2k-p-p'}$ to compute ${\sf reduce}(\hat{\cal F}^i \bullet \hat{\cal F}_{i+1})$, this takes time

\begin{eqnarray*}
&& |\hat{\cal F}^i \bullet \hat{\cal F}_{i+1}| \left(\frac{2k-p-p'}{2k-2p-2p'}\right)^{k-p-p'}2^{o(k)}  \log n\\
&\leq&|\hat{\cal F}^i \bullet \hat{\cal F}_{i+1}| \left(\frac{2k-p}{2k-2p}\right)^{k-p} \left(\frac{2k-2p}{2k-2p-2p'}\right)^{k-p-p'}2^{o(k)}  \log n\\
&\leq&|\hat{\cal F}^i \bullet \hat{\cal F}_{i+1}| \left(\frac{2k-p}{2k-2p}\right)^{k-p} \left(1+\frac{p'}{k-p-p'}\right)^{k-p-p'}2^{o(k)} \log n\\
&\leq&|\hat{\cal F}^i \bullet \hat{\cal F}_{i+1}| \left(\frac{2k-p}{2k-2p}\right)^{k-p} e^{p'}2^{o(k)}  \log n\\
&\leq& {\left(\frac{2k-p}{p}\right)^{p} \left(\frac{2k-p}{2k-2p}\right)^{2k-2p}}2^{(\varepsilon+3/c) k + o(k)}\log n\\
\end{eqnarray*}
Since there are $n^{\cO(\widehat{|W|})}$ (which is equal to $n^{\cO(c)}$, where $c$ is a constant) independent computations of the families $\hat{\cal F}^{xy}_{g}$, the total running time is upper bounded by  
$$
{\left(\frac{2k-p}{p}\right)^{p} \left(\frac{2k-p}{2k-2p}\right)^{2k-2p}}2^{(\varepsilon+3/c) k + o(k)} n^{\cO(1)}
$$
The maximum value of ${\left(\frac{2k-p}{p}\right)^{p} \left(\frac{2k-p}{2k-2p}\right)^{2k-2p}}$ is when $p=(1-\frac{1}{\sqrt{5}})k$ and the maximum value is 
$\phi^{2k}$, where $\phi$ is the golden ratio $\frac{1+\sqrt{5}}{2}$.
Now we can choose the value of $c$ in such a way that $\varepsilon+3/c$ is small enough and the above running time is bounded by $2.619^{k}n^{\cO(1)}$. 
This yields the following theorem.
\begin{theorem}\label{thm:k-tree}
{\sc $k$-Tree} can be solved in time $2.619^{k}n^{\cO(1)}$.
\end{theorem}

The algorithm for {\sc $k$-Tree}  can be generalized  to {\sc $k$-Subgraph Isomorphism} for the case when the pattern graph $F$ has treewidth at most $t$. Towards this we need a result analogous to Lemma~\ref{lem:treeDivide} for trees, which can be proved using the separation properties of graphs of treewidth at most $t$.  This will lead to an algorithm with running time  $2.619^{k} \cdot n^{\cO{(t)}}$.

\subsection{Other Applications}
 

 Marx~\cite{Marx09} gave algorithms for several problems  based on matroid optimization. The main theorem in his work is 
 Theorem 1.1~\cite{Marx09} on which most applications of~\cite{Marx09} are based. The proof of the theorem uses an algorithm to find representative sets as a black box. Applying our algorithm (Theorem~\ref{thm:repsetlovasz} of this paper) instead gives an improved version of 
 Theorem 1.1 of \cite{Marx09}. 

\begin{proposition}
\label{prop:marxmainresult}
Let \mat{} be a linear matroid where the ground set is partitioned into blocks of size $\ell$. Given a linear representation $A_M$ of $M$, it can be determined in $\cO(2^{\omega k\ell} ||A_M||^{\cO(1)})$ randomized time whether there is an independent set that is the union of $k$ blocks. ($||A_M||$ denotes 
the length of $A_M$ in the input.)
\end{proposition}

Finally, we mention another application from~\cite{Marx09} which we believe could be useful to obtain single exponential time 
parameterized and exact algorithms.  

 \medskip
\begin{center} 
\fbox{\begin{minipage}{0.96\textwidth}
\noindent{\sc  $\ell$-Matroid Intersection} \hfill {\bf Parameter:} $k$ \\
\noindent {\bf Input}: Let $M_1=(E,{\cal I}_1),\dots, M_1=(E,{\cal I}_\ell) $ be matroids on the same ground set $E$  given \\
 \noindent{\phantom{{\em Input}:}} by their representations $A_{M_1},\ldots, A_{M_{\ell}}$ over the same field $\mathbb{F}$ and a positive integer $k$.\\
\noindent{\bf Question}: Does there exist $k$ element set that is independent in each $M_i$ ($X\in {\cal I}_1\cap \ldots  \cap {\cal I}_\ell$)?  
\end{minipage}}
\end{center}
\medskip
 
 Using Theorem 1.1 of ~\cite{Marx09}, Marx~\cite{Marx09} gave a randomized algorithm for {\sc  $\ell$-Matroid Intersection}. By using 
Proposition~\ref{prop:marxmainresult} instead we get  the following result. 
\begin{proposition}
{\sc  $\ell$-Matroid Intersection} can be solved in $\cO(2^{\omega k\ell} ||A_M||^{\cO(1)})$ randomized time. 
\end{proposition}

\section{Conclusion and Recent Developments}\label{sec:conclusion}

In this paper,  we gave an efficient algorithm for computing a representative familiy of a family of independent sets in a linear matroid.  For the special case where the underlying matroid is uniform we  developed an even faster algorithm. We also showed interesting links between representative families of matroids and the design of single-exponential parameterized and exact exponential algorithms. We believe that these connections have a potential for a wide range of applications.  This works opens up an interesting avenue for further research, we list some of the natural open problems below.
\begin{itemize}
\item What is the best possible running time of an algorithm that computes a $q$-representative family of size at most ${p + q \choose p}$ for a $p$-family ${\cal F}$ of independent sets of a linear matroid? Does an algorithm with linear dependence of the running time on $|{\cal F}|$ exist, or is it possible to prove superlinear lower bounds?
\item It would be interesting to find faster algorithms even for special classes of linear matroids. Uniform matroids and graphic matroids are especially interesting in this regard.
\item Finally, the only matroids we used in our algorithmic applications were graphic, uniform, and partition matroids. It would be interesting to see what kind of applications can be handled by other kinds of matroids. 
\end{itemize}

\medskip

The results and methods from the preliminary conference version of this paper have already  been utilized to obtain several deterministic parameterized  
algorithms~\cite{FominG14,GoyalMP13,GoyalMPPS14,PinterSZ14,wgShachnaiZ14,ShachnaiZ14}. The results also have been used in the context exact learning~\cite{AbasiBM14} and linear time constructions of some $d$-restriction problems~\cite{Bshouty14}.  Lokshtanov et al.~\cite{LokshtanovMPS14} obtained a deterministic algorithm for computing a $\ell$-truncation of a given matrix and using this obtained a deterministic version of Theorem~\ref{thm:repsetlovaszrandomized} for those matroids whose representation can be found in deterministic polynomial time. Very recently  Zehavi~\cite{Zehavi14} has announced a further improvement for {\sc $k$-Path}  algorithm. The algorithm presented in~\cite{Zehavi14} runs in time $2.597^k \cdot n^{\cO(1)}$.  It has also been brought to out attention by Marek Cygan~\cite{megMarek}, in a private communication, that one can obtain single exponential time algorithms for {\sc Minimum Equivalent Graph} based on the methods described in~\cite{BodlaenderCK12,Cygan11}. 


\bibliographystyle{siam}
{\bibliography{ref}}

\begin{thebibliography}{10}

\bibitem{AbasiBM14}
{\sc H.~Abasi, N.~H. Bshouty, and H.~Mazzawi}, {\em On exact learning monotone
  {DNF} from membership queries}, in Algorithmic Learning Theory - 25th
  International Conference, {ALT} 2014, Bled, Slovenia, October 8-10, 2014.
  Proceedings, 2014, pp.~111--124.

\bibitem{AlonYZ}
{\sc N.~Alon, R.~Yuster, and U.~Zwick}, {\em Color-coding}, J. Assoc. Comput.
  Mach., 42 (1995), pp.~844--856.

\bibitem{AminiFS12}
{\sc O.~Amini, F.~V. Fomin, and S.~Saurabh}, {\em Counting subgraphs via
  homomorphisms}, SIAM J. Discrete Math., 26 (2012), pp.~695--717.

\bibitem{BangG089_book}
{\sc J.~Bang-Jensen and G.~Gutin}, {\em Digraphs}, Springer Monographs in
  Mathematics, Springer-Verlag London Ltd., London, second~ed., 2009.
\newblock Theory, algorithms and applications.

\bibitem{BjHuKK10}
{\sc A.~Bj{\"o}rklund, T.~Husfeldt, P.~Kaski, and M.~Koivisto}, {\em Narrow
  sieves for parameterized paths and packings}, CoRR, abs/1007.1161 (2010).

\bibitem{BjorklundHK04}
{\sc A.~Bj{\"o}rklund, T.~Husfeldt, and S.~Khanna}, {\em Approximating longest
  directed paths and cycles}, in Proceedings of the 31st International
  Colloquium, Automata, Languages and Programming (ICALP 2004), vol.~3142 of
  Lecture Notes in Comput. Sci., Springer, 2004, pp.~222--233.

\bibitem{Bodlaender93a}
{\sc H.~L. Bodlaender}, {\em On linear time minor tests with depth-first
  search}, J. Algorithms, 14 (1993), pp.~1--23.

\bibitem{BodlaenderCK12}
{\sc H.~L. Bodlaender, M.~Cygan, S.~Kratsch, and J.~Nederlof}, {\em Solving
  weighted and counting variants of connectivity problems parameterized by
  treewidth deterministically in single exponential time}, CoRR, abs/1211.1505
  (2012).

\bibitem{Bollobas65}
{\sc B.~Bollob{\'a}s}, {\em On generalized graphs}, Acta Math. Acad. Sci.
  Hungar, 16 (1965), pp.~447--452.

\bibitem{Bshouty14}
{\sc N.~H. Bshouty}, {\em Linear time constructions of some
  {\textdollar}d{\textdollar}-restriction problems}, CoRR, abs/1406.2108
  (2014).

\bibitem{bunch1974triangular}
{\sc J.~Bunch and J.~Hopcroft}, {\em Triangular factorization and inversion by
  fast matrix multiplication}, Mathematics of Computation, 28 (1974),
  pp.~231--236.

\bibitem{ChenKLMR09}
{\sc J.~Chen, J.~Kneis, S.~Lu, D.~M{\"o}lle, S.~Richter, P.~Rossmanith, S.-H.
  Sze, and F.~Zhang}, {\em Randomized divide-and-conquer: improved path,
  matching, and packing algorithms}, SIAM J. Comput., 38 (2009),
  pp.~2526--2547.

\bibitem{ChenLSZ07}
{\sc J.~Chen, S.~Lu, S.-H. Sze, and F.~Zhang}, {\em Improved algorithms for
  path, matching, and packing problems}, in Proceedings of the18th Annual
  ACM-SIAM Symposium on Discrete Algorithms (SODA 2007), SIAM, 2007,
  pp.~298--307.

\bibitem{CohenFG10}
{\sc N.~Cohen, F.~V. Fomin, G.~Gutin, E.~J. Kim, S.~Saurabh, and A.~Yeo}, {\em
  Algorithm for finding {$k$}-vertex out-trees and its application to
  {$k$}-internal out-branching problem}, J. Comput. System Sci., 76 (2010),
  pp.~650--662.

\bibitem{CormenLRS01}
{\sc T.~H. Cormen, C.~Leiserson, R.~Rivest, and C.~Stein}, {\em Introduction to
  Algorithms}, The MIT Press, Cambridge, Mass., second~ed., 2001.

\bibitem{megMarek}
{\sc M.~Cygan}, {\em private communication.},  (2013).

\bibitem{Cygan11}
{\sc M.~Cygan, J.~Nederlof, M.~Pilipczuk, M.~Pilipczuk, J.~M.~M. van Rooij, and
  J.~O. Wojtaszczyk}, {\em Solving connectivity problems parameterized by
  treewidth in single exponential time}, in Proceedings of the 52nd Annual
  Symposium on Foundations of Computer Science (FOCS 2011), IEEE, 2011.

\bibitem{DowneyF99}
{\sc R.~G. Downey and M.~R. Fellows}, {\em Parameterized complexity},
  Springer-Verlag, New York, 1999.

\bibitem{Edmonds67}
{\sc J.~Edmonds}, {\em Optimum branchings}, J. Res. Nat. Bur. Standards Sect.
  B, 71B (1967), pp.~233--240.

\bibitem{FominG14}
{\sc F.~V. Fomin and P.~A. Golovach}, {\em Long circuits and large euler
  subgraphs}, {SIAM} J. Discrete Math., 28 (2014), pp.~878--892.

\bibitem{Fomin:2011rr}
{\sc F.~V. Fomin and D.~Kratsch}, {\em Exact exponential algorithms}, Springer,
  2011.

\bibitem{FominLPS14}
{\sc F.~V. Fomin, D.~Lokshtanov, F.~Panolan, and S.~Saurabh}, {\em
  Representative sets of product families}, in Algorithms - {ESA} 2014 - 22th
  Annual European Symposium, Wroclaw, Poland, September 8-10, 2014.
  Proceedings, vol.~8737, 2014, pp.~443--454.

\bibitem{FominLRS12}
{\sc F.~V. Fomin, D.~Lokshtanov, V.~Raman, S.~Saurabh, and B.~V.~R. Rao}, {\em
  Faster algorithms for finding and counting subgraphs}, J. Comput. System
  Sci., 78 (2012), pp.~698--706.

\bibitem{FominLS14}
{\sc F.~V. Fomin, D.~Lokshtanov, and S.~Saurabh}, {\em Efficient computation of
  representative sets with applications in parameterized and exact algorithms},
  in Proceedings of the Twenty-Fifth Annual {ACM-SIAM} Symposium on Discrete
  Algorithms, {SODA} 2014, Portland, Oregon, USA, January 5-7, 2014, 2014,
  pp.~142--151.

\bibitem{Frankl82}
{\sc P.~Frankl}, {\em An extremal problem for two families of sets}, European
  J. Combin., 3 (1982), pp.~125--127.

\bibitem{GabowN08}
{\sc H.~N. Gabow and S.~Nie}, {\em Finding a long directed cycle}, ACM
  Transactions on Algorithms, 4 (2008).

\bibitem{GoyalMP13}
{\sc P.~Goyal, N.~Misra, and F.~Panolan}, {\em Faster deterministic algorithms
  for r-dimensional matching using representative sets}, in {IARCS} Annual
  Conference on Foundations of Software Technology and Theoretical Computer
  Science, {FSTTCS} 2013, December 12-14, 2013, Guwahati, India, vol.~24, 2013,
  pp.~237--248.

\bibitem{GoyalMPPS14}
{\sc P.~Goyal, P.~Misra, F.~Panolan, G.~Philip, and S.~Saurabh}, {\em Finding
  even subgraphs even faster}, in 35th {IARCS} Annual Conference on Foundation
  of Software Technology and Theoretical Computer Science, {FSTTCS} 2015,
  December 16-18, 2015, Bangalore, India, 2015, pp.~434--447.

\bibitem{Hsu75}
{\sc H.~T. Hsu}, {\em An algorithm for finding a minimal equivalent graph of a
  digraph}, J. Assoc. Comput. Mach., 22 (1975), pp.~11--16.

\bibitem{ImpagliazzoPZ01}
{\sc R.~Impagliazzo, R.~Paturi, and F.~Zane}, {\em Which problems have strongly
  exponential complexity}, Journal of Computer and System Sciences, 63 (2001),
  pp.~512--530.

\bibitem{jukna2011extremal}
{\sc S.~Jukna}, {\em Extremal combinatorics}, Springer Verlag Berlin
  Heidelberg, 2011.

\bibitem{Kloks94}
{\sc T.~Kloks}, {\em Treewidth, Computations and Approximations}, vol.~842 of
  Lecture Notes in Computer Science, Springer, 1994.

\bibitem{KneisMRR06}
{\sc J.~Kneis, D.~M{\"o}lle, S.~Richter, and P.~Rossmanith}, {\em
  Divide-and-color}, in Proceedings of the 34th International Workshop
  Graph-Theoretic Concepts in Computer Science (WG 2008), vol.~4271 of Lecture
  Notes in Computer Science, Springer, 2008, pp.~58--67.

\bibitem{Koutis08}
{\sc I.~Koutis}, {\em Faster algebraic algorithms for path and packing
  problems}, in Proceedings of the 35th International Colloquium on Automata,
  Languages and Programming (ICALP 2008), vol.~5125 of Lecture Notes in
  Computer Science, 2008, pp.~575--586.

\bibitem{KW09}
{\sc I.~Koutis and R.~Williams}, {\em Limits and applications of group algebras
  for parameterized problems}, in Proceedings of the 36th International
  Colloquium on Automata, Languages and Programming (ICALP 2009), vol.~5555 of
  Lecture Notes in Computer Sci., Springer, 2009, pp.~653--664.

\bibitem{KratschW12}
{\sc S.~Kratsch and M.~Wahlstr{\"o}m}, {\em Representative sets and irrelevant
  vertices: New tools for kernelization}, in Proceedings of the 53rd Annual
  Symposium on Foundations of Computer Science (FOCS 2012), IEEE, 2012,
  pp.~450--459.

\bibitem{LokshtanovMPS14}
{\sc D.~Lokshtanov, P.~Misra, F.~Panolan, and S.~Saurabh}, {\em Deterministic
  truncation of linear matroids}, in Automata, Languages, and Programming -
  42nd International Colloquium, {ICALP} 2015, Kyoto, Japan, July 6-10, 2015,
  Proceedings, Part {I}, 2015, pp.~922--934.

\bibitem{Lovasz77}
{\sc L.~Lov{\'a}sz}, {\em Flats in matroids and geometric graphs.}, in In
  Combinatorial surveys (Proc. Sixth British Combinatorial Conf., Royal
  Holloway Coll., Egham), Academic Press, London, 1977, pp.~45--86.

\bibitem{Martello78}
{\sc S.~Martello}, {\em An algorithm for finding a minimal equivalent graph of
  a strongly connected digraph}, Computing, 21 (1978/79), pp.~183--194.

\bibitem{MartelloT82}
{\sc S.~Martello and P.~Toth}, {\em Finding a minimum equivalent graph of a
  digraph}, Networks, 12 (1982), pp.~89--100.

\bibitem{Marx:2006ys}
{\sc D.~Marx}, {\em Parameterized coloring problems on chordal graphs}, Theor.
  Comput. Sci., 351 (2006), pp.~407--424.

\bibitem{Marx09}
\leavevmode\vrule height 2pt depth -1.6pt width 23pt, {\em A parameterized view
  on matroid optimization problems}, Theor. Comput. Sci., 410 (2009),
  pp.~4471--4479.

\bibitem{mitzenmacher2005probability}
{\sc M.~Mitzenmacher and E.~Upfal}, {\em Probability and computing: Randomized
  algorithms and probabilistic analysis}, Cambridge University Press, 2005.

\bibitem{Monien85}
{\sc B.~Monien}, {\em How to find long paths efficiently}, in Analysis and
  design of algorithms for combinatorial problems ({U}dine, 1982), vol.~109 of
  North-Holland Math. Stud., North-Holland, Amsterdam, 1985, pp.~239--254.

\bibitem{MoylesT69}
{\sc D.~M. Moyles and G.~L. Thompson}, {\em An algorithm for finding a minimum
  equivalent graph of a digraph}, J. ACM, 16 (1969), pp.~455--460.

\bibitem{murota2000matrices}
{\sc K.~Murota}, {\em Matrices and matroids for systems analysis}, vol.~20,
  Springer, 2000.

\bibitem{NaorSS95}
{\sc M.~Naor, L.~J. Schulman, and A.~Srinivasan}, {\em Splitters and
  near-optimal derandomization}, in Proceedings of the 36th Annual Symposium on
  Foundations of Computer Science (FOCS 1995), IEEE, 1995, pp.~182--191.

\bibitem{oxley2006matroid}
{\sc J.~G. Oxley}, {\em Matroid theory}, vol.~3, Oxford University Press, 2006.

\bibitem{PapadimitriouY96}
{\sc C.~H. Papadimitriou and M.~Yannakakis}, {\em On limited nondeterminism and
  the complexity of the {V-C} dimension}, J. Comput. Syst. Sci., 53 (1996),
  pp.~161--170.

\bibitem{PinterSZ14}
{\sc R.~Y. Pinter, H.~Shachnai, and M.~Zehavi}, {\em Deterministic
  parameterized algorithms for the graph motif problem}, in Mathematical
  Foundations of Computer Science 2014 - 39th International Symposium, {MFCS}
  2014, Budapest, Hungary, August 25-29, 2014. Proceedings, Part {II},
  vol.~8635, 2014, pp.~589--600.

\bibitem{Plehn:1991fk}
{\sc J.~Plehn and B.~Voigt}, {\em Finding minimally weighted subgraphs}, in
  Proceedings of the 16th Workshop on Graph-Theoretic Concepts in Computer
  Science (WG 1991), vol.~484 of Lecture Notes in Comput. Sci., Springer, 1991,
  pp.~18--29.

\bibitem{wgShachnaiZ14}
{\sc H.~Shachnai and M.~Zehavi}, {\em Parameterized algorithms for graph
  partitioning problems}, in Graph-Theoretic Concepts in Computer Science -
  40th International Workshop, {WG} 2014, Nouan-le-Fuzelier, France, June
  25-27, 2014. Revised Selected Papers, vol.~8747, 2014, pp.~384--395.

\bibitem{ShachnaiZ14}
\leavevmode\vrule height 2pt depth -1.6pt width 23pt, {\em Representative
  families: {A} unified tradeoff-based approach}, in Algorithms - {ESA} 2014 -
  22th Annual European Symposium, Wroclaw, Poland, September 8-10, 2014.
  Proceedings, vol.~8737, 2014, pp.~786--797.

\bibitem{Tuza94}
{\sc Z.~Tuza}, {\em Applications of the set-pair method in extremal hypergraph
  theory}, in Extremal problems for finite sets ({V}isegr\'ad, 1991), vol.~3 of
  Bolyai Soc. Math. Stud., J\'anos Bolyai Math. Soc., Budapest, 1994,
  pp.~479--514.

\bibitem{Tuza96}
\leavevmode\vrule height 2pt depth -1.6pt width 23pt, {\em Applications of the
  set-pair method in extremal problems. {II}}, in Combinatorics, {P}aul
  {E}rd{\H o}s is eighty, {V}ol.\ 2 ({K}eszthely, 1993), vol.~2 of Bolyai Soc.
  Math. Stud., J\'anos Bolyai Math. Soc., Budapest, 1996, pp.~459--490.

\bibitem{Williams09}
{\sc R.~Williams}, {\em Finding paths of length $k$ in ${O}^*(2^k)$ time}, Inf.
  Process. Lett., 109 (2009), pp.~315--318.

\bibitem{Williams12}
{\sc V.~V. Williams}, {\em Multiplying matrices faster than
  {C}oppersmith-{W}inograd}, in Proceedings of the 44th Symposium on Theory of
  Computing Conference (STOC 2012), ACM, 2012, pp.~887--898.

\bibitem{Zehavi14}
{\sc M.~Zehavi}, {\em Mixing color coding-related techniques}, in Algorithms -
  {ESA} 2015 - 23rd Annual European Symposium, Patras, Greece, September 14-16,
  2015, Proceedings, 2015, pp.~1037--1049.

\end{thebibliography}

\end{document}